\newcount\Comments  
\documentclass[11pt]{article}
\usepackage[margin=1in]{geometry}
\usepackage{amsmath}
\usepackage{amsthm}
\usepackage{amssymb}
\usepackage[ruled,vlined]{algorithm2e}
\usepackage{algorithmic}
\usepackage{bbm}
\usepackage{graphicx}
\usepackage{comment}
\usepackage{natbib}

\usepackage[usenames,dvipsnames]{color}
\usepackage[usenames,dvipsnames,svgnames,table]{xcolor}
\usepackage[colorinlistoftodos]{todonotes}
\usepackage[colorlinks=true, allcolors=NavyBlue]{hyperref}

\newtheorem{lemma}{Lemma}[section]
\newtheorem{definition}{Definition}[section]
\newtheorem{theorem}{Theorem}[section]
\newtheorem{corollary}{Corollary}[section]
\newtheorem{claim}{Claim}[section]

\newcommand{\E}[0]{\mathbb{E}}
\newcommand{\Var}[0]{\textrm{Var}}

\newcommand{\ignore}[1]{}
\newcommand{\kibitz}[2]{\ifnum\Comments=1\textcolor{#1}{#2}\fi}
\newcommand{\sherry}[1]{\kibitz{red}{\noindent[Sherry: #1]}}
\newcommand{\yc}[1]{\kibitz{magenta}{\noindent[YC: #1]}}

\renewcommand{\vec}[1]{\mathbf{#1}}

\begin{document}


\title{Prior-free Data Acquisition for Accurate Statistical Estimation}
\author{
Yiling Chen \\
Harvard University\\
yiling@seas.harvard.edu\\
\and 
Shuran Zheng \\
Harvard University\\
shuran\_zheng@seas.harvard.edu\\
}
\date{}
\maketitle


\begin{abstract}
We study a data analyst's problem of acquiring data from self-interested individuals to obtain an accurate estimation of some statistic of a population, subject to an expected budget constraint. Each data holder incurs a cost, which is unknown to the data analyst, to acquire and report his data. The cost can be arbitrarily correlated with the data. The data analyst has an expected budget that she can use to incentivize individuals to provide their data. The goal is to design a joint acquisition-estimation mechanism to optimize the performance of the produced estimator, without any prior information on the underlying distribution of cost and data. We investigate two types of estimations: unbiased point estimation and confidence interval estimation.
  \begin{description}
  \item[Unbiased estimators:] We design a truthful, individually rational, online mechanism to acquire data from individuals and output an unbiased estimator of the population mean when the data analyst has no prior information on the cost-data distribution and individuals arrive in a random order. The performance of this mechanism matches that of the optimal mechanism, which knows the true cost distribution, within a constant factor.  The performance of an estimator is evaluated by its variance under the worst-case cost-data correlation. 
  \item[Confidence intervals:] We characterize an approximately optimal (within a factor $2$) mechanism for obtaining a confidence interval of the population mean when the data analyst knows the true cost distribution at the beginning. This mechanism is efficiently computable. We then design a truthful, individually rational, online algorithm that is only worse than the approximately optimal mechanism by a constant factor. The performance of an estimator is evaluated by its expected length under the worst-case cost-data correlation. 
  \end{description}

\end{abstract}


\section{Introduction}
We study a data analyst's problem of estimating a population statistic (e.g. mean workout time in November) when data need to be acquired from self-interested data holders and the analyst has an expected budget constraint. Each data holder has a heterogeneous private cost to acquire and report his data (e.g. record duration of each workout in a month and report the total) and needs to be compensated at least by this cost to reveal his data. Individuals cannot fabricate their data if they decide to reveal it. Moreover, the values of the data and the private costs can be arbitrarily correlated in the population (e.g. those who work out regularly may use some fitness tracker which automatically records workout durations) and the correlation is unknown to the analyst a priori. A naive way for the analyst to acquire data in this setting is to offer a fixed compensation for each individual's data. But unless the payment level is higher than everyone's cost, in which case the analyst may run out of budget quickly and only be able to obtain a small sample, the collected sample will bias toward a low-cost subpopulation.  Thus, the problem is how to design a joint pricing-estimation mechanism to get accurate estimations when data holders are strategic.


The problem of purchasing data for unbiased estimation of population mean was first formulated by \citet{Roth2012} and then further studied by \citet{Chen2018}. Both works however assume that the cost distribution is known to the analyst and aim at obtaining an optimal unbiased estimator with minimum worst-case variance for population mean, where the worst-case is over all data-cost distributions consistent with the known cost distribution, subject to an expected budget constraint. The mechanism proposed by \citet{Roth2012} achieves optimality approximately when the cost distribution has  piece-wise differentiable PDF, while the mechanism proposed by \citet{Chen2018} achieves the exact optimality when the cost distribution is regular. \citet{Chen2018} also extend the result to linear regression. The high-level idea of both mechanisms is to acquire a data point with reported cost $c_i$ with a positive probability $A(c_i)$ (and some payment that is greater than or equal to $c_i$), then remove the sampling bias by re-weighting each collected data by $1/A(c_i)$, and finally average the re-weighted data to obtain an unbiased estimation (the Horvitz-Thompson Estimator). The assumption that the cost distribution is known allows the analyst to turn the mechanism design problem into a constrained optimization problem for finding an optimal allocation rule $A(c_i)$.  

Our paper makes two novel contributions, both of which do away from the main limiting assumptions of the prior works.  
First, we consider a data analyst with no prior information on the data holders' costs. We design an online mechanism that outputs an unbiased estimation of population mean, with the same goal as in the prior works: minimize the variance of the unbiased estimator subject to a budget constraint. Our only assumption is that data holders show up  in a uniformly random order. Here the challenge is that, in order to price well, the analyst needs to learn the cost distribution, but the pricing decisions need to be made for every arriving data holder. 
Our second contribution is to consider the bias-variance trade-off of the estimator. The previous works only consider unbiased estimators and the goal is to minimize the variance of the estimator. In this work, we allow the estimator to be biased and try to minimize the length of the confidence interval around the population mean, given a budget constraint. This necessarily requires us to reason about bias-variance trade-off together with data pricing, an aspect that, to the best of our knowledge, has not been explored in the literature. We design mechanisms for both the scenario where the analyst knows the cost distribution and the scenario where there is no prior information on costs.

\subsection{Summary of Our Results and Techniques}
Our work mainly addresses two questions:
\begin{enumerate}
\item If the data analyst does not have any prior information on data holders' private costs (as well as their private data), is it possible to design an online data acquisition mechanism for unbiased estimation of population mean that is competitive with the optimal mechanism that knows the cost distribution a priori? 
\item Can we design an optimal joint acquisition-estimation mechanism for estimating confidence intervals of population mean, when cost distribution is known? Optimality here means minimum length of the confidence interval. When cost distribution is unknown, can we design an online joint acquisition-estimation mechanism for confidence intervals that is competitive with the optimal mechanism that knows the cost distribution a priori? 
\end{enumerate}
For the first question, we design an online mechanism that is only worse than the optimal mechanism by a constant factor. The only substantial assumption we make in our setting is that the data holders come in random order, so if there are $n$ data holders in total, the cost-data distribution at each round is the discrete uniform distribution over the set of cost-data pairs of these $n$ data holders. Our mechanism satisfies the budget constraint in expectation, with the guarantee that the data holders will always be willing to participate and truthfully report their costs. 
\begin{theorem}[Informal] 
For the problem of purchasing data to get an unbiased estimator of population mean, assuming that the data holders come in random order, our online mechanism satisfies the following properties: (1) it is truthful and individually rational, (2) it satisfies the expected budget constraint, and (3) for any cost distribution, the variance of the produced unbiased estimator approaches that of the benchmark within a constant factor, where the benchmark is the optimal mechanism that knows the true cost distribution a priori.
\end{theorem}

For the second question, we extend our mechanism to output a confidence interval (using sample mean and sample variance). The mechanism may introduce some bias to mean estimation in exchange for a lower variance, so that the length of the confidence interval is approximately optimized. We provide the characterization of the approximately optimal confidence interval mechanism when the cost distribution is known. This characterization allows us to efficiently compute the mechanism. We then design an online mechanism that matches the performance of the optimal mechanism that knows the cost distribution within a constant factor.
\begin{theorem}[Informal]
For the problem of purchasing data to obtain a confidence interval, the approximately optimal mechanism that knows the cost distribution can be computed in polynomial time. 
\end{theorem}
This approximately optimal mechanism with known costs is constructed by analyzing the bias and variance trade-off for estimators for the mean. At any given bias level, by producing an estimator that has the lowest variance (for that bias level), we can construct a confidence interval using this biased mean estimation. We hence can design a mechanism to optimize for the length of the confidence interval. Since the optimal mechanism is difficult to compute, we approximate it to gain computational efficiency. 

\begin{theorem}[Informal]
For the problem of purchasing data to obtain a confidence interval, assuming that the data holders come in random order, our online mechanism has the following properties: (1) it is truthful and individually rational, (2) it satisfies the expected budget constraint, and (3) for any cost distribution, the performance of the produced confidence interval approaches that of the benchmark within a constant factor, where the benchmark is the optimal mechanism that knows the true cost distribution a priori.
\end{theorem}

Our online mechanisms for both unbiased mean estimation and confidence interval are designed using approximately optimal mechanisms with known costs as building blocks. At any round $i$, the reported costs in previous rounds gives us an empirical cost distribution. We then apply the optimal mechanism for this cost distribution for data holder $i$. Each round's mechanism has a fraction of the total expected budget. Our online mechanisms allocate more budget for early rounds in a way so that the performance of the final produced estimator is only worse than the benchmark by a constant factor.

\subsection{Other Related Work}
There is a growing interest in understanding statistical estimation and learning in environments with strategic agents. The works can be put in a few categories depending on the sources and types of strategic considerations.  

In this work, as well as \cite{Roth2012} and \cite{Chen2018}, agents do not derive utility or disutility from the estimation outcome, cannot fabricate their data, and have a cost for revealing their data. The mechanism uses payment to incentivize data revelation. \cite{Yiling15} is similar on these fronts, but the work considers general supervised learning. They do not seek to achieve a notion of optimality. Instead, they take a learning-theoretic approach and design mechanisms to obtain learning guarantees (risk bounds). \citet{CDP15} focused on incentivizing individuals to exert effort to obtain high-quality data for the purpose of linear regression. 

Another line of research examines data acquisition using differential privacy \citep{GR11,FL12, GLRS14, NVX14, CIL15}. Agents care about their privacy and hence may be reluctant to reveal their data.  The mechanism designer uses payments to balance the trade-off between privacy and accuracy. In this work, we implicitly assume that data holders to not have a privacy cost and hence they don't worry about potential leaking of their data by reporting their cost. In Section~\ref{sec:discussion}, we discuss the complication when data holders care about their privacy and their data and costs are correlated. 

A third line of research studies settings where data holders may strategically misreport their data, there is no ground truth to verify the acquired data, and the analyst would like to design payment mechanisms to incentivize truthful data reporting for the purpose of regression or other analyses~\citep{DBLP:journals/corr/abs-1802-09158,liu2016learning,liu2017sequential}. Because of the lack of verification, this line of work is closely related to the literature on peer prediction~\citep{miller2005eliciting,Shnayder:2016}.  

In a fourth line of research, individuals' utilities directly depend on the inference or learning outcome (e.g. they want a regression line to be as close to their own data point as possible) and they can manipulate their reported data to influence the outcome. In these works, there often is no cost for reporting one's data and the data analyst doesn't use monetary payments. These works attempt to design or identify mechanisms (inference or learning processes) that are robust to potential data manipulations~\citep{DFP10,MAMR11,MPR12, PP03,PP04,HMPW16,DRSWW17,Chen:2018:SLR:3219166.3219175}. 



\subsection{Organization of the Paper}

In the rest of the paper, we first formulate and characterize the optimal (or approximately optimal) mechanisms for the unbiased mean estimation and the confidence interval estimation when {\em the cost distribution is known} (Section \ref{sec:prelim} and Section \ref{sec:CI} respectively). 
These mechanisms serve both as building blocks for developing our online mechanisms when the cost distribution is unknown and as benchmarks to which our online mechanisms are compared. Section~\ref{unbiased} turns to the setting when the cost distribution is unknown. We develop online mechanisms with proven performance guarantees for both the unbiased mean estimation and the confidence interval estimation. We conclude with discussions and future directions in Section \ref{sec:discussion}.

\section{Model}
Consider a data analyst who conducts a survey to estimate some statistic of a population of $n$ people. In this work we focus on estimating the mean of some parameter of interest (e.g. alcohol consumption or BMI of an individual), denoted by $z$, and the confidence interval of the mean.  Each individual incurs a cost $c_i$, unknown to the data analyst, to acquire and report his data $z_i$. The cost and data pair can be correlated (e.g. those who consume more alcohol may have a higher cost recording their consumption), and follows an unknown distribution $\mathcal{D}$ supported on $(\mathcal{C}, \mathcal{Z})$. We assume that the cost is bounded by $\overline{C}$, i.e., $\mathcal{C} \subseteq [0, \overline{C}]$. The parameter $z$ is also bounded, and, without loss of generality, we assume $z$ is between $0$ and $1$, i.e., $\mathcal{Z} \subseteq [0,1]$. The data analyst has a budget $B = n \overline{B}$ that she can use to purchase data from the data holders.

We study an online setting where data holders arrive one by one to the survey, and no prior information on the distribution $\mathcal{D}$ (including the marginal distribution of the cost) is available before the survey. The analyst can gradually learn the distribution as data holders report their data. We make the following assumptions about the data sequence: (1) each individual only appears once, and (2) the data holders arrive in \emph{a random order}, i.e., each permutation of the $n$ people is equally likely. We use $(\vec{c}, \vec{z}) = (c_1, z_1), \dots, (c_n, z_n)$ to denote a random sequence of costs and data points, and $\{(c_1, z_1), \dots, (c_n, z_n)\}$ to represent a set of people's cost and data without the consideration of order.

When data holder $i$ arrives, the analyst asks the data holder to report his cost. We use $\widehat{c}_i$ to denote the reported cost of data holder $i$. Based on the reported cost, the analyst may offer a price to acquire the data $z_i$. Formally, the analyst uses a \emph{survey mechanism}, $M=(A,P)$, which consists of an allocation rule $A: \mathcal{C} \to [0,1]$ and a payment rule $P: \mathcal{C} \to \mathbb{R}$. With probability $A(\widehat{c}_i)$, the analyst offers payment $P(\widehat{c}_i)$ to purchase data $z_i$. If the data holder accepts this payment, he gives his data $z_i$ to the analyst. We assume that data holders do not misreport their data $z_i$. This assumption holds in situations when data can be verified once collected (e.g. medical records). The data holder walks away without revealing his data if $P(\widehat{c}_i) < c_i$. With probability $1-A(\widehat{c}_i)$, the analyst does not attempt to acquire the data.\footnote{We describe survey mechanisms as direct-revelation mechanisms, where date holders report their costs. Any survey mechanism can be implemented as a posted-price mechanism, where a menu of (price, probability) pairs are presented and each data holder chooses one from the menu~\cite{Roth2012,Chen2018}.}


The analyst can adaptively choose a survey mechanism for each arriving data holder. We use $\vec{M} = (\vec{A}, \vec{P}) = (A^1, P^1), \dots, (A^n, P^n)$ to represent a sequence of survey mechanisms. At round $i$, the analyst chooses an allocation rule $A^i$ and a payment rule $P^i$ based on all observed information before round $i$, denoted by $\mathcal{H}_{i-1}$. $\mathcal{H}_{i-1}$ includes the reported costs of the previous $i-1$ data holders and data points that have been acquired. The survey mechanism $(A^i, P^i)$ applies to the $i$-th arriving data holder. At the end of round $n$, the data analyst outputs an estimator $S(\vec{M}, (\vec{c}, \vec{z}))$ based on all observed information $\mathcal{H}_n$.


We want to design survey mechanisms that have the following incentive and budget properties:
\begin{description}
\item[Individual rationality:] The utility of each data holder is always non-negative, i.e., $P^i(\widehat{c}_i) \ge \widehat{c}_i$ for all $i$ and $\widehat{c}_i$.
\item[Truthfulness in expectation:]  A data holder maximizes his expected utility by reporting his cost truthfully, i.e., $A^i(c_i)(P^i(c_i) - c_i) \ge A^i(\widehat{c}_i)(P^i(\widehat{c}_i) - c_i)$ for all $i$ and $\widehat{c}_i \neq c_i$.  
\item[Expected budget feasibility:] $\E \left[\sum_{i=1}^n A^i(c_i)\cdot P^i(c_i)\right] \le B = n \cdot \overline{B}$, where the expectation is taken over the random arriving order of the data holders and the internal randomness of the mechanism.
\end{description}

In this work we mainly investigate two types of estimation tasks: (1) get an unbiased estimator of the population mean, with the goal that the variance of the estimator is minimized; (2) find a confidence interval of the population mean, with the goal that the length of the confidence interval is minimized. 
As an estimator uses data obtained via survey mechanisms $\vec{M}$, it necessarily depends on $\vec{M}$. 
We now formally define unbiased estimator and confidence interval of population mean in our setting. The randomness of an estimator $S(\vec{M}, (\vec{c}, \vec{z}))$ comes in two parts: (1) the external randomness, which is the random order of $(\vec{c},\vec{z})$, and (2) the internal randomness of the mechanisms $\vec{M}$.  Our definitions require the estimators to be unbiased or a valid confidence interval for any realization of the external randomness.


\begin{definition}[Unbiased estimator of population mean]
An estimator $S(\vec{M}, (\vec{c}, \vec{z}))$  is an unbiased estimator of the population mean $\E[z] = \frac{1}{n} \sum_{i=1}^n z_i$ if for any fixed sequence $(\tilde{\vec{c}},\tilde{\vec{z}})$, 
$$
\E[S(\vec{M}, (\tilde{\vec{c}},\tilde{\vec{z}})] = \E[z],
$$
where the expectation in $\E[S(\vec{M}, (\tilde{\vec{c}},\tilde{\vec{z}}))]$ is taken over the internal randomness of the mechanisms $\vec{M}$.
\end{definition}

\begin{definition}[Confidence interval of population mean]
An estimator $S(\vec{M}, (\vec{c},\vec{z}))$ is a confidence interval for the population mean $\E[z] = \frac{1}{n} \sum_{i=1}^n z_i$ with confidence level $\gamma$ if it is an interval and for any fixed sequence $(\tilde{\vec{c}},\tilde{\vec{z}})$,
$$
\Pr\left( \E[z] \in S(\vec{M}, (\tilde{\vec{c}},\tilde{\vec{z}}))\right) \ge \gamma,
$$
where the randomness is due to the internal randomness of the mechanisms $\vec{M}$.
\end{definition}

 Our goal is to design joint survey and estimation mechanisms, $(\vec{M}, S(\vec{M}, (\vec{c},\vec{z}))$, such that the estimator $S(\vec{M}, (\vec{c},\vec{z}))$ has good statistical performance on the population. For unbiased estimators, we prefer estimators with smaller variance. For confidence intervals, we prefer ones with smaller length. However, the performance of a mechanism on a population depends on the correlation between the population's cost and data, i.e. the distribution $\mathcal{D}$.\footnote{
For example, consider a mechanism that purchases each agent's data $z_i$ with a constant probability $p = B/\sum_{i=1}^n c_i$ and payment $c_i$, then outputs $1/(pn)$ times the sum of all purchased data as an unbiased estimation of population mean. When $z$ is always equal to $0$, the variance will be zero; when $z$ is always equal to $1$, the variance will be $(1/p-1)/n$.} We hence take a worst-case analysis approach: measure the performance of a mechanism under worst-case cost-data correlation.

\begin{definition}[Worst-case variance]
Given that the set of data holders' costs is $C = \{c_1, \dots, c_n\}$, the worst-case variance of a point estimator $S(\vec{M}, (\vec{c}, \vec{z}))$ is defined as
$$
\Var^*(S) = \max_{\mathcal{D} \textrm{ consistent with $C$}} \Var_{\mathcal{D}} (S(\vec{M}, (\vec{c}, \vec{z}))) 
$$
where the maximum is taken over all distributions $\mathcal{D}$ consistent with the set of costs $C$. The randomness is due to the random order of $(c_1, z_1), \dots, (c_n, z_n)$ and the internal randomness of the mechanism $\vec{M}$.
\end{definition}
\begin{definition}[Worst-case expected length]
Given that the set of data holders' costs is $C = \{c_1, \dots, c_n\}$, the worst-case expected length of a confidence interval $S(\vec{M}, (\vec{c}, \vec{z}))$ is defined as 
$$
L^*(S) = \max_{\mathcal{D} \textrm{ consistent with $C$}}  \E(|S(\vec{M}, (\vec{c}, \vec{z}))|) 
$$
where $|S(\vec{M}, (\vec{c}, \vec{z}))|$ represents the length of the confidence interval. The maximum is taken over all distributions $\mathcal{D}$ consistent with the set of costs $C$. The randomness is due to the random order of $(c_1, z_1), \dots, (c_n, z_n)$ and the internal randomness of the mechanisms $\vec{M}$.
\end{definition}

\citet{Roth2012} and \citet{Chen2018} have also considered the design of joint survey and estimation mechanism for statistical estimation. The main differences between their model and our model are: (1) they assume the marginal cost distribution is known to the data analyst, while our data analyst doesn't have such information, (2) they have the same survey mechanism for all individuals, while we consider an online setting where the analyst can adaptively change the survey mechanism, and (3) they only consider the estimation of mean, while we also investigate the estimation of confidence intervals. 
\yc{May consider to move this paragraph to the related work.}


\section{Preliminaries} \label{sec:prelim}
In this section, we first show that we can easily extend known results on one-shot truthful mechanisms to achieve truthfulness and individual rationality for a sequence of survey mechanisms $\vec{M}$. Then, we introduce the formulation proposed by \citet{Chen2018} for obtaining the optimal unbiased estimator of population mean {\em when the cost distribution is known to the analyst}. Later in Section~\ref{unbiased} we will use this known cost case as our benchmark for evaluating the performance of our optimal unbiased estimator when the cost distribution is unknown. 
While this optimal unbiased estimator has been studied by \citet{Chen2018}, their optimal mechanism requires the cost distribution to be regular. This is often not satisfied when we consider the costs of a finite set of data holders. We develop the characterization of the optimal unbiased estimator for arbitrary discrete cost distribution without any regularity assumption. We show that the optimal purchasing rule of cost-$c$ data is decided by a quantity which we define as \emph{regularized virtual costs} of the data. 


\subsection{Truthful and Individually Rational Survey Mechanisms} \label{sec:truthful}

Since each data holder appears only once in our setting, requiring a sequence of survey mechanisms to be truthful and individually rational is equivalent to requiring that each $(A^i, P^i)$ is truthful and individually rational, which can be achieved by a straight-forward extension of known results for truthful mechanisms.  

The well-known Myerson's lemma states that  monotonicity is the necessary and sufficient condition for an allocation rule to be truthful with some payment rule.
\begin{lemma} [\citet{myerson1983efficient}] \label{true_myerson} 
An allocation rule $A(c)$ is the allocation rule of some truthful survey mechanism $(A(c)$, $P(c))$ if and only if 
$A(c)$ is  monotone non-increasing in $c$.
\end{lemma}
%
%

The following lemmas from \cite{Chen2018} are analogies of the original Myerson's Lemma, which are tailored for discrete cost distributions. 
Firstly, they show that given a fixed monotone non-increasing allocation rule $A(c)$  defined on a discrete cost set $\{c_1, \dots, c_m\}$, there exists an optimal payment rule $P(c)$ that guarantees truthfulness and individual rationality. 
\begin{lemma}[\citet{Chen2018}, Claim 2 in Section B.1.2] \label{myerson}
Let $A(c)$ be a monotone non-increasing allocation rule defined for set $\{c_1, \dots, c_m\}$ with $c_1 \le \cdots \le c_m$. Define payment rule $P(c_i) = c_i + \frac{1}{A(c_i)} \sum_{j=i+1}^m A(c_j)(c_{j}-c_{j-1})$.  Then $(A(c), P(c))$ is truthful and individually rational for all $c \in \{c_1, \dots, c_m\}$, and any payment rule $P'(c)$ that guarantees the truthfulness and individual rationality of $(A(c), P'(c))$ must have $P'(c) \ge P(c)$ for all $c\in \{c_1, \dots, c_m\}$.
\end{lemma}
Furthermore for any cost distribution supported on $\{c_1, \dots, c_m\}$, the expected total payment of $(A(c),P(c))$, with the optimal payment rule $P(c)$ defined in Lemma~\ref{myerson}, can be equivalently represented in a simpler form in terms of virtual costs. 
\begin{definition}[Virtual costs] \label{def:virtual_costs}
Let $f(c)$ and $F(c)$ be the PDF and the CDF of a cost distribution $\mathcal{F}$ supported on $\{c_1, \dots, c_m\}$ with $c_1 \le \cdots \le c_m$. Let $c_0 = 0$. The virtual cost function $\psi(c)$ of this cost distribution is defined as
$$\psi(c_i) = c_i + \frac{c_i -c_{i-1}}{f(c_i)} F(c_{i-1})$$
for all $1\le t \le m$. 
\end{definition}

We remark that it is very likely that $\psi(c_i)$ is not regular, i.e., is not a monotone non-decreasing function of $c_i$. Consider when the cost distribution is the uniform distribution over a finite set, $\psi(c_i) = c_i + (c_i - c_{i-1}) * (i-1)$. If $c_i$ is very close to $c_{i-1}$, then $\psi(c_i)$ will be roughly equal to $c_i$. If $c_i$ is much larger than $c_{i-1}$, $\psi(c_i)$ can be close to $i*c_i$. So if there exist three consecutive costs, the difference between the first two costs is large, and the difference between the last two is small, the virtual costs will very likely be irregular. For example, $c_1 = 1$, $c_2 = 10$, $c_3 = 11$, then $\psi(c_2) = 19$ and $\psi(c_3) = 13$, so $\psi$ is not monotone non-decreasing.

\begin{lemma}[\citet{Chen2018}, Lemma 10 in Section B.1.2] \label{lem:expected_payment}
Let $A(c)$ be a monotone non-increasing allocation rule defined on set $\{c_1, \dots, c_m\}$ with $c_1 \le \cdots \le c_m$. Let $P(c)$ be the optimal truthful and individually rational payment rule defined in Lemma~\ref{myerson}.
When cost follows a distribution $\mathcal{F}$ supported on $\{c_1, \dots, c_m\}$, the expected total payment $\E_{c\sim \mathcal{F}}[A(c)P(c)]$ is equal to the expected virtual cost $\E_{c\sim \mathcal{F}}[A(c) \psi(c)]$ where $\psi(c)$ is the virtual cost function of $\mathcal{F}$.
\end{lemma}

The above lemmas assume that costs are from a finite discrete set. Our benchmark mechanism where the analyst already knows all the costs satisfies this assumption. We'll use the above result to establish the performance of our benchmark mechanism. However, our mechanisms developed in this paper for the unknown cost case do not have any restriction on the set of possible costs. The allocation rules and the payment rules of our mechanisms are first computed on a discrete set using the above result, and then extended to all other values of cost. We show below that such extension preserves truthfulness and individual rationality. 
%
\begin{definition}[Extended allocation rule and payment rule] \label{def:extend_mech}
Given a survey mechanism $(A^d,P^d)$ that is defined on a discrete cost set $\{c_1, \dots, c_m\}$ with $c_1 \le \cdots \le c_m$. The extended allocation rule and payment rule $A,P$ are defined as follows
$$
A(c) = A^d(\lceil c \rceil), \quad P(c) = P^d(\lceil c \rceil), \textrm{ for all } c\in [0, c_m],
$$
where $\lceil c \rceil$ is the minimum cost in $\{c_1, \dots, c_m\}$ that is greater than or equal to $c$. 
\end{definition}
\begin{lemma} \label{lem:extended_mech}
Let $A^d(c)$ be a monotone non-increasing allocation rule defined on set $\{c_1, \dots, c_m\}$ with $c_1 \le \cdots \le c_m$. Let $P^d(c)$ be the optimal payment rule defined in Lemma~\ref{myerson}.
Then the extended allocation rule and payment rule of $(A^d, P^d)$ is still truthful and individually rational.
\end{lemma}
The proof of the lemma can be found  in~\ref{app:extended_alloc}.

\subsection{Formulating the Optimal Unbiased Estimator with Known Costs} \label{sec:formulating_opt_unbiased}

In this section, we formulate an optimization problem that solves the optimal survey mechanism for an unbiased estimator {\em when the cost distribution is known}, based on the results of~\cite{Chen2018}. The optimization problem is only slightly different from~\cite{Chen2018} in the objective function because in our setting, the agents are not i.i.d. drawn from the same distribution but come as a random permutation. The value of statistic $z$ is assumed  to be bounded and without loss of generality $0\le z \le 1$. 

\paragraph{Horvitz-Thompson Estimator:} When we use truthful survey mechanisms $\vec{M} = (A^1, P^1), \dots,$ $ (A^n, P^n)$ to purchase the data points, the data of agent $i$ will be collected with probability $A^i(c_i)$. Define 
$$
\widehat{x}_i = \left\{ \begin{array}{ll}
						z_i, \textrm{ with probability } A^i(c_i)\\
                        0, \textrm{ otherwise}
						\end{array}
                        \right.
$$
to be the observed data point which is set to zero if no purchase is made. Define $y_i = \frac{\widehat{x}_i}{A^i(c_i)}$. 
To get unbiased estimation, we use Horvitz-Thompson Estimator, which is the unique unbiased linear estimator in our setting (see~\citep{Roth2012} for more details),  
$$
S(\vec{M}, (\vec{c}, \vec{z})) = \frac{1}{n} \sum_{i=1}^n y_i.
$$
Notice that an unbiased estimator always buys the data points with probability greater than $0$, i.e., $A^i(c_i) > 0$ for all $i$ and $c_i$. If $A^i(c_i) = 0$, the mechanism will never know the value of $z_i$ and thus cannot be unbiased. 

When the cost set  $C = \{c_1, \dots, c_n\}$ is known to the analyst, the optimal mechanism that uses the same survey mechanism for all data holders has been derived by~\citet{Chen2018}. They reduce the mechanism design problem to a min-max optimization problem. The optimal allocation rule that minimizes the worst-case variance of the Horvitz-Thompson Estimator can be formulated as the solution of an optimization problem $OPT(n, C, B)$, which is defined as follows,
\begin{align} \label{prog:unbiased_benchmark_def}
OPT(n, C, B) = \arg \min_A  \max_{\vec{z}\in [0,1]^n} & \quad \frac{1}{n^2} \left(\sum_{i=1}^n  \frac{z_i^2}{A(c_i)} - \sum_{i=1}^n z_i^2  \right)  \\
\textrm{ s.t. }& \quad \sum_{i=1}^n A(c_i) \psi(c_i) \le B \notag\\
 & \quad A(c) \textrm{ is monotone non-increasing in } c \notag\\
 & \quad 0\le A(c) \le 1, \quad \forall c \notag 
\end{align}

Here the objective function is changed from the original formulation in~\cite{Chen2018} so that it is equal to the worst-case variance of the Horvitz-Thompson Estimator in our setting. According to the law of total variance,  
$\Var(S(\vec{M}, (\vec{c}, \vec{z}))) = \E[\Var(S|(\vec{c}, \vec{z}))] + \Var( \E[S|(\vec{c}, \vec{z})])$.
Since  the estimator is always unbiased for any order $(\vec{c}, \vec{z})$, the second term is zero. Furthermore, when conditioning on a sequence,  $y_i = \frac{\widehat{x}_i}{A(c_i)}$ become independent when a fixed allocation rule $A$ is used.
 Therefore the variance of the Horvitz-Thompson Estimator is equal to 
\begin{align*}
 \Var(S(\vec{M}, (\vec{c}, \vec{z}))) = & \E_{(\vec{c}, \vec{z})}[\Var(S|(\vec{c}, \vec{z}))] =  \E_{(\vec{c}, \vec{z})}\left[ \frac{1}{n^2} \sum_{i=1}^n \Var\left(y_i | (\vec{c}, \vec{z})\right) \right]\\
 = &\frac{1}{n^2}  \cdot \E_{(\vec{c}, \vec{z})} \left[  \sum_{i=1}^n  \E[y_i^2 | (\vec{c}, \vec{z})] -\E[y_i | (\vec{c}, \vec{z})]^2 \right]. 
\end{align*}
For any arriving sequence $(\vec{c}, \vec{z})$, $\sum_{i=1}^n  \E[y_i^2 | (\vec{c}, \vec{z})] -\E[y_i | (\vec{c}, \vec{z})]^2$ stays the same, which is equal to 
$\sum_{i=1}^n  \frac{z_i^2}{A(c_i)} - \sum_{i=1}^n z_i^2$.
Therefore by maximizing over $\vec{z}$, we get the worst-case variance of the Horvitz-Thompson estimator
$$
 \max_{\vec{z} \in [0,1]^n}  \frac{1}{n^2} \left(\sum_{i=1}^n  \frac{z_i^2}{A(c_i)} - \sum_{i=1}^n z_i^2  \right).
$$

The last constraint $0\le A(c) \le 1$ makes sure that $A$ is an allocation rule of a survey mechanism. The second constraint is the sufficient and necessary condition for $A$ to be the allocation rule of a truthful mechanism. The first constraint guarantees expected budget feasibility according to Lemma~\ref{lem:expected_payment}, where $\psi(c)$ is the virtual cost function. 

To simplify the analysis, in this work we will use an approximation of the above optimization problem, which is defined as $APPROX(n, C, B)$ by removing the second term of the objective function of $OPT(n, C,B)$,
\begin{align} \label{prog:unbiased_benchmark_approx}
APPROX(n, C, B) = \arg \min_A  \max_{\vec{z}\in [0,1]^n} & \quad \sum_{i=1}^n  \frac{z_i^2}{A(c_i)} \\
\textrm{ s.t. }& \quad \sum_{i=1}^n A(c_i) \psi(c_i) \le B \notag\\
& \quad A(c) \textrm{ is monotone non-increasing in } c \notag\\
 & \quad 0\le A(c) \le 1, \quad \forall c \notag 
\end{align}
\begin{lemma}
The worst-case variance of  $APPROX(n, C,B)$ is no more than the worst-case variance of $OPT(n,C,B)$ plus $\frac{1}{n}$.
\end{lemma}
\begin{proof}
As $z_i \in [0,1]$, the second term of the objective function of $OPT(n, C,B)$ $\le 1/n$. 
\end{proof}
\subsection{Characterization of the Optimal Unbiased Estimator}
The characterization of the optimal unbiased estimator when the cost distribution is known has been studied in both~\cite{Chen2018} and~\cite{Roth2012}. But neither of the solutions can be directly applied to our problem. \citet{Roth2012} require the PDF of the cost distribution to be piecewise differentiable except on a measure zero set, and~\citet{Chen2018} assume that the cost distribution is regular, i.e., the virtual cost function $\psi(c)$ is monotone non-decreasing in $c$.   
Although the optimization problem~\eqref{prog:unbiased_benchmark_approx} has a convex objective function and thus can be solved efficiently by the convex optimization algorithms (see~\cite{Boyd:2004:CO:993483}), the closed-form solution is still non-trivial to derive.  Below we give the exact characterization of the optimal solution, which has a very simple form and will further be used to derive the optimal confidence interval mechanism: the optimal allocation rule $A^i(c)$ is inversely proportional to the square root of the regularized virtual cost of $c$, which is defined as follows,
\begin{definition}[Regularized virtual costs] \label{def:regularized_virtual_costs}
For a  discrete uniform distribution  supported on $\{c_1, \dots, c_m\}$ with $c_1 \le \cdots \le c_m$ and its virtual costs function $\psi(c_1), \dots, \psi(c_m)$. For every $i\le k$, let $Avg(i,k)$ be the average of $\psi(c_i), \dots, \psi(c_k)$. We define regularized virtual cost $\phi(c_i)$ as follows
\begin{align*}
&\phi(c_i) = \max \{ \psi'(c_1), \dots, \psi'(c_i)\},\\
&\psi'(c_i) = \min_{k: k\ge i} Avg(i,k).
\end{align*}
\end{definition}
The form of the regularized virtual costs is very similar to the ironed virtual value  used in  revenue-maximizing auction design~\cite{myerson1981optimal}.  The idea is to replace the exact virtual cost (value) with the average virtual cost (value) on an interval that has non-regular virtual cost (value). But our proof is different because the underlying optimization problem is not the same. 
\begin{theorem} \label{lem:unbiased_chrz}
 The optimal solution of $APPROX(n,C,B)$ is 
$$A(c_{(k)}) = \min\left\{ 1, \frac{\lambda}{\sqrt{\phi(c_{(k)})}} \right\}, \quad \textrm{ for all } 1\le k \le n,$$
 where $\phi(c)$ is the regularized virtual cost function when the cost distribution is the uniform distribution over $C$, and $\lambda$ is chosen such that the budget constraint is satisfied with equality $$ \sum_{k=1}^n A(c_{(k)}) \psi(c_{(k)}) = B.$$
The value of $\lambda$ can be computed using binary search within time $O(\log |C|)$.
\end{theorem}
The proof of the theorem can be found  in~\ref{app:unbiased_chrz}, in which we demonstrate some properties of the regularized virtual cost function, and show that the KKT conditions are satisfied because of these properties.


\section{Optimal Confidence Interval with Known Costs}\label{sec:CI}
In this section, we design purchasing mechanisms to get the best confidence interval of the statistic when the cost distribution is known to the analyst at the beginning of the survey.  In this case, the optimal mechanism needs to find the optimal trade-off of the bias and the variance in order to minimize the length of the interval. We consider the class of confidence intervals that are defined around the sample mean, and the length of which is decided by a bias term and the sample variance.  We will first introduce an extended survey mechanism that allows biased estimation, e.g., by ignoring the high-cost data. Then again we formulate an optimization problem and present the characterization of the optimal solution.

Before introducing the new estimator, we first show how to convert our unbiased estimator into a confidence interval. In this work, we use the most classic approach to construct confidence interval based on sample mean and sample variance. 
\paragraph{Construct confidence interval using unbiased estimator:} Consider an unbiased estimator $S(\vec{M}, (\vec{c}, \vec{z}))$ that uses survey mechanism $\vec{M} = (\vec{A},\vec{P})$, and we want to construct a confidence interval for $\E[z] = \frac{1}{n} \sum_{i=1}^n z_i$. Again we use 
$$
\widehat{x}_i = \left\{ \begin{array}{ll}
						z_i, \textrm{ with probability } A^i(c_i)\\
                        0, \textrm{ otherwise}
						\end{array}
                        \right.
$$
to denote the observed data point and define $y_i = \frac{\widehat{x}_i}{A^i(c_i)}$. Notice that the random variables $y_1, \dots, y_n$ are not independent since the allocation rule $A^i$ can depend on $c_1, \dots, c_{i-1}$. But if we consider a fixed realization $(\tilde{\vec{c}}, \tilde{\vec{z}})$, the mechanisms $(A^1, P^1), \dots, (A^n, P^n)$ will also be fixed.  Then $y_1, \dots, y_n$ become independent, because when the probability of purchasing each data point $A^1(c_1), \dots, A^n(c_n)$ is fixed, whether to purchase each data point or not is independently decided. Therefore given a confidence level $\gamma$, we can construct a confidence interval of the expected mean $\E\left[\sum y_i/n \big |  (\tilde{\vec{c}}, \tilde{\vec{z}})\right]$
using the sample mean $\sum_{i=1}^n y_i / n$ and sample variance $\widehat{\sigma}^2 = \sum_{i=1}^n \left(y_i - \sum_{i=1}^n y_i/n\right)^2/(n-1)$ according to Bernstein's inequality (see more details in~\ref{app:concentration}):
$$
\left[\sum_{i=1}^n y_i / n - \frac{\alpha_\gamma}{\sqrt{n}} \cdot \widehat{\sigma}, \quad \sum_{i=1}^n y_i/ n + \frac{\alpha_\gamma}{\sqrt{n}} \cdot \widehat{\sigma}\right]
$$
where $\alpha_\gamma$ is a constant that is chosen to achieve confidence level $\gamma$.  When the estimator is unbiased, $\E\left[\sum y_i/n \big |  (\tilde{\vec{c}}, \tilde{\vec{z}})\right] = \E[z]$ for all $ (\tilde{\vec{c}}, \tilde{\vec{z}})$, this interval is just a confidence interval of $\E[z]$ with confidence level $\gamma$.

\subsection{Confidence Interval and Bias-variance Tradeoff}

 This ``unbiased'' confidence interval does not necessarily have the minimum length. Observe that a small portion of high-cost data can drastically increase the variance of the unbiased estimator as $A(c)$ must be small. We can allow the mechanism to simply ignore these data points, i.e., to have $A^i(c_i)=0$ for some $i$ and set $y_i = 0$. This can probably reduce the variance of the estimator. However, doing so causes the estimator to be biased. Therefore, we need to increase the length of the confidence interval to compensate for this added bias.  

For this reason, we extend the standard survey mechanism to incorporate biased estimation.
The mechanism ignores a data point with probability $U(c)$. The added bias can be represented as a function of $U(c)$. Although the mechanisms now have a new component $U(c)$, the results in Section~\ref{sec:prelim} can nevertheless be applied by seeing $(1-U(c)) A(c)$ as the allocation rule.

\paragraph{Survey mechanisms that allow bias:} We add a new component $\vec{U} = (U^1, \dots, U^n)$ into our allocation rule to allow biased estimation, where each $U^i$ is a function of reported cost $c_i$. A  mechanism that allows bias consists of $(\vec{A}, \vec{U}, \vec{P})$. When a data point with cost $c_i$ comes at time $i$, the mechanism first flips a coin $\widehat{U}_i$ to decide whether to ignore this data point or not, and the probability of $\widehat{U}_i$ being $1$ (which means ignoring the data) is equal to $U^i(c_i)$. If the data is ignored, a bias term will be added into the final estimation to compensate the error.  If $\widehat{U}_i = 0$, then the mechanism purchases the data with probability $A^i(c_i)>0$ and pays $P^i(c_i)$ if the data is purchased. Then the observed data $\widehat{x}_i$ follows
$$
\widehat{x}_i = \left\{ \begin{array}{ll}
						z_i, \textrm{ with probability } (1- U^i(c_i)) A^i(c_i)\\
                        0, \textrm{  with probability } (1- U^i(c_i)) (1 - A^i(c_i))\\
                        \textrm{ignored, with probability } U^i(c_i).
						\end{array}
                        \right.
$$
We re-define $y_i$ as
$$
y_i = \left\{ \begin{array}{ll}
						\frac{\widehat{x}_i}{A^i(c_i)}, \textrm{ if  } \widehat{U}_i =0 \\
                        0, \textrm{ if }\widehat{U}_i =1.
						\end{array}
                        \right.
$$
Then for a fixed arriving sequence $(\tilde{\vec{c}}, \tilde{\vec{z}})$, the bias of estimator $\sum y_i/n$ is equal to 
$$
err = \E[z] - \frac{1}{n} \sum_{i=1}^n \E[y_i | (\tilde{\vec{c}}, \tilde{\vec{z}})]  = \frac{1}{n} \sum_{i=1}^n z_i - \frac{1}{n} \sum_{i=1}^n (1- \widehat{U}_i) \tilde{z}_i  = \frac{1}{n} \sum_{i=1}^n \tilde{z}_i \cdot \widehat{U}_i.
$$
Notice that this bias is not observable because the mechanism does not know the $\tilde{z}_i$ that is not purchased. But since $\tilde{z}_i$ is between $0$ and $1$ and we use worst-case analysis  in this work, we can just assume $\tilde{z}_i$ equals its worst-case value $1$. (This can be seen more clearly in our formulation of the optimization problem  in the next section.) Then the confidence interval of $\E[z]$ with confidence level $\gamma$ can be constructed as follows
$$
\left[\sum_{i=1}^n y_i / n - \frac{\alpha_\gamma}{\sqrt{n}} \cdot \widehat{\sigma}, \quad \sum_{i=1}^n y_i/ n + \frac{1}{n} \sum_{i=1}^n \widehat{U}_i + \frac{\alpha_\gamma}{\sqrt{n}} \cdot \widehat{\sigma}\right]
$$
where  $\widehat{U}_i$ is the indicator of whether the $i$-th data point is ignored and  $\widehat{\sigma}^2$ is the sample variance of $y_1, \dots, y_n$.

For convenience, in the rest of the paper, we write $U^i_c = U^i(c)$ for short, and in cases when the costs are indexed as $c_1, \dots, c_n$ or $c_{(1)}, \dots, c_{(n)}$, we use $U^i_j$ to represent $U^i(c_{j})$ or $U^i(c_{(j)})$.

\subsection{Formulation of the Optimal Confidence Interval} \label{sec:biased_benchmark}
 We formulate a min-max optimization problem that approximately solves the optimal allocation rule. 
The expected length of the interval we construct is 
$2 \cdot  \frac{\alpha_\gamma}{\sqrt{n}} \cdot \E[\widehat{\sigma} ]+ \E[err]$.
Since the expectation of sample standard deviation $\E[\widehat{\sigma} ]$ is difficult to compute, we estimate $\E[\widehat{\sigma} ]$ with $\sqrt{\E \sum_{i=1}^n y_i^2/n}$. 
When $0 \le \E[y_i] \le 1$, the difference between $\E[\widehat{\sigma} ]$ and $\sqrt{\E \sum_{i=1}^n y_i^2/n}$ is no more than $1 + O(1/n)$ (see~\ref{app:difference} for more details). 

The approximate expected length of the confidence interval can thus be written into a function of $A$ and $U$ and $\vec{z}$
$$
2 \cdot  \frac{\alpha_\gamma}{\sqrt{n}} \cdot \sqrt{\E \sum_{i=1}^n y_i^2/n} + \E[err] = 2 \cdot  \frac{\alpha_\gamma}{\sqrt{n}} \cdot \sqrt{\frac{1}{n} \sum_{i=1}^n \frac{(1-U_i) z_i^2}{A(c_i)}} + \frac{1}{n} \sum_{i=1}^n z_i \cdot U_i.
$$
Then we only need to take maximum over all possible $\vec{z}$ to get the worst-case expected length. Therefore, suppose the underlying cost set  is $C = \{c_{(1)}, \dots, c_{(n)}\}$ with $c_{(1)} \le  \cdots \le c_{(n)}$, then the approximately optimal allocation rule $(A^*, U^*)$ can be again formulated as the solution of $OPT_{CI}(n, C, B)$, which is defined as
\begin{align} \label{prog:biased_benchmark}
 OPT_{CI}(n, C,B) = \arg\min_{A, U} \max_{\vec{z}\in[0,1]^n} \quad & \beta \cdot \sqrt{\frac{1}{n} \cdot \sum_{i=1}^{n} \frac{(1-U_i) z_i^2}{A(c_{(i)})}} + \frac{\sum_{i=1}^{n} z_i \cdot U_i}{n}  \\
\textrm{s.t.} \quad & \sum_{i=1}^{n} (1-U_i) \cdot A(c_{(i)}) \psi(c_{(i)})  \le B \notag\\
& (1-U_c)A(c) \textrm{ is monotone non-increasing in } c \notag\\
& 0\le A(c) \le 1, \quad  0 \le U_c \le 1, \quad \forall c \notag
\end{align}
where $\beta =  2 \cdot  \frac{\alpha_\gamma}{\sqrt{n}}$. 
\begin{lemma} \label{lem:biased_benchmark}
Let $L^*$ be the value of the objective function~\eqref{prog:biased_benchmark} when $A^*$ and $U^*$ is used. ( $L^*$ is an approximation of the worst-case expected length of the confidence interval produced by $(A^*, U^*)$. ) Then the difference between $L^*$ and the worst-case expected length of the actual optimal confidence interval $MIN$ is no more than $2 \beta (1 + O(1/n)) = 4 \cdot  \frac{\alpha_\gamma}{\sqrt{n}} (1 + O(1/n))$.
\end{lemma}
The optimal solution of~\eqref{prog:biased_benchmark} is still difficult to solve. But if we replace the objective function by the sum of the squares of its two terms 
\begin{align} \label{prog:biased_approx}
APPROX_{CI}(n, C, B) = \arg \min_{A, U} \max_{\vec{z}\in[0,1]^n} \quad & \beta^2 \cdot \frac{1}{n} \cdot \sum_{i=1}^{n} \frac{(1-U_i) z_i^2}{A(c_{(i)})} + \left(\frac{\sum_{i=1}^{n} z_i \cdot U_i}{n}\right)^2
\end{align}
the optimal solution can be computed efficiently.  The optimization problem with the new objective function will give a $2$-approximation of~\eqref{prog:biased_benchmark} because $a^2 + b^2 \le (a+b)^2 \le 2(a^2 + b^2)$ (see more details in the last paragraph of~\ref{app:CI_main}).  
\begin{lemma}
$APPROX_{CI}( n, C, B)$ will give a $2$-approximation of $OPT_{CI}(n, C,B)$. 
\end{lemma}

In this work, we consider the bias-variance tradeoff in the setting of getting the shortest confidence interval. By changing the value of $\beta$ in the objective function, our mechanism can also be applied to other estimation tasks that involve the bias-variance tradeoff. For example, if we want to minimize the mean squared error of the output estimator, or equivalently the risk of the estimator when the loss function is the squared error function, we can set $\beta = 1$ and apply the same mechanism. 

\subsection{Characterization of the Optimal Confidence Interval}
So how should the optimal mechanism look like when some data can simply be ignored? It is natural to believe that the optimal mechanism should ignore the data points with the highest costs. This is corroborated in the following theorem that the optimal mechanism ignores all the data points with \emph{regularized virtual costs} above a threshold $H$, and purchases (with probability) all the data points below the threshold. The characterization of the optimal $A$ remains the same as the unbiased case.
\begin{theorem} \label{thm:CI_roundi}
The optimal solution of $APPROX_{CI}( n, C, B)$ is as follows:
\begin{align*}
&U_j = \left\{ \begin{array}{ll}
				 0, & \textrm{ if $\phi(c_{(j)}) < H$} \\
				 p \in (0,1], & \textrm{ if } \phi(c_{(j)}) = H \\
				 1, & \textrm{ if $\phi(c_{(j)})>H$}
			     \end{array}
			\right.\\
&A(c_{(j)}) =  \min \left\{1, \frac{\lambda}{\sqrt{\phi(c_{(j)})}} \right\} 
\end{align*}
where $p$ is a constant in $(0,1]$, and $\phi(c)$ is the regularized virtual cost function (Definition~\ref{def:regularized_virtual_costs}) when the cost distribution is the uniform distribution over $C$, and $\lambda$ is chosen such that the budget constraint is satisfied with equality.
The value of $\lambda$ and $H$ can be computed using binary search over set $C$ within time $O(\log |C|)$.  
\end{theorem}
The theorem is proved in ~\ref{app:CI_roundi}.
 The optimal $H$ can be found by binary search because it can be proved that  the objective function is a convex function of $M = \sum_{j=1}^n U_j$, when $A$ is optimized after $U$ is decided.  Let $g(M)$ be the the optimal value of the first term $\beta^2 \cdot \frac{1}{n}  \sum_{j = 1}^n \frac{1-U_j}{A(c_{(j)})}$ when $\sum_{j=1}^n U_j$ is set to $M$. Then the objective function is just $g(M) + \left(\frac{M}{n}  \right)^2$.  The second term $\left(\frac{M}{n}  \right)^2$ is a convex function of $M$. We prove that $g(M)$ is also a convex function of $M$.
\begin{lemma} \label{lem:baised_eq_point}
The function $g(M)$ is a convex function of $M$. Furthermore, let $A_M$ be an optimal allocation rule  when $\sum_{j=1}^n U_j = M$ and let $c_{(r)}$ be the largest cost that is not ignored with probability $1$.  Then for non-integer $M$ that has $A_M(c_{(r)}) < 1$,
$$
\frac{\partial g(M)}{\partial M} =  - \beta^2 \cdot \frac{2}{n} \cdot  \frac{1}{A_M(c_{(r)})}.
$$
which is an non-decreasing function of $M$.
\end{lemma}

Therefore the optimal $H$ can be found by binary search over the optimal $M$ such that $\frac{\partial g(M)}{\partial M} + \frac{2M}{n^2} \ge 0$ on the right and $\frac{\partial g(M)}{\partial M} + \frac{2M}{n^2} \le 0$ on the left. The complete proof of the lemma is in~\ref{app:biased_eq_point}.

\section{Online Mechanisms}\label{unbiased}

We now move to the case when the cost distribution is unknown at the beginning. 
The idea of our mechanism is very simple: at time $i$, use the approximately optimal allocation rule $APPROX(\cdot)$ as if (1) there are $i$ data holders with costs $c_1, \dots, c_{i-1}$ and $\{\overline{C}\}$, and (2) the analyst's total budget for these $i$ data points is proportional to $\sqrt{i}$, so that the average budget for each data point  is proportional to $\frac{1}{\sqrt{i}}$. So the average budget is a decreasing function of $i$, which means we use more budget in the earlier stages when the estimation of the cost distribution is not accurate.  
 We prove that for both of the unbiased estimator problem and the confidence interval problem, such an online mechanism will only be worse for a constant factor compared to the optimal mechanisms $OPT(n, C, B)$ that knows the cost distribution at the beginning.

\begin{algorithm}[!h]               
\SetAlgorithmName{Mechanism}{mechanism}{List of Mechanisms}
\begin{algorithmic}
	\FOR {i = 1,\dots,n}
    \STATE 1) Let $A^i = APPROX\left(i, \ T_i,\ \xi_n B \sqrt{i}\right)$ be the optimal allocation rule  when there are $i$ data holders with costs $$T_i=\{c_1, \dots, c_{i-1},\overline{C}\},$$  and the total budget for these $i$ data points is $ \xi_n B \sqrt{i}$.
 \STATE 2) Ask agent $i$ to report cost $c_i$ and purchase the data using (an approximation of) $(A^i, P^i)$, where $P^i$ is computed as in Lemma~\ref{myerson}.
    \ENDFOR
    \STATE Aggregate all collected to output an estimator. 
\end{algorithmic}
\caption{Online Mechanism Outline}
\label{alg:online_unbiased}
\end{algorithm}
The outline of our online mechanisms (for both the unbiased estimator and the confidence interval) is given in Mechanism~\ref{alg:online_unbiased}.

We will describe the specifics of the online mechanism for the unbiased estimation of mean and the confidence interval respectively in the next two parts. For each, we'll prove that the online mechanism is optimal within a constant factor. We sketch the high-level idea of our proofs as follows. 
\begin{enumerate}
\item We compare both of our online algorithm and the benchmark with an intermediate mechanism 
$$
(A'^{,i},P'^{,i}) = APPROX\left(i, \ \{ c_1, \dots, c_{i-1}, c_i \}, \ \xi_n B \sqrt{i}\right)
$$ at each step $i$. This intermediate mechanism $(A'^{,i},P'^{,i})$ is basically the same as $(A^i, P^i)$, but is ``one-step-ahead''. $(A', P')$ is the optimal mechanism when the same amount of budget is assigned, but knows an additional piece of information, the value of $c_i$, beforehand. 
\item We show that the difference between $(A'^{,i},P'^{,i})$ and $(A^i, P^i)$ is no more than a constant factor. This is mainly due to the similarity of $(A'^{,i},P'^{,i})$and $(A^i, P^i)$. Since the two mechanisms only differ in one element in the cost set, the regularized virtual costs are not going to change a lot, and so is the allocation rule.
\item Then we prove that if  $(A'^{,i},P'^{,i})$ is used at each round,  the performance of the output estimator is no worse than a constant times the benchmark. Our budget allocation method and the random arriving order play crucial roles here. We give the basic idea of our budget allocation rule below. 
\end{enumerate}

\paragraph{A simplified modeling of the budget allocation.} We're able to show that if we allocate average budget $\overline{b}_i$ at round $i$,  the ``loss'' (of performance) occurred by our mechanism at round $i$ is bounded by $\frac{r\cdot B}{i \cdot \overline{b}_i}$ times the ``loss'' occurred by the benchmark $A^*$, where $r$ is a constant. So the optimal budget allocation is essentially the following problem
\begin{align*}
\min_{\overline{b}_i} \quad & r \cdot B \sum_{i=1}^n \frac{1}{i \cdot \overline{b}_i} \\
\textrm{s.t. } \quad & \sum_{i=1}^n \overline{b}_i \le B,\\
 & \overline{b}_i \ge 0.
\end{align*}
which gives $\overline{b}_i \propto \frac{1}{\sqrt{i}}$. In addition, it can be shown that the total ``loss'' of our mechanism is no more than a constant times the total ``loss'' of the benchmark when this budget allocation rule is applied.   

\paragraph{Optimality of the online mechanisms.} Our mechanisms are proved to be optimal within constant factors. But it remains open whether our constant factors are optimal.  One possible method to improve the constant is to only collect agents' costs without purchasing any data at the beginning and start collecting data after a more accurate cost estimation is acquired. But such a method would weaken the incentive guarantee for the agents to report their true costs: if the agents' report costs are not going to affect their rewards, why would they report the true costs? It is also challenging to analyze the performance of such a method. As we've shown in Theorem~\ref{lem:unbiased_chrz} and Theorem~\ref{thm:CI_roundi}, the optimal allocation rule is defined by the regularized virtual costs. It may not be easy to estimate the regularized virtual costs because (1) the regularized virtual cost function is a global property of the cost distribution; (2) the value of the virtual cost function is very sensitive to the PDF of the cost distribution, which appears in the denominator  of the second term: $f(c_i)/(c_i -c_{i-1})$ in $\psi(c_i) = c_i + \frac{1}{f(c_i)/(c_i -c_{i-1})} F(c_{i-1})$ (see Definition~\ref{def:virtual_costs}). The analysis of our online mechanisms is possible because the intermediate mechanism $(A'^{,i},P'^{,i})$ only differs from the optimal mechanism $(A^i, P^i)$ in one element in the cost set, which doesn't affect the regularized virtual costs much. However, our analysis cannot be easily extended to analyze the performance of the above mentioned alternate method as its difficult to bound the estimation error of the regularized virtual costs. The performance of this alternate method remains an open question. 

\subsection{Unbiased Estimator}
We first introduce the benchmark to which we compare our online algorithm. 
\begin{definition} \label{def:unbiased_benchmark}
Let $c_{(1)}  \le  \cdots \le c_{(n)}$ be the $n$ data holders' costs ordered from smallest to largest. Suppose there is one more data holder with cost $\overline{C}$. We define the benchmark $(A^*, P^*)$ to be the mechanism that purchases data from these $n+1$ data holders, and minimizes the worst case variance when it knows the set of costs  $\{c_{(1)}, c_{(2)}, \dots, c_{(n)}, \overline{C}\}$ at the beginning.
\begin{align} \label{prog:unbiased_benchmark}
A^* =  OPT(n+1, \ \{c_{(1)}, c_{(2)}, \dots, c_{(n)},  \overline{C}\}, \ B).
\end{align}
\end{definition}
This additional cost $\overline{C}$ can be interpreted as the loss of unknown upper bound of the possible cost. When the cost distribution is known, the mechanism knows the exact maximum cost $c_{(n)}$, and thus the optimal mechanism will never have a positive probability to buy a data point with cost higher than $c_{(n)}$. But when $c_{(n)}$ is unknown, the mechanism always has to buy any data point (with cost under $\overline{C}$) with a positive probability.

We show that our online mechanism satisfies all three constraints and its worst-case variance is roughly within a constant times the benchmark.
\begin{theorem} \label{thm:unbiased_online}
If we use Mechanism~\ref{alg:online_unbiased} with
\begin{itemize}
    \item $\xi_n = \frac{1}{4\sqrt{n}}$.
    \item At round $i$, use the extended allocation rule and payment rule (Definition~\ref{def:extend_mech}) of $(A^i, P^i)$, where $$A^i = APPROX\left( i, \ T_i,\ \xi_n B \sqrt{i}\right)$$ and $P^i$ is computed as in Lemma~\ref{myerson}. Let the collected data be $\widehat{x_i}$.
    \item Output unbiased estimator $S = \frac{1}{n} \sum_{i=1}^n y_i = \frac{1}{n} \sum_{i=1}^n \frac{\widehat{x}_i}{A^i(c_i)}$ at last.
\end{itemize}
Then the mechanism satisfies (1) truthfulness and individually rationality, (2) the expected total spending is no more than $B=n\overline{B}$, and  (3) for any cost distribution $\{c_1, \dots, c_n\}$ the worst-case variance of the final estimator $S$ is no more than 
$$16\cdot \left(\left(1+\frac{1}{n}\right)^2 \cdot \Var^*(A^*)+\frac{1}{n}+ \frac{1}{n\sqrt{n}} \cdot \frac{1}{A^*(\overline{C})}\right), $$
where  $A^*$ is the benchmark defined in Definition~\ref{def:unbiased_benchmark}.
\end{theorem}
\paragraph{Discussion:} We have the factor $\left(1+\frac{1}{n}\right)^2$ in the first term of our upper bound because the benchmark mechanism has one more data point. It is no more than $4$ when $n\ge 1$ and goes to $1$ when $n$ gets large. The second additive term $\frac{1}{n}$ is due to our estimation of $\Var(S)$ by $\E[S^2]$. We know that  $\Var^*(A^*)$ is roughly $\frac{1}{n+1} \cdot \E[\frac{1}{A^*(c)}]$. So when the problem is non-trivial, we should have the average $\frac{1}{A^*(c)}$ much larger than $1$, and  $\frac{1}{n}$ will be small compared to $\Var^*(A^*)$. The last additive term $\frac{1}{n\sqrt{n}} \cdot \frac{1}{A^*(\overline{C})}$ is dominated by $\frac{1}{n}$. It is only comparable to $\Var^*(A^*)$ when $\sqrt{n} \le  \frac{1}{A^*(\overline{C}) \cdot \E[\frac{1}{A^*(c)}]}$

The complete proof can be found in Appendix~\ref{app:unbiased_online}.

\subsection{Confidence Interval}

Our benchmark for the online mechanism is again the optimal mechanism that knows the cost distribution at the beginning and uses a single optimal mechanism $(A^*, U^*, P^*)$ throughout the survey. We still add an additional cost $\overline{C}$ into the underlying cost set of the benchmark mechanism, in order to make the comparison possible without knowing the exact maximum cost.
\begin{definition} \label{def:biased_benchmark}
Let $c_{(1)}  \le  \cdots \le c_{(n)}$ be the $n$ data holders' costs ordered from smallest to largest. Suppose there is one more data holder with cost $\overline{C}$. We define the benchmark $(A^*, U^*, P^*)$ to be the mechanism that purchases data from these $n+1$ data holders, and minimizes the worst case variance when it knows the set of costs  $\{c_{(1)}, c_{(2)}, \dots, c_{(n)}, \overline{C}\}$ at the beginning.
\begin{align} \label{prog:unbiased_benchmark}
A^*, U^* =  OPT_{CI}(n+1, \ \{c_{(1)}, c_{(2)}, \dots, c_{(n)}, \overline{C}\}, \ B).
\end{align}
\end{definition}

\begin{theorem} \label{thm:CI_main}
If we use Mechanism~\ref{alg:online_unbiased} with
\begin{itemize}
    \item $\xi_n = \frac{1}{16\sqrt{n}}$.
    \item At round $i$, compute $$(A^i, U^i) = APPROX_{CI}(i, \ T_i, \ \xi_n B \sqrt{i}).$$ Let $\mathbbm{1}\left(U^i(c) \ge \frac{1}{2}\right)$ be the rule that completely ignores data with cost $c$ if $U^i(c) \ge \frac{1}{2}$, and never ignores the data if $U^i(c) < \frac{1}{2}$. Then the mechanism purchases agent $i$'s data using the extended allocation and payment rule of $$\left(A^i,  \mathbbm{1}\left(U^i(c) \ge \frac{1}{2}\right), P^i\right),$$ where $P^i$ is computed as in Lemma~\ref{myerson}. Let the collected data be $\widehat{x_i}$.
    \item Output confidence interval 
        $$
        \left[\sum_{i=1}^n y_i / n - \frac{\alpha_\gamma}{\sqrt{n}} \cdot \widehat{\sigma}, \quad \sum_{i=1}^n y_i/ n + \frac{\sum_{i=1}^n \widehat{U}_i}{n} + \frac{\alpha_\gamma}{\sqrt{n}} \cdot \widehat{\sigma}\right],
        $$
        where $y_i = \frac{\widehat{x}_i}{A^i(c_i)}$ and  $\widehat{\sigma}^2$ is the sample variance of $y_1,\dots, y_n$, and $\widehat{U}_i$ represents whether the $i$-th data point is ignored or not.
\end{itemize}
Then the mechanism is (1) truthful in expectation and individually rational; (2) satisfies the budget constraint $B=n\overline{B}$ in expectation;  (3) and the for any cost distribution $\{c_1, \dots, c_n\}$, the worst-case expected length of the output confidence interval is no more than 
$$
8\sqrt{10} \cdot L^*+ \frac{2\sqrt{10}}{\sqrt{n}} +o\left(\frac{1}{\sqrt{n}}\right), 
$$
where $L^*$ is the approximate worst-case expected length of the benchmark as defined in Lemma~\ref{lem:biased_benchmark}.
\end{theorem}
\begin{corollary}
The worst-case expected length of our mechanism's output confidence interval is no more than 
$$
8 \sqrt{10} \cdot MIN+ \frac{2\sqrt{10}+32\sqrt{10}\cdot \alpha_\gamma}{\sqrt{n}} +o\left(\frac{1}{\sqrt{n}}\right), 
$$
where $MIN$ is the worst-case expected length of the optimal confidence interval estimator.
\end{corollary}
\paragraph{Discussion:} As we show in the previous section, $L^*$ is roughly $\frac{1}{\sqrt{n}}\cdot \left(2\alpha_\gamma \sqrt{ \E[\frac{1-U^*_c}{A^*(c)}]} + \sqrt{n} \cdot \E[U^*_c]\right)$. If the problem is non-trivial, we should have  $\frac{1}{\sqrt{n}}$ dominated by $L^*$. 
The complete proof can be found in~\ref{app:CI_main}.
\section{Discussions}\label{sec:discussion}
In this work, we restrict our estimators to use only the collected data.  When the data are correlated with the costs, the data analyst may gradually learn the cost-data correlation based on the collected pairs. This means that if a data holder arrives and reports his cost, the data analyst may form a prediction for his data based on his reported cost and the learned cost-data correlation, even if the data is not collected. This leads to an interesting open question: can the final estimation be improved if the data analyst makes use of such predicted data? 

Allowing the data analyst to leverage on the cost-data correlation brings up an additional level of challenge when the data holders care about the privacy of their data. Such data holders may hesitate to report their costs, because reporting the cost itself reveals some information about his data. This makes it more challenging to achieve truthfulness in design an online mechanism. 

Another open problem is whether it is possible to do better than the worst-case analysis. The optimality of our mechanism is based on the worst-case cost-data correlation. When the designer can gradually learn the cost-data correlation, is it possible to adjust the purchasing mechanism accordingly so that it achieves optimality with respect to the actual cost-data correlation?

More generally, it would be interesting to develop mechanisms for other more complicated statistical estimation tasks, such as hypothesis testing. 


\newpage
\bibliographystyle{plainnat}
\bibliography{sample,bibliography,library}

\newpage
\appendix

\section{Concentration Bounds}
\subsection{Empirical Bernstein's Inequality} \label{app:concentration}
We show how to construct confidence interval using sample mean and sample variance of $y_i$'s.
\begin{lemma}[\cite{maurer2009empirical}] \label{empBern}
Let $X_1, \dots, X_n$ be independent random variables taking their values in
$[0, M]$. Consider the sample expectation $\widehat{\mu}$ and sample variance $\widehat{\sigma}^2$ defined respectively by $\widehat{\mu} = \sum_i X_i/n$ and
$
\widehat{\sigma}^2 = \sum_{i=1}^n (X_i - \widehat{\mu})^2/(n-1).
$
Then with probability at least $1-\gamma$,
$$
\left\vert \widehat{\mu} - \E\widehat{\mu}\right\vert \le \widehat{\sigma}\sqrt{\frac{2 \ln(4/\gamma)}{n}} + \frac{7M\ln(4/\gamma)}{3n}.
$$
\end{lemma}
As we proved in the main text, $y_1 = \frac{\widehat{x}_1}{A^1(c_1)}, \dots, y_n = \frac{\widehat{x}_n}{A^n(c_n)}$ become independent when conditioning on a realization of arriving sequence $(\tilde{\vec{c}}, \tilde{\vec{z}})$, and when the estimator is unbiased, for any realization $(\tilde{\vec{c}}, \tilde{\vec{z}})$, $\E[\sum_i y_i/n] = \sum_i z_i/n = \E[z]$. Therefore by the above lemma, with probability at least $1-\gamma$,
$$
\left| \sum_i y_i/n - \E[z] \right| \le \widehat{\sigma}\sqrt{\frac{2 \ln(4/\gamma)}{n}} + \frac{7M\ln(4/\gamma)}{3n},
$$
where $M = \max_i y_i$. Therefore the following interval is a confidence interval of $\E[z]$ with confidence level $\gamma$,
$$
\left[\sum_{i=1}^n y_i / n - \frac{\alpha_\gamma}{\sqrt{n}} \cdot \widehat{\sigma}, \quad \sum_{i=1}^n y_i/ n + \frac{\alpha_\gamma}{\sqrt{n}} \cdot \widehat{\sigma}\right]
$$
where $\alpha_\gamma = \sqrt{2 \ln(4/\gamma)} +  \frac{7M\ln(4/\gamma)}{3\sqrt{n}\cdot \widehat{\sigma}}$. Notice that the second term in $\alpha$, $ \frac{7M\ln(4/\gamma)}{3\sqrt{n}\cdot \widehat{\sigma}}$ is not a constant. But since we estimate $\widehat{\sigma}^2$ with $\sum y_i^2/n$ (see Section~\ref{sec:biased_benchmark}) when computing the optimal confidence interval mechanism, 
$$
\frac{M}{\sqrt{n} \cdot \widehat{\sigma}}= \frac{\max_i y_i}{\sqrt{n}\cdot \sqrt{\sum_{i=1}^n y_i^2/n}}\le \frac{\max_i y_i}{\sqrt{\sum_{i=1}^n y_i^2}}\le 1,
$$
we can just assume that $\alpha_\gamma$ is a constant in the range of $\sqrt{2 \ln(4/\gamma)}$ to $\sqrt{2 \ln(4/\gamma)} +  \frac{7\ln(4/\gamma)}{3}$ when solving the optimization problem.

\subsection{Estimation of Sample Standard Deviation} \label{app:difference}
We show that when $0 \le \E[y_i] \le 1$, the difference between $\E[\widehat{\sigma} ]$ and $\sqrt{\E \sum_{i=1}^n y_i^2/n}$ is no more than $1 + O(1/n)$, where $\widehat{\sigma}^2$ is the sample variance of $y_1, \dots, y_n$. 

We first show that the difference between $\E[\widehat{\sigma}^2]$ and $\frac{1}{n} \cdot \sum_{i=1}^n \E[ y_i^2]$ is no more than $1$.
Because $y_i$'s are independent conditioned on $c_1, \dots, c_n$,
\begin{align*}
\E[\widehat{\sigma}^2 \vert c_1, \dots, c_n]= &\E\left[\sum_{i=1}^n \left(y_i - \sum_{i=1}^n y_i/n\right)^2/(n-1) \Big \vert c_1, \dots, c_n\right]\\
= & \frac{1}{n-1} \cdot \E\left[\sum_{i=1}^n y_i^2 - \frac{2}{n}\sum_{i=1}^n y_i \sum_{i=1}^n y_i + n \cdot \left( \sum_{i=1}^n y_i/n\right)^2 \Big\vert c_1, \dots, c_n\right]\\
= & \frac{1}{n-1} \cdot \E\left[\sum_{i=1}^n y_i^2 - \frac{1}{n} \left(\sum_{i=1}^n y_i \right)^2 \Big \vert c_1, \dots, c_n\right]\\
= & \frac{1}{n-1} \cdot \E\left[\sum_{i=1}^n y_i^2 - \frac{1}{n} \sum_{i=1}^n y_i^2 -\frac{1}{n} \sum_{i=1}^n \sum_{j\neq i} y_i y_j \Big \vert c_1, \dots, c_n\right]\\
= & \frac{1}{n} \cdot \sum_{i=1}^n \E[ y_i^2]  -\frac{1}{n(n-1)} \sum_{i=1}^n \sum_{j\neq i} \E[y_i \vert c_1, \dots, c_n] \E[ y_j \vert c_1, \dots, c_n]\\
\in & \left[\frac{1}{n} \cdot \sum_{i=1}^n \E[ y_i^2] -1, \quad \frac{1}{n} \cdot \sum_{i=1}^n \E[ y_i^2]\right]
\end{align*}
Therefore 
$$
\sqrt{\frac{1}{n} \cdot \sum_{i=1}^n \E[ y_i^2]} - 1 \le \sqrt{\E[\widehat{\sigma}^2 \vert c_1, \dots, c_n]} \le \sqrt{\frac{1}{n} \cdot \sum_{i=1}^n \E[ y_i^2]}
$$
for all $c_1, \dots, c_n$.
It has been shown (e.g. in~\cite{gurland1971}) that the difference between $\E[\widehat{\sigma} | c_1, \dots, c_n]$ and $\sqrt{\E[\widehat{\sigma}^2 | c_1, \dots, c_n]}$ decreases as sample size grows, dropping off as $O(1/n)$. Therefore 
$$
\sqrt{ \frac{1}{n} \cdot \sum_{i=1}^n \E[ y_i^2]} -1 -O(1/n) \le \E[\widehat{\sigma}] \le \sqrt{  \frac{1}{n} \cdot \sum_{i=1}^n \E[ y_i^2]}.
$$ 
The difference between $\E[\widehat{\sigma} ]$ and $\sqrt{\E \sum_{i=1}^n y_i^2/n}$ is no more than $1 + O(1/n)$.
\section{Truthfulness and Virtual Cost Function}
\subsection{Extended Allocation Rule and Payment Rule} \label{app:extended_alloc}
We apply the well-known Myerson's lemma to prove Lemma~\ref{lem:extended_mech}. 
\begin{lemma}[\cite{myerson1983efficient}] \label{lem:true_myerson}
A survey mechanism $(A(c), P(c))$ defined on $c\in [0, \overline{C}]$ is truthful and pays $0$ when $A(c) = 0$ if and only if 
\begin{itemize}
\item $A(c)$ is a monotone non-increasing function of $c$.
\item $P(c) = c + \frac{1}{A(c)} \cdot \int_{c}^{\overline{C}} A(v) \, dv$.
\end{itemize}
\end{lemma}

Let $(A^d,P^d)$ be a truthful and individually rational survey mechanism defined on a discrete cost set $\{c_1, \dots, c_m\}$ with $c_1 \le \cdots \le c_m$. Then when we use the extended allocation rule and payment rule
$$
A(c) = A^d(\lceil c \rceil), \quad P(c) = P^d(\lceil c \rceil), \textrm{ for all } c\in [0, c_m].
$$
The allocation rule $A(c)$ must be monotone non-increasing because $A^d$ must be monotone non-increasing. 
Suppose $\lceil c \rceil = c_i$, then 
\begin{align*}
P(c) = P^d(c_i) = & c_i + \frac{1}{A(c_i)} \sum_{j=i+1}^m A(c_j)(c_{j}-c_{j-1})\\
 =& c + (c_i - c) + \frac{1}{A(c)}\sum_{j=i+1}^m A(c_j)(c_{j}-c_{j-1}) \\
  =& c + (c_i - c) + \frac{1}{A(c)}\sum_{j=i+1}^m \int_{c_{j-1}}^{c_{j}} A(c_j) \,dv\\
 = & c + \frac{1}{A(c)} \cdot \int_{c}^{c_m} A(v) \, dv,
\end{align*}
which is just the payment rule that satisfies the second condition of the Myerson's lemma. Therefore $(A(c), P(c))$ is truthful. It is also individually rational because we always have $P(c) \ge c$. 

\subsection{Virtual Cost Function}
The following results on the virtual cost function (Definition~\ref{def:virtual_costs}) will be used in our proofs.
\begin{lemma} \label{lem:VC_subset}
Let $\psi(c)$ be the virtual cost function of the discrete uniform distribution over set $T = \{c_1, \dots, c_n\}$. Let $T_i$ be an arbitrary subset of  $T$ with length $i$ and let $\psi^i(c)$ be the virtual cost function of the discrete uniform distribution over set $T_i$. Then for any non-increasing allocation rule $A(c)$, 
$$
\sum_{c \in T_i} A(c) \psi^i(c) \le \sum_{c \in T} A(c) \psi(c).
$$
\end{lemma}
\begin{proof}
Without loss of generality, assume $c_1 \le \cdots \le c_n$ and $T_i = \{ c_{p_1}, \dots, c_{p_i}\}$.  Then according to the definition of $\psi(c)$ and $\psi^{i}(c)$, for all $1\le k \le i$,
\begin{align*}
\sum_{j=p_{k-1} + 1}^{p_k} \psi(c_j) = & \sum_{j=p_{k-1} + 1}^{p_k} j \cdot c_j - (j-1) c_{j-1} \\
= & (p_k - p_{k-1} )\cdot c_{p_k} + p_{k-1} \cdot ( c_{p_k} - c_{p_{k-1}})\\
\ge & c_{p_k} + (k-1) (c_{p_k} - c_{p_{k-1}}) \\
= & \psi^i(c_{p_k}).
\end{align*}
since $p_k - p_{k-1} \ge 1$ and $p_{k-1} \ge k-1$. 

For all $c_j \in T$, define $\lceil c_j \rceil$ to be the smallest cost in $T_i$ that is greater than or equal to $c_j$.
\begin{align*}
\sum_{c \in T_i} A(c) \psi^i(c) = & \sum_{k = 1}^i   A(c_{p_k}) \psi^i(c_{p_k}) \\
\le  &  \sum_{k = 1}^i A(c_{p_k}) \left(\sum_{j=p_{k-1} + 1}^{p_k} \psi(c_j) \right)\\
= & \sum_{j=1}^n A(\lceil c_j\rceil) \psi(c_j)\\
\le & \sum_{j=1}^n A( c_j ) \psi(c_j)\\ 
\end{align*}
since $A$ is monotone non-increasing,
\end{proof}

\section{Proofs for Optimal Unbiased Estimator} 
In this section, we give the proofs for the optimal online unbiased estimator. In Section~\ref{app:unbiased_chrz}, we characterize the  optimal $A$ by the use of the Lagrangian function and the KKT conditions. Since the objective function of our optimization problem is convex, we prove the optimality by constructing the dual variables that satisfy the KKT conditions with our primal solution. This is based on some properites of the regularized virtual cost function $\phi$, which is presented at the beginning of the section. In Section~\ref{app:unbiased_adj}, we prove some lemmas that will be used in Section~\ref{app:unbiased_online} when we show the budget feasibility and the performance of our online mechanism. 

\subsection{Proof of Theorem~\ref{lem:unbiased_chrz}} \label{app:unbiased_chrz}

We first prove the following properties of the regularized virtual cost function $\phi$ when compared with the non-regularized virtual cost function $\psi$. 
\begin{lemma} \label{lem:unbiased_chara_lem1}
For a discrete uniform distribution  supported on $\{c_1, \dots, c_m\}$ with $c_1 \le \cdots \le c_m$. Let $\psi$ be its virtual cost function and $\phi$ be the regularized virtual cost function. For any  $1\le i \le m$, let $I_i$ be the set of all the $j$ that has $\phi(c_j) = \phi(c_i)$. Then 
\begin{enumerate}
\item $\phi(c)$ is a non-decreasing function of $c$. 
\item  $\sum_{j \in I_i} \phi(c_j) = \sum_{j \in I_i} \psi(c_j)$.
\item  $\sum_{j=1}^i \phi(c_j) \le \sum_{j=1}^i \psi(c_j)$ for all $1\le i \le m$. If $\phi(c_i) \neq \phi(c_{i+1})$ then $\sum_{j=1}^i \phi(c_j) = \sum_{j=1}^i \psi(c_j)$.
\end{enumerate}
\end{lemma}
\begin{proof}
Define $Avg(i,k)$ as the average of $\psi(i), \dots, \psi(k)$, i.e.,
$
Avg(i,k) = \frac{1}{k - i + 1} \sum_{j = i}^{k} \psi(c_j).
$
Recall that the definition of $\phi(c_i)$ is as follows
\begin{equation}  \label{unbiased_phi_eq1}
\phi(c_i) = \max \{ \psi'(c_1), \dots, \psi'(c_i)\},
\end{equation}
\begin{equation*}
\psi'(c_i) = Avg(i, R_i),
\end{equation*}
where $R_i$ is the $k$ to the right of $i$ that minimizes $Avg(i,k)$, i.e.,
\begin{align*}
R_i = \arg \min_{k: k\ge i}  Avg(i,k).
\end{align*}
If there are multiple $k$'s that achieve the minimum, without loss of generality let $R_i$ be the maximum of them.

We prove the following properties of the intervals $[i, R_i]$. 
\begin{claim} \label{unbiased_chara_claim1}
The intervals $[1, R_1], [2, R_2], \dots, [m, R_m]$ will not partially intersect, i.e., if $i<j$ and $j\le R_i$ then $R_j \le R_i$.
\end{claim}
\begin{proof}
 We prove by contradiction. Suppose $i < j \le R_i < R_j$, then since
\begin{eqnarray*}
Avg(j, R_i) \le Avg(i,j)\\
Avg(j, R_i) \ge Avg(j, R_j)
\end{eqnarray*}
it satisfies that
$$
Avg(i, R_i) \ge Avg(i, R_j),
$$
which contradicts the definition of $R_i$. 
\end{proof}

\begin{claim} \label{unbiased_chara_claim2}
 For any $j$ that is contained in $[i, R_i]$, it always holds that $Avg(j, R_j) \le Avg(i, R_i)$, or equivalently $\psi'(c_j) \le \psi'(c_i)$. 
 \end{claim}
\begin{proof}
First by Claim~\ref{unbiased_chara_claim1}, $[j,R_j]\subseteq [i, R_i]$. Then the claim can also be proved by contradiction. Assume there exists $i < j \le R_j \le R_i$ that has $Avg(j, R_j) > Avg(i, R_i)$. Then since $R_j$ minimizes the average $Avg(j, k)$ over all $k\ge j$, we must also have $Avg(j, R_i)>Avg(i, R_i)$. But $Avg(j, R_i)>Avg(i, R_i)$ will immediately lead to $Avg(i, j-1) < Avg(i, R_i)$, as
\begin{align*}
&Avg(i, j-1) * (j-i) + Avg(j, R_i) * (R_i - j +1) \\
=& Avg(i, R_i) * ( j-i) + Avg(i, R_i) * (R_i - j +1),
\end{align*}
which is contradictory to the definition of $R_i$ that $R_i$ minimizes $Avg(i, k)$ over all $k\ge i$. 
\end{proof}

\begin{claim}\label{unbiased_chara_claim3}
If an interval $[i, R_i]$ is not contained in any other intervals. Then for any $j<i$, it holds that $Avg(j, R_j) < Avg(i, R_i)$, or equivalently $\psi'(j)<\psi'(i)$.
\end{claim}
\begin{proof}
By Claim~\ref{unbiased_chara_claim1}, $R_j<i$. We then prove by induction. First for $j=i-1$, by the definition of $R_j$, we must have $Avg(j, R_j) = Avg(j,j) < Avg(i, R_i)$, or equivalently $\psi'(j)<\psi'(i)$. Now suppose the claim holds for any $j\in [k+1, i-1]$, then for $j=k$, it must hold that $Avg(j, R_j)<\psi'(R_j + 1) \le \psi'(i)$.  Again this is because $R_j$ minimizes $Avg(j, R_j)$.
\end{proof}

With the three results above, our lemma can be proved as follows: 
\paragraph{(1):} According to definition~\eqref{unbiased_phi_eq1}, $\phi$ is non-decreasing.
\paragraph{(2):} Consider a fixed $i$, let $L$ be the minimal $j$ that has $R_j \ge i$, in other words let $[L, R_L]$ be the maximal interval that contains $i$. Then by Claim~\ref{unbiased_chara_claim2} and Claim~\ref{unbiased_chara_claim3}, 
$$ \phi(c_i) = \max \{ \psi'(c_1), \dots, \psi'(c_i)\} = \psi'(c_L) =Avg(L, R_L).$$
The same holds for any other $i \in [L, R_L]$. By Claim~\ref{unbiased_chara_claim3}, for any $j<L$, $$\phi(c_j) =  \max \{ \psi'(c_1), \dots, \psi'(c_j)\} < \psi'(c_L).$$ And similarly for any $j>R_L$, $\psi'(c_j) > \psi'(c_L)$.  Therefore $I_i = \{L, \dots, R_L\}$ and $\sum_{j\in I_i} \phi(c_j) = Avg(L, R_L) * (R_L - L +1) = \sum_{j\in I_i} \psi(c_j)$.
\paragraph{(3):} Consider a fixed $i$, let $I_i = \{L, \dots, R_L\}$ as proved in (2). If $\phi(c_i) \neq \phi(c_{i+1})$, then $i$ must be the right end point of $I_i$. Then by summing over all $I_j$ to the left of $I_i$, we get the equality $\sum_{j=1}^i \phi(c_j) = \sum_{j=1}^i \psi(c_j)$. Otherwise we have $\sum_{j=1}^{L-1} \phi(c_j) = \sum_{j=1}^{L-1} \psi(c_j)$ and $\sum_{j=L}^i \psi(c_j) = Avg(L, i) * (i-L+1) \ge Avg(L, R_L)*(i-L+1) = \sum_{j=L}^i \phi(c_j)$. 
\end{proof}

Now we prove the theorem:  the solution of 
\begin{align} \label{prog:A_i_lalalala}
A^*= &\arg \min_A \max_{\vec{z} \in [0,1]^n} \sum_{k = 1}^n \frac{z_k^2}{A(c_{(k)})}\\ 
& \quad \textrm{s.t. } \sum_{k=1}^n A(c_{(k)}) \psi(c_{(k)}) \le B\notag\\
& \quad \qquad A(c_{(k)})\ge A(c_{(k+1)}), \forall k \notag\\
& \quad  \qquad 0 \le A(c_{(k)}) \le 1, \quad \forall k \notag
\end{align}
is
$A^*(c_{(j)}) = \min\left\{ 1, \frac{\lambda}{\sqrt{\phi(c_{(j)})}} \right\}$ for all $1\le j\le n$, and $\lambda$ is chosen such that the budget constraint is satisfied with equality
$$
 \sum_{j=1}^n A(c_{(j)}) \psi(c_{(j)}) = B
$$
where $\psi$ and $\phi$ are the non-regularized and regularized virtual cost function when the underlying cost distribution is the uniform distribution over $C$. The value of $\lambda$ can be computed using binary search.

\begin{proof}
For simplicity, we write $A_k = A(c_{(k)})$ and $A^*_k = A^*(c_{(k)})$.  Notice that the objective function is an increasing function of $z_j$ for all $j$ and thus the maximum is obtained when $z_j = 1$ for all $j$.
Then the objective function of our minimization problem becomes $\sum_{j=1}^n \frac{1}{A_j}$. In the region $A_j > 0$ for all $j$, the objective function is convex. Therefore  the KKT conditions are necessary and sufficient for optimality of the primal and dual variables (see~\cite{Boyd:2004:CO:993483}). 
The Lagrangian function of the optimization problem is
\begin{align*}
L( A, \xi, \vec{\pi}, \vec{\eta}^0, \vec{\eta}^1) = &\sum_{j=1}^n \frac{1}{A_j} \\
 & + \xi \left( \sum_{j=1}^n A_j \psi(c_{(j)}) - B \right) \\
 & + \sum_{j=1}^{n-1} \pi_j(A_{j+1} - A_j)\\
 & + \sum_{j=1}^n \eta^0_j A_j + \sum_{j=1}^n \eta^1_j (A_j - 1).
\end{align*}
We prove that the optimal primal variables are $A^*(c) = \min\left\{ 1, \frac{\lambda}{\sqrt{\phi(c)}} \right\}$, where  $\lambda$ is chosen such that the budget constraint is satisfied with equality,
and the optimal dual variables are 
\begin{align*}
&\xi = \frac{1}{\lambda^2}, \\
& \pi_j = \pi_{j-1} + \xi\left(\psi(c_{(j)}) - \phi(c_{(j)})\right),  \textrm{ (Here we assume $\pi_{0} = 0$.)}\\
& \eta_j^0 = 0 \textrm{ for all } j, \\
& \eta_j^1 = \left\{ \begin{array}{ll}
				1-\xi \cdot \phi(c_{(j)}), \textrm{ if } \xi \cdot \phi(c_{(j)}) < 1,\\
				0, \textrm{ otherwise.}
			\end{array} \right.
\end{align*}
\paragraph{Primal feasibility:} We first prove the primal feasibility.
\begin{enumerate}
\item By the definition of $A^*$, the budget constraint is satisfied with equality.
\item By (1) in Lemma~\ref{lem:unbiased_chara_lem1}, $\phi(c)$ is a non-decreasing function of $c$. Then $A^*(c) = \min\left\{ 1, \frac{\lambda}{\sqrt{\phi(c)}} \right \}$ is a non-increasing function. 
\item It is easy to verify that $0 \le A^*(c) \le 1$. 
\end{enumerate}
\paragraph{Dual feasibility:}  By (3) in Lemma~\ref{lem:unbiased_chara_lem1}, it is easy to verify that all of the dual variables greater or equal to $0$.  
\paragraph{Stationarity:} The partial derivative of $L( A, \xi, \vec{\pi}, \vec{\eta}^0, \vec{\eta}^1)$ with respect to each $A_j$ is
\begin{align*}
\frac{\partial L( A, \xi, \vec{\eta}^0, \vec{\eta}^1) }{\partial A_j} =  &-   \frac{1}{A^2_j} +\xi  \cdot  \psi(c_{(j)}) +\pi_{j-1} - \pi_j + \eta^0_j + \eta^1_j\\
= & -  \max\left\{ 1, \  \xi \cdot \phi(c_{(j)}) \right\} +\xi  \cdot  \psi(c_{(j)}) +\pi_{j-1} - \pi_j+ \eta^1_j\\
= & - \xi \cdot \phi(c_{(j)})+\xi  \cdot  \psi(c_{(j)}) +\pi_{j-1} - \pi_j.
\end{align*}
By the definition of $\pi_j$, we have the above quantity equal to $0$. 
\paragraph{Complementary slackness:} 
\begin{enumerate}
\item The budget constraint in the primal is satisfied with equality. 
\item For all $A^*_j \neq A^*_{j+1}$, we must have $\phi(c_{(j)}) \neq \phi(c_{(j+1)})$. Then by (3) in Lemma~\ref{lem:unbiased_chara_lem1}, $\pi_j = 0$. Thus $\pi_j(A^*_{j+1} -A^*_j) = 0$ for all $j$.
\item We have $\eta_j^0 = 0$ for all $j$ and $\eta_j^1 = 0$ for all $A^*_j = \frac{1}{\sqrt{\xi \cdot \phi(c_{(j)})}} <1$.
\end{enumerate}

As we have proved, the optimal allocation rule is monotone non-increasing.  Since the sum $\sum_{j=1}^n A^*_j \psi(c_{(j)})$ is an increasing function of $\lambda$, we can perform binary search to find the right value of $\lambda$ such that the sum equals the budget. Moreover, we can reduce the search space to $|C|$ by searching the critical point $c^*\in C$ that has $\frac{\lambda}{\sqrt{ \phi(c) }} \ge 1$ for $c\le c^*$ and $\frac{\lambda}{\sqrt{ \phi(c) }} \le 1$ for $c\ge c^*$.

This completes the proof of Lemma~\ref{lem:unbiased_chrz}.
\end{proof}

\subsection{Optimal Mechanism for Adjacent Cost Sets} \label{app:unbiased_adj}
In this section, we show that the optimal solution $A$ will not change a lot if we slightly modify the optimization problem~\eqref{prog:A_i_lalalala}, more specifically, if the set of costs contains one more element.

We first show that if the set of costs contains one more element, the regularized virtual costs function $\phi(c)$ of the uniform distribution over this set will change no more than a factor of $2$, which is mainly because the virtual costs $\psi(c_i) = i\cdot c_i - (i-1)c_{i-1}$ change at most  by a factor of $2$.
\begin{lemma} \label{lem:unbiased_adj_lem1}
Let $T_1$ and $T_2$ be two costs sets that only differ in one element $c_k$. More specifically, suppose $T_2 = \{c_1, \dots, c_m\}$ with $c_1\le \cdots \le c_m$ and $T_1 = \{c_1, \dots, c_{k-1}, c_{k+1}, \dots, c_m\}$. Let $\phi_1(c)$ be the regularized virtual cost function (Definition~\ref{def:regularized_virtual_costs}) of the discrete uniform distribution over $T_1$ and let $\phi_2(c)$ be the regularized virtual cost function of the discrete uniform distribution over $T_2$. Then it holds that for all $i\neq k$,
    $$
    \frac{1}{2} \cdot \phi_1(c_i) \le \phi_2(c_i) \le 2\cdot \phi_1(c_i).
    $$
\end{lemma}
\begin{proof}
According to Definition~\ref{def:virtual_costs}, the (non-regularized) virtual costs are
\begin{align*}
&\psi_1(c_{1}) = c_{1}, \ \dots,\  \psi_1(c_{k-1}) = (k-1) c_{k-1} - (k-2) c_{k-2}, \\
&\psi_1(c_{k+1}) = k \cdot c_{k+1} - (k-1) c_{k-1}, \\
&\psi_1(c_{k+2}) = (k+1) c_{k+2} - k \cdot c_{k+1}, \ \dots, \ \psi_1(\overline{C}) = i \cdot \overline{C} - (i-1) c_{i}
\end{align*}
and
\begin{align*}
&\psi_2(c_{1}) = c_{1}, \ \dots, \  \psi_2(c_{k-1}) = (k-1) c_{k-1} - (k-2) c_{k-2}, \\
&\psi_2(c_{k}) = k\cdot c_{k} - (k-1) c_{k-1}, \ \psi_2(c_{k+1}) = (k+1) \cdot c_{k+1} - k \cdot  c_{k}, \\
&\psi_2(c_{k+2}) = (k+2) c_{k+2} - (k+1) \cdot c_{k+1}, \ \dots, \  \psi_2(\overline{C}) = (i+1) \cdot \overline{C} - i \cdot c_{i}
\end{align*}
\renewcommand{\labelenumi}{(\roman{enumi})}
\begin{enumerate}
\item For $i,j<k$, $Avg_1(i,j) = Avg_2(i,j)$. 
\item For $Avg_2(i,j)$ that has $i<k$ and $j=k$, we have 
\begin{align*}
Avg_1(i, k-1) - Avg_2(i,k) = &  \frac{(k-1) \cdot c_{k-1} - (i-1) c_{i-1}}{k-i} - \frac{k \cdot c_k - (i-1) c_{i-1}}{k-i+1} \\
 =& \frac{(i-1)(c_k - c_{i-1})}{(k-i) (k-i+1)} \\
 = & \frac{(i-1)(c_k - c_{i-1})}{k-i+1} \cdot \frac{1}{k-i}\\
 \in & \left[ 0, Avg_2(i,k) \right]
\end{align*}
Thus $ \frac{1}{2} \cdot Avg_1(i, k-1) \le Avg_2(i,k) \le  Avg_1(i, k-1)$.
\item For $i < k$ and $j\ge k+1$, we have $Avg_1(i,j) = \frac{(j-1)c_j - (i-1) c_{i-1}}{j-i}$ and $Avg_2(i,j) = \frac{j \cdot c_j - (i-1) c_{i-1}}{j-i+1}$. Then
$$
Avg_1(i,j) - Avg_2(i,j) = \frac{(i-1)(c_j - c_{i-1})}{(j-i)(j-i+1)} \le  \frac{(i-1)(c_j - c_{i-1})}{j-i} \cdot \frac{1}{j-i+1} \in \left[0, Avg_1(i,j) \cdot \frac{1}{2} \right].
$$
Thus $\frac{1}{2} Avg_1(i,j) \le Avg_2(i,j) \le Avg_1(i,j)$.

\item For $i,j \ge k+2$, we have $Avg_1(i,j) = \frac{(j-1)c_j - (i-2) c_{i-1}}{j-i+1}$ and $Avg_2(i,j) = \frac{j \cdot c_j - (i-1) c_{i-1}}{j-i+1}$. Then 
$$
 Avg_2(i,j) - Avg_1(i,j) = \frac{c_j - c_{i-1}}{j-i+1} \le  \frac{c_j }{j-i+1} \le \frac{(j-1)c_j - (i-2) c_{i-1}}{j-i+1} \in \left[0, Avg_1(i,j)\right]. 
$$
Thus $Avg_1(i,j) \le Avg_2(i,j) \le 2\cdot Avg_1(i,j)$.
\end{enumerate}
According to the above four properties of the $Avg$ function and the definition of $\psi'(c_i)$,  we have 
\begin{align} \label{eqn:unbiased_adj_eq1}
\frac{1}{2} \cdot \psi'_1(c_i) \le \psi'_2(c_i) \le 2\psi'_1(c_i), \textrm{ for } i <k \textrm{ and } i\ge k+2.
\end{align}
To compare $\psi'_1(c_i)$ and $\psi'_2(c_i)$ for the case $i = k+1$, we prove the follows
\begin{description}
\item[$\bullet $ $\boldsymbol{\mathbf{\max\{ \psi'_2(c_k), \psi'_2(c_{k+1})\} \ge \frac{1}{2}\psi'_1(c_{k+1}):}}$] Observe that 
$$
\psi_2(c_k) + \psi_2(c_{k+1}) = (k+1) c_{k+1} - (k-1) c_{k-1} \ge k \cdot c_{k+1} - (k-1) c_{k-1} = \psi_1(c_{k+1}).
$$
Therefore $Avg_2(k, k+1) \ge \frac{1}{2} Avg_1(k+1, k+1)$ and at least one of the $\psi_2(c_k)$ and $\psi_2(c_{k+1})$ must be no less than $\psi_1(c_{k+1})/2$. If $\psi_2(c_k) \ge \psi_1(c_{k+1})/2$, then 
\begin{align*}
&Avg_2(k, k) \ge \frac{1}{2} Avg_1(k+1, k+1)\\
&Avg_2(k, k+1) \ge \frac{1}{2} Avg_1(k+1, k+1)
\end{align*}
In addition, by (iv) we have $Avg_2(k+2,j) \ge Avg_1(k+2,j)$ for all $j\ge k+2$, so 
\begin{align*}
& Avg_2(k, j) \ge \frac{1}{2} Avg_1(k+1, j)
\end{align*}
for all $j\ge k+2$. Then 
$$
\psi'_2(c_k) = \min_{j \ge k} Avg_2(k,j) \ge \frac{1}{2} \min_{j \ge k+1} Avg_1(k+1,j) = \frac{1}{2}\psi'_1(c_{k+1}).
$$
Similarly for the other case $\psi_2(c_{k+1}) \ge \psi_1(c_{k+1})/2$ we can get $\psi'_2(c_{k+1}) \ge \frac{1}{2}\psi'_1(c_{k+1})$.

\item[$\bullet $ $\boldsymbol{\mathbf{\max\{ \psi'_2(c_k), \psi'_2(c_{k+1})\} \le 2\psi'_1(c_{k+1}):}}$] Observe that
\begin{align*}
\psi_2(c_k) + \psi_2(c_{k+1}) = &(k+1) \cdot c_{k+1} - (k-1) c_{k-1}\\
= &c_{k+1}+ (k \cdot c_{k+1} - (k-1) c_{k-1}) \\
\le & 2(k \cdot c_{k+1} - (k-1) c_{k-1}) \\
= & 2\psi_1(c_{k+1}).
\end{align*}
Thus both $\psi_2(c_k)$ and $\psi_2(c_{k+1})$ is no greater than $2\psi_1(c_{k+1})$. Furthermore by (iv), $Avg_2(k+2, j) \le 2 Avg_1(k+2, j)$ for all $j\ge k+2$, it holds that
\begin{align*}
Avg_2(k, j) \le 2 Avg_1(k+1, j), \quad Avg_2(k+1, j)\le 2 Avg_1(k+1, j)
\end{align*}
for all $j\ge k+1$. Thus 
\begin{align*}
&\psi'_2(c_k) = \min_{j\ge k} Avg_2(k, j) \le \min_{j\ge k+1} 2 Avg_1(k+1, j) = 2\psi'_1(c_{k+1})\\
&\psi'_2(c_{k+1}) = \min_{j\ge k+1} Avg_2(k+1, j) \le \min_{j\ge k+1} 2 Avg_1(k+1, j) = 2\psi'_1(c_{k+1})
\end{align*}
\end{description}
Combining the above two inequalities with~\eqref{eqn:unbiased_adj_eq1}, 
\begin{align*}
 \phi_2(c_i) = \max_{j\le i} \psi'_2(c_j) \le 2\cdot \max_{j\le i} \psi'_1(c_j) = 2\phi_1(c_i), \textrm{ for all } i\neq k\\
 \phi_1(c_i) = \max_{j\le i} \psi'_1(c_j) \le 2\cdot \max_{j\le i} \psi'_2(c_j) = 2\phi_2(c_i), \textrm{ for all } i\neq k
\end{align*}
which completes the first part of the proof. 
\end{proof}
Based on the above analysis, we further prove the follows.
\begin{lemma} \label{lem:unbiased_adj_lem2}
Let $T_1$ and $T_2$ be two costs sets that have $T_2 = \{c_1, \dots, c_m\}$ with $c_1\le \cdots \le c_m$ and $T_1 = \{c_1, \dots, c_{k-1}, c_{k+1}, \dots, c_m\}$. Let $\phi_1(c)$ be the regularized virtual cost function (Definition~\ref{def:regularized_virtual_costs}) of the discrete uniform distribution over $T_1$ and let $\phi_2(c)$ be the regularized virtual cost function of the discrete uniform distribution over $T_2$. Then it holds that for all number $K\ge 0$, 
$$
 \sum_{c \in T_1} \min\left\{ \phi_1(c),  K  \sqrt{\phi_1(c)}\right\} \le \sum_{c \in T_2} \min\left\{2 \phi_2(c), K \sqrt{2 \phi_2(c)}\right\} 
$$
$$
\sum_{c \in T_2} \min\left\{\phi_2(c), K \sqrt{\phi_2(c)}\right\}  \le  \sum_{c \in T_1} \min\left\{2 \phi_1(c),  2 K \sqrt{ \phi_1(c)}\right\}.
$$
\end{lemma}
\begin{proof}
By Lemma~\ref{lem:unbiased_adj_lem1}, for $i\neq k$, $\phi_1(c_i) \le 2 \phi_2(c_i)$. Thus the first inequality  holds as follows,
\begin{align*}
    \sum_{c \in T_1} \min\left\{ \phi_1(c),  K  \sqrt{\phi_1(c)}\right\} \le \sum_{c \in T_2} \min\left\{2\phi_2(c), K \sqrt{2\phi_2(c)}\right\}.
\end{align*}
By Lemma~\ref{lem:unbiased_adj_lem1}, for $i\neq k$, $\phi_2(c_i) \le 2 \phi_1(c_i)$. In addition, we have 
\begin{align*}
\psi_2(c_k) + \psi_2(c_{k+1}) = & (k+1) c_{k+1} - (k-1) c_{k-1} \\
=& c_{k+1} + k \cdot c_{k+1} - (k-1) c_{k-1}\\
\le& 2(k \cdot c_{k+1} - (k-1) c_{k-1}) \\
= & 2 \psi_1(c_{k+1}).
\end{align*}
By our definition of regularized virtual cost function, it always holds that
\begin{align*}
    \sum_{c \in T_2} \min\left\{ \phi_2(c),  K  \sqrt{\phi_2(c)}\right\}  \le  \sum_{c \in T_1} \min\left\{2 \phi_1(c), 2 K  \sqrt{ \phi_1(c)}\right\}.
\end{align*}
\end{proof}
The following lemma will be used in our proof of the main theorem.
\begin{lemma} \label{lem:unbiased_adj}
Let $T_1$ and $T_2$ be two costs sets that have $T_2 = \{c_1, \dots, c_k, \dots, c_m\}$ with $c_1\le \cdots \le c_m$ and $T_1 = \{c_1, \dots, c_{k-1}, c_{k+1}, \dots, c_m\}$. Let $B$ be an arbitrary non-negative number. We use $M^{OPT}(T, B)$ to represent the optimal unbiased estimator mechanism defined in~\eqref{prog:unbiased_roundi}, when the cost set is $T$ and the budget is $B$. Then define 
\begin{itemize}
    \item $(A_1,P_1) = M^{OPT}(T_1, B/2)$, i.e., $(A_1,P_1)$ is the optimal mechanism when the cost set is $T_1$ and the budget is $B/2$.
    \item $(A_2,P_2) = M^{OPT}(T_2, B)$, i.e., $(A_2,P_2)$ is the optimal mechanism when the cost set is $T_2$ and the budget is $B$.
    \item $(A_3, P_3) = M^{OPT}(T_2, B/4)$, i.e., $(A_3,P_3)$ is the optimal mechanism when the cost set is $T_2$ and the budget is $B/4$.
\end{itemize}  
Then we have
\begin{align*}
& A_1(c_i) \le A_2(c_i) \textrm{ for all } i\ge k+1, \quad A_1(c_{k+1})P_1(c_{k+1}) \le A_2(c_{k+1}) P_2(c_{k+1}),\\
& A_1(c_i) \ge A_3(c_i) \textrm{ for all } i\ge k+1, \quad A_1(c_{k+1})P_1(c_{k+1}) \ge A_3(c_{k+1}) P_3(c_{k+1}).
\end{align*}
\end{lemma}
\begin{proof}
Let $\phi_1(c)$ be the regularized virtual cost function (Definition~\ref{def:regularized_virtual_costs}) of the discrete uniform distribution over $T_1$ and let $\phi_2(c)$ be the regularized virtual cost function of the discrete uniform distribution over $T_2$. Then according to Lemma~\ref{lem:unbiased_chrz}, 
$$
A_i(c) = \min\left\{ 1, \frac{\lambda_i}{\sqrt{\phi_i(c)}}\right\}
$$
and the value of  $\lambda_1$, $\lambda_2$ and $\lambda_3$ should satisfy
$$
\sum_{c \in T_1} \min\left\{\phi_1(c), \lambda_1  \sqrt{\phi_1(c)}\right\} = B/2,
$$
$$
 \sum_{c \in T_2} \min\left\{\phi_2(c), \lambda_2  \sqrt{\phi_2(c)}\right\} = B,
$$
$$
 \sum_{c \in T_2} \min\left\{\phi_2(c), \lambda_3  \sqrt{\phi_2(c)}\right\} = B/4,
$$
We first compare $A_1$ and $A_3$. By the first inequality of Lemma~\ref{lem:unbiased_adj_lem2}, 
$$
 \sum_{c \in T_1} \min\left\{ \phi_1(c),  \sqrt{2} \cdot \lambda_3 \sqrt{\phi_1(c)}\right\} \le \sum_{c \in T_2} \min\left\{2 \phi_2(c), \sqrt{2} \cdot \lambda_3 \sqrt{2 \phi_2(c)}\right\}  \le 2\cdot B/4 = B/2.
$$
Because the value of $\lambda_1$ is optimal when the cost set is $T_1$ and the budget is $B/2$, it should satisfy that  $\lambda_1 \ge \sqrt{2} \cdot \lambda_3$. And since $\phi_1(c_i) \le 2 \phi_2(c_i)$ for all $i\ge k+1$ according to Lemma~\ref{lem:unbiased_adj_lem1}, we get
$$
A_3(c_i) = \min \left\{1, \frac{\lambda_3}{\sqrt{ \phi_2(c_i) }} \right \} \le \min \left\{ 1, \frac{\lambda_1/\sqrt{2}}{\sqrt{\phi_1(c_i)/2}}\right\} = A_1(c_i)
$$
for all $i\ge k+1$. 
Similarly we have
$$
A_1(c_i) = \min \left\{1, \frac{\lambda_1}{\sqrt{ \phi_1(c_i) }} \right \} \le \min \left\{ 1, \frac{\lambda_2/\sqrt{2}}{\sqrt{\phi_2(c_i)/2}}\right\} = A_2(c_i)
$$
for all $i\ge k+1$. 

The expected payment of the two allocation rules can be compared according to the definition in Lemma~\ref{myerson},
\begin{align*}
A_1(c_{k+1}) P_1(c_{k+1}) = & A_1(c_{k+1})c_{k+1} +  \int_{c_{k+1}}^{\overline{C}} A_1(v) \, dv \\
\le & A_2(c_{k+1})c_{k+1} +  \int_{c_{k+1}}^{\overline{C}} A_2(v) \, dv = A_2(c_{k+1}) P_2(c_{k+1}).
\end{align*}
Here we use integrals to equivalently represent the payment rule as in Lemma~\ref{lem:true_myerson}.
\end{proof}

\subsection{Proof of Theorem~\ref{thm:unbiased_online}} \label{app:unbiased_online}
We prove that when we use the following allocation rule 
\begin{align} \label{prog:app_unbiased}
A^i= &\arg \min_A \max_{\vec{z} \in [0,1]^i} \sum_{k = 1}^i \frac{z_k^2}{A(c_{(k)})}\\ 
& \quad \textrm{s.t. } \sum_{k=1}^i A(c_{(k)}) \psi^i(c_{(k)}) \le \frac{B}{4\sqrt{n/i}}\notag\\
& \quad \qquad A(c_{(k)})\ge A(c_{(k+1)}), \forall k \notag\\
& \quad  \qquad 0 \le A(c_{(k)}) \le 1, \quad \forall k \notag
\end{align}
at round $i$ and output Horvitz-Thompson Estimator  $S = \frac{1}{n} \sum_{i=1}^n y_i =  \frac{1}{n} \sum_{i=1}^n \frac{\widehat{x}_i}{A^i(c_i)}$ at last,  the mechanism (1) is truthful and individually rational; (2) the expected total spending is no more than $B$;  (3) for any cost distribution $\{c_1, \dots, c_n\}$, the worst-case variance of the final estimator $S$ is no more than roughly $16$ times the benchmark defined in Definition~\ref{def:unbiased_benchmark}.  
The worst case variance of the benchmark  equals
\begin{align} \label{prog:unbiased_benchmark_app}
\Var^*(A^*) =  \frac{1}{(n+1)^2} \ & \min_A  \max_{\vec{z}\in [0,1]^n} \quad \sum_{i=1}^{n+1}  \frac{z_i^2}{A(c_i)} - \sum_{i=1}^{n+1} z_i^2 \\
\textrm{ s.t. }& \ \sum_{i=1}^{n+1} A(c_i) \psi(c_i) \le B \notag\\
& \   A(c_i)\ge A(c_{i+1}), \quad \forall 1\le i \le n \notag \\
& \  0\le A(c) \le 1, \quad \forall c \notag
\end{align}
where we let $c_{n+1} = \overline{C}$ for notation simplicity.

To prove budget feasibility and the performance of the mechanism, we construct an intermediate mechanism $(A',P')$ at each step $i$. This intermediate mechanism $(A',P')$ is basically the same as $(A^i, P^i)$, but is ``one-step-ahead''. Loosely speaking, $(A', P')$ is the optimal mechanism when the same amount of budget is assigned at round $i$, but knows an additional piece of information, the value of $c_i$, beforehand. We compare $(A',P')$ with $(A^i, P^i)$ based on the results in Section~\ref{app:unbiased_adj}. 

\paragraph{Notations:} Before the proof, we define some notations. Let $c_{(1)}  \le  \cdots \le c_{(n)}$ be the ordered costs from smallest to largest. We use $\E[X\vert \{c_1, \dots, c_k\}]$ to represent the conditional expectation of random variable $X$ conditioning on the event that  the set of the first $k$ data holders' costs is $\{c_1, \dots, c_k\}$. We use  $\E[X\vert c_1, \dots, c_k]$ to represent the conditional expectation of random variable $X$ given that the sequence of the first $k$ data holders' costs is $c_1, \dots, c_k$, i.e., the first data holder has cost $c_1$, the second has $c_2$ and so on so forth. Notations are the same for the conditional variance $\Var(X\vert \{c_1, \dots, c_k\})$ and $\Var(X\vert c_1, \dots, c_k)$.
Unless otherwise stated, the randomness is taken over the random arriving order and the internal randomness of the mechanism.

\subsubsection{Truthfulness and Individual Rationality} 
It is easy to see that the extended allocation rule and payment rule of a truthful mechanism is still truthful. The payment rule also guarantees individual rationality as $P^i(c)$ is always greater or equal to $c$ by definition.

\subsubsection{Expected Budget Feasibility} 
Suppose the costs of the population is $\{c_1, \dots, c_n\}$. The total expected spending of the mechanism is $\sum_{i=1}^n \E[A^i(c_i)\cdot P^i(c_i)]$. Consider a fixed round $i$. Let $S_i$ be the set of first $i-1$ agents' costs and define $T_i = S_i \cup \{\overline{C}\}$. Similarly let $S_{i+1}$ be the set of first $i$ agents' costs and define $T_{i+1} = S_{i+1} \cup \{\overline{C}\}$. Then conditioning on the event that the set of first $i-1$ costs is $S_i$, the allocation rule $A^i$ can be uniquely decided, which is the solution of 
\begin{align*}
A^i = \min_A  &\sum_{c \in T_i} \frac{1}{A(c)}\\
\textrm{ s.t. }& \sum_{c \in T_i}  A(c) \psi^i(c) \le \frac{B}{4\sqrt{n/i}}\\
& A \textrm{ is monotone non-increasing} \\
& 0 \le A(c) \le 1, \quad \forall c \in T_i\notag
\end{align*}
Notice that $c_i \notin T_i$, so $A(c_i)$ is not a decision variable of the mathematical program defined above. The value of $A(c_i)$ is decided by Definition~\ref{def:extend_mech}. Let $c_{(1)}, \dots, c_{(i)}$ be the first $i$ costs in non-decreasing order and let $c_{(i+1)} = \overline{C}$. Suppose $c_i$ is the $k$-th smallest cost, i.e., $c_i = c_{(k)}$.  
Then according to Definition~\ref{def:extend_mech}, 
$$
 A^i(c_i)=A^i(c_{(k)}) = A^i(c_{(k+1)})
 $$
 where $c_{(k+1)}$ belongs to $T_i$. Now consider the following allocation rule $A'$
 \begin{align*}
A'= \min_A  &\sum_{c \in T_{i+1}} \frac{1}{A(c)}\\
\textrm{ s.t. }& \sum_{c \in T_{i+1}}  A(c) \psi^{i+1}(c) \le \frac{B}{2\sqrt{n/i}}\\
& A \textrm{ is monotone non-increasing} \\
& 0 \le A(c) \le 1, \quad \forall c \in T_{i+1}\notag
\end{align*}
Since $T_i,\,T_{i+1}$ only differ by one element $c_i$, and $A'$ uses twice the budget of $A^i$, according to Lemma~\ref{lem:unbiased_adj}, 
 \begin{align*}
 A^i(c_i) P^i(c_i) = &  A^i(c_{(k+1)})  P^i(c_{(k+1)})  \le  A'(c_{(k+1)})  P'(c_{(k+1)}).
 \end{align*}
 Now assume the set of  the first $i$ costs is $S_{i+1}$. When the data holders come in random order, $c_i$ is a random element chosen from $S_{i+1}$. Therefore $c_i$'s  rank $k$ should be uniformly distributed over $\{1, \dots, i\}$, and thus
\begin{align*}
\E[A^i(c_i) \cdot P^i(c_i) | S_{i+1}] \le & \E[ A'(c_{(k+1)})  P'(c_{(k+1)}) | S_{i+1}]\\
= & \frac{1}{i} \cdot \sum_{j=1}^i A'(c_{(j+1)})  P'(c_{(j+1)})\\
\le & \frac{1}{2\sqrt{n\cdot i}}  \cdot B.
\end{align*}
Therefore the total spending of the mechanism is bounded as 
\begin{align*}
\sum_{i=1}^n \E[A^i(c_i)\cdot P^i(c_i)] & \le  \sum_{i=1}^n \frac{B}{2\sqrt{ n \cdot i}} \le B,
\end{align*}
since $ \sum_{i=1}^n \frac{1}{\sqrt{i}} \le 2\sqrt{n}$.

\subsubsection{Competitive Analysis} \label{prg:compareA}
We first show that the variance of the final estimator $S = \frac{1}{n} \sum_{i=1}^n y_i$  can be upper bounded by the sum of the variances ``occur at each round''.
\begin{lemma} \label{lem:unbiased_lem1}
For a population with costs $\{c_1, \dots, c_n\}$ and any cost-data distribution $\mathcal{D}$ that is consistent with the costs, the variance of the output estimator of Mechanism~\ref{alg:online_unbiased} can be upper bounded as
$$
\Var(S) \le \frac{1}{n^2} \sum_{i=1}^n \E \left[y_i^2\right],
$$
where $y_i = \frac{\widehat{x}_i}{A^i(c_i)}$ is the re-weighted data.
\end{lemma}
\begin{proof}
According to the law of total variance,
\begin{align*}
\Var(S) & = \mathbb{E} \left[\Var(S\vert c_1, \dots, c_n) \right] + \Var\left( \mathbb{E}[S \vert c_1,\dots, c_n]\right),
\end{align*}
Since the estimator is always unbiased, conditional mean $\E[S \vert c_1,\dots, c_n]$ always equals to the mean of $n$ data points $\sum_{i=1}^n z_i/n$ for any order $c_1, \dots, c_n$. Because the variance of a constant is zero, the second term $\Var\left( \mathbb{E}[S \vert c_1,\dots, c_n]\right)$ equals $0$.
Therefore $$\Var(S)  = \mathbb{E} \left[\Var(S\vert c_1, \dots, c_n) \right].$$ 
Notice that conditioning on the cost sequence $c_1, \dots, c_n$, the allocation rules used at each round $A^1, \dots, A^n$ will be fixed. Thus $y_1, \dots, y_n$ become independent. So we have 
\begin{align*}
 \E\left[\Var(S \vert c_1, \dots, c_n) \right]  &= \E\left[\frac{1}{n^2}\cdot \Var\left(\sum_{i=1}^n y_i \,\Big \vert \, c_1, \dots, c_n \right) \right] \\
&= \frac{1}{n^2} \cdot \E \sum_{i=1}^n \Var\left( y_i \, \Big \vert \, c_1, \dots, c_n \right) \\
&\le \frac{1}{n^2} \cdot \E \sum_{i=1}^n \E \left[y_i^2 \, \vert \, c_1, \dots, c_n \right]\\
& = \frac{1}{n^2} \sum_{i=1}^n \E[y_i^2]
\end{align*}
\end{proof}

Let $y_i = \frac{\widehat{x}_i}{A^i(c_i)}$ be the re-weighted observed data at round $i$. The following lemma compares the expectation $\E[y_i^2]$ with the worst-case variance of $A^*$.
\begin{lemma}
For any distribution $\mathcal{D}$ that is consistent with the cost distribution $\{c_1, \dots, c_n\}$, when the data holders come in random order,   Mechanism~\ref{alg:online_unbiased} has 
$$
\E[y_i^2] \le \frac{8 \cdot \sqrt{n}}{\sqrt{i}} \cdot \left(n \Var^*(A^*)+1 +  \frac{1}{i} \cdot \frac{1}{A^*(\overline{C})}\right).
$$
\end{lemma}
\begin{proof} We fix a time step $i$. Let $S_i$ be the set of first $i-1$ agents' costs and define $T_i = S_i \cup \{\overline{C}\}$. And similarly let $S_{i+1}$ be the set of first $i$ agents' costs and define $T_{i+1} = S_{i+1} \cup \{\overline{C}\}$.  We first consider a fixed $S_{i+1}$ and compare $\E[y_i^2|S_{i+1}]$ with $ \Var^*(A^*)$, where the randomness of $\E[y_i^2|S_{i+1}]$ is taken over the random arriving order and the internal randomness of the mechanism. We will define an intermediate allocation rule $A'$ and compare it with both $A^i$ and $A^*$. 

  Conditioning on the event that the set of first $i-1$ costs is $S_i$, the allocation rule $A^i$ can be uniquely decided, which is the solution of 
\begin{align*}
A^i = \min_A  &\sum_{c \in T_i} \frac{1}{A(c)}\\
\textrm{ s.t. }& \sum_{c \in T_i}  A(c) \psi^i(c) \le \frac{B}{4\sqrt{n/i}}\\
& A \textrm{ is monotone non-increasing} \\
& 0 \le A(c) \le 1, \quad \forall c \in T_i\notag
\end{align*}
Notice that agent $i$'s cost $c_i \notin T_i$, so $A(c_i)$ is not a decision variable of the mathematical program defined above. The value of $A(c_i)$ is decided by Definition~\ref{def:extend_mech}. Let $c_{(1)}, \dots, c_{(i)}$ be the first $i$ costs in non-decreasing order and let $c_{(i+1)} = \overline{C}$. Suppose $c_i$ is the $k$-th smallest cost, i.e., $c_i = c_{(k)}$.  
Then
$$
 A^i(c_i)=A^i(c_{(k)}) = A^i(c_{(k+1)})
 $$
 where $c_{(k+1)}$ belongs to $T_i$. 
 Now consider the following allocation rule $A'$
 \begin{align*}
A'= \min_A  &\sum_{c \in T_{i+1}} \frac{1}{A(c)}\\
\textrm{ s.t. }& \sum_{c \in T_{i+1}}  A(c) \psi^{i+1}(c) \le \frac{B}{8\sqrt{n/i}}\\
& A \textrm{ is monotone non-increasing} \\
& 0 \le A(c) \le 1, \quad \forall c \in T_{i+1}\notag
\end{align*}

\paragraph{Comparing $A'$ with $A^*$:} 
 We first compare $A'$ with $A^*$ by proving that $\frac{A^*}{8\sqrt{n/i}}$ is a feasible solution of the mathematical program that defines $A'$. Let $\psi(c)$ be the virtual cost function when the underlying cost distribution is the uniform distribution over $T = \{c_1, \dots, c_{n}, \overline{C}\}$. Then according to definition, the benchmark $A^*$ satisfies the constraint
$ \sum_{c \in T}  A^*(c) \psi(c) \le B$.
Combine this with Lemma~\ref{lem:VC_subset},
$$
\sum_{c \in T_{i+1}} \frac{A^*}{8\sqrt{n/i}} \cdot  \psi^{i+1}(c) \le \sum_{c\in T}  \frac{A^*}{8\sqrt{n/i}} \cdot \psi(c) \le \frac{B}{8\sqrt{n/i}}. 
$$
Therefore $\frac{A^*}{8\sqrt{n/i}}$ is a feasible solution of the mathematical program that defines $A'$. Then since $A'$ is the optimal solution, we have
\begin{equation} \label{eqn:unbiased_eqn1}
\sum_{c \in T_{i+1}} \frac{1}{A'(c)} \le  \sum_{c \in T_{i+1}} \frac{8\sqrt{n/i}}{A^*(c)}.
\end{equation}

\paragraph{Comparing $A'$ with $A^i$:}
 Now we compare $A^i(c_i) = A^i(c_{(k+1)})$ with $A'(c_{(k+1)})$. 
Since $T_i$ and $T_{i+1}$ only differs in one element $c_i = c_{(k)}$ and $A'$ uses half of the budget of $A^i$, according to Lemma~\ref{lem:unbiased_adj}, 
\begin{equation} \label{eqn:unbiased_eqn2}
 A^i(c_{(k+1)}) \ge A'(c_{(k+1)}). 
\end{equation}
for any set of costs $S_{i+1}$ of first $i$ agents.

Conditioning on the event that the set of  the first $i$ costs is $S_{i+1}$, when the data holders come in random order, $c_i$ is a random element chosen from $S_{i+1} = \{c_{(1)}, \dots, c_{(i)}\}$. Therefore $c_i$'s  rank $k$ should be uniformly distributed over $\{1, \dots, i\}$.  Together with~\eqref{eqn:unbiased_eqn2}, we have
\begin{align*}
 \E[y_i^2 | S_{i+1}] \le  \E\left[\frac{1}{A^i(c_{(k+1)})} \Big | S_{i+1}\right]  \le  \E\left[\frac{ 1}{A'(c_{(k+1)})} \Big | S_{i+1}\right] = \frac{1}{i} \sum_{k=1}^{i} \frac{1}{A'(c_{(k+1)})}.
\end{align*}
Then by equation~\eqref{eqn:unbiased_eqn1},
\begin{align*}
 \E[y_i^2 | S_{i+1}] \le & \frac{1}{i} \sum_{k=1}^{i} \frac{1}{A'(c_{(k+1)})} 
 \le  \frac{1 }{i}  \sum_{c\in T_{i+1}} \frac{1}{A'(c)} 
 \le  \frac{8\cdot \sqrt{n/i} }{i}  \sum_{c\in T_{i+1}} \frac{1}{A^*(c)}.
\end{align*}

Now we are ready to compute $\E[y_i^2]$ by averaging $ \E[y_i^2 | S_{i+1}] $ over random subset $S_{i+1}$ (when the agents come in random order, $S_{i+1}$ is a random subset of $\{c_1, \dots, c_n\}$ with length $i$). 
\begin{align*}
\E[y_i^2] =& \E_{S_{i+1}} \E[y_i^2 | S_{i+1}] \\
\le &  \frac{8 \cdot \sqrt{n/i} }{i} \cdot \E_{S_{i+1}}\left[ \sum_{c\in T_{i+1}} \frac{1}{A^*(c)} \right]\\
= &  \frac{8 \cdot \sqrt{n/i} }{i} \cdot \E_{S_{i+1}}\left[ \sum_{c\in S_{i+1}} \frac{1}{A^*(c)} + \frac{1}{A^*(\overline{C})}\right]\\
= &  8 \cdot \sqrt{n/i}  \left( \frac{1}{n} \cdot \sum_{i=1}^n \frac{1}{A^*(c_i)} + \frac{1 }{i} \cdot \frac{1}{A^*(\overline{C})}\right)\\
\le& 8 \cdot \sqrt{n/i}  \left(\frac{(n+1)^2}{n} \cdot \Var^*(A^*) + 1 +  \frac{1}{i} \cdot \frac{1}{A^*(\overline{C})}\right).
\end{align*}

\end{proof}
Combining the two lemmas, for any joint distribution of cost and data, the variance of the output estimator of our mechanism satisfies
\begin{align*}
\Var(S) \le &  \frac{1}{n^2} \sum_{i=1}^n \E[y_i^2] \\
\le & \frac{1}{n^2} \sum_{i=1}^n \frac{8\cdot \sqrt{n}}{\sqrt{i}}  \left(\frac{(n+1)^2}{n} \cdot \Var^*(A^*) + 1 +  \frac{1}{i} \cdot \frac{1}{A^*(\overline{C})}\right)\\
=& \frac{ 8 \sqrt{n}}{n^2} \left( \frac{(n+1)^2}{n} \cdot \Var^*(A^*)+ 1\right) \cdot \sum_{i=1}^n \frac{1}{\sqrt{i}} + \frac{8 \cdot \sqrt{n}}{n^2} \cdot \frac{1}{A^*(\overline{C})} \cdot \sum \frac{1}{i\sqrt{i}}\\
\le & 16 \cdot \left(\frac{(n+1)^2}{n^2} \Var^*(A^*)+\frac{1}{n}+ \frac{\sqrt{n}}{n^2} \cdot \frac{1}{A^*(\overline{C})}\right) 
\end{align*}
since $\sum_{i=1}^n \frac{1}{\sqrt{i}} \le 2\sqrt{n}$ and $\sum_{i=1}^n \frac{1}{\sqrt{i}\cdot i} \le 2$.


This completes the proof of Theorem~\ref{thm:unbiased_online}.

\section{Proofs for Optimal Confidence Interval}
We then give the proofs for the optimal online confidence interval estimator. 
In Section~\ref{app:CI_roundi}, we give the characterization of the optimal solution $(A^*, U^*, P^*)$ and show how to compute it. In Section~\ref{app:CI_adj}, we prove some lemmas that will be used to compare our mechanism with the benchmark in Section~\ref{app:CI_main}.
\subsection{Proof of Theorem~\ref{thm:CI_roundi}} \label{app:CI_roundi}

We prove that the optimal solution of
\begin{align} \label{prog:biased_roundi_app}
A^*, \ U^* = &\arg\min_{A, U}  \quad \beta^2 \cdot \frac{1}{n}  \sum_{j = 1}^n \frac{1-U_j}{A(c_{(j)})} + \left(\frac{1}{n} \sum_{j=1}^n U_j   \right)^2 \\
& \qquad \qquad \textrm{s.t.} \  \sum_{j=1}^n (1-U_j) \cdot A(c_{(j)}) \psi(c_{(j)}) \le B \notag\\
& \qquad  \qquad \qquad  (1-U_j) \cdot A(c_{(j)})  \textrm{ is monotone non-increasing in $j$} \notag\\
& \qquad  \qquad \qquad 0\le A(c_{(j)}) \le 1, \quad 0 \le U_j \le 1,  \quad \forall j \notag\end{align}
is as follows
\begin{align*}
&U^*_j = \left\{ \begin{array}{ll}
				 0, & \textrm{ if $\phi(c_{(j)}) < H$} \\
				 p \in (0,1], & \textrm{ if } \phi(c_{(j)}) = H \\
				 1, & \textrm{ if $\phi(c_{(j)})>H$}
			     \end{array}
			\right.\\
&A^*(c_{(j)}) =  \min \left\{1, \frac{\lambda}{\sqrt{\phi(c_{(j)})}} \right\}
\end{align*}
$\lambda$ is chosen such that $ \sum_{j=1}^n (1-U_j) \cdot A(c_{(j)}) \psi(c_{(j)}) = B$ and the optimal value of $\lambda$ and $H$ can be found using binary search.

The following lemma proves that the optimal solution should have $U^*$ as described above for some $H$, and $A^*$ should be non-increasing. 
\begin{lemma} \label{lem:biased_chara_lem1}
Let $U^*$ and  $A^*$ be the optimal solution of~\eqref{prog:biased_roundi_app}, then $U^*_j$ should be non-decreasing in $j$ and $A^*(c_{(j)})$ should be non-increasing in $j$. In addition, if $A^*(c_{(j)}) \neq A^*(c_{(j+1)})$, then at least one of $U^*_j$ and $U^*_{j+1}$ should be equal to $0$ or $1$. 
\end{lemma}
\begin{proof}
We first prove the monotonicity of $U^*$ and $A^*$. Let $U'$ and $A'$ be a feasible solution.  Suppose $U'$ decreases at some position $j$, i.e., $U'_j >U'_{j+1}$. We prove that that $U'$ cannot be optimal because it is not the optimal solution of the following optimization problem (for simplicity we write $A_j = A(c_{(j)})$).
\begin{align} \label{prog:gotosleep}
& \min_{A_j, A_{j+1}, U_j, U_{j+1}} \quad  \frac{(1-U_j)}{A_j} + \frac{(1-U_{j+1})}{A_{j+1}}\\
&  \textrm{ s.t. } \quad (1-U_j) A_j = (1-U'_j) A'(c_{(j)}) \notag \\
& \qquad \quad (1-U_{j+1}) A_{j+1} = (1-U'_{j+1}) A'(c_{(j+1)})  \notag\\
& \qquad \quad (1-U_j) + (1-U_{j+1}) = (1-U'_j) + (1-U'_{j+1}) \notag\\
& \qquad \quad 0 \le U_j, U_{j+1}, A_j, A_{j+1} \le 1. \notag
\end{align}
It is easy to see that a feasible solution of the above problem is also a feasible solution of~\eqref{prog:biased_roundi_app}. And a strictly better solution than $A', U'$ in the above problem will give a strictly better solution of~\eqref{prog:biased_roundi_app}. 
If we write the right-hand constant of the three equalities as $D_1, D_2, D_3$ and let $O_j = 1-U_j$, then the objective function can be represented as a single variable function of $O_j$,
$$
 \frac{(1-U_j)}{A_j} + \frac{(1-U_{j+1})}{A_{j+1}} = \frac{O_j^2}{D_1} + \frac{(D_3 - O_j)^2}{D_2}.
$$
Its derivative with respect to $O_j$ equals $ \frac{2 O_j}{D_1} - \frac{2(D_3 - O_j)}{D_2}$.
The value of this derivative at point $O_j = 1-U'_j$ is 
$$
2\left(\frac{ (1-U'_j)}{(1-U'_j) A'(c_{(j)})} - \frac{(1-U'_{j+1})}{(1-U'_{j+1}) A'(c_{(j+1)}) } \right).
$$
Since $A', U'$ is a feasible solution, it should satisfy the constraint $(1-U'_j) A'(c_{(j)}) \ge (1-U'_{j+1}) A'(c_{(j+1)})$. Therefore if we have $1\ge U'_j >U'_{j+1}\ge 0$, this derivative will be negative, which means that the objective value can be decreased by slightly decrease $U_j'$ (it will still be feasible because $U'_j > U'_{j+1}\ge 0$). Therefore the optimal $U^*$ must be monotone non-decreasing. The monotonicity of $A^*$ can be proved by the same arguments, that is, to assume that $A'_j<A'_{j+1}$ for some $j$ and show that it cannot be the optimal solution of the above optimization problem.

The same approach can be used to prove that the second part.  Let $U'$ and $A'$ be an optimal solution with $A'(c_{(j)}) \neq A'(c_{(j+1)})$ and $U'_j, \ U'_{j+1} \in (0,1)$ for some $j$. We again consider the optimization~\eqref{prog:gotosleep}. Again the value of the derivative at point $O_j = 1-U'_j$ is 
$$
2\left(\frac{ 1}{A'(c_{(j)})} - \frac{1}{A'(c_{(j+1)}) } \right).
$$
Since an optimal allocation rule must be non-increasing, the value of this derivative is negative. By the same argument in the first part, $U'$ cannot be optimal, which is contradictory. 
\end{proof}
Therefore the optimal $U^*_j\in (0,1)$ only for $j$'s that have the same $A^*(c_{(j)})$. Based on this, we further prove the following lemma.
\begin{lemma}
Let $l$ be the smallest number that has  $U^*_l \in (0,1)$ and let $r$ be the largest number that has $U^*_r \in (0,1)$. Then $\phi(c_{(l)}) = \phi(c_{(r)})$. 
\end{lemma}
\begin{proof}
Suppose to the contrary,  $\phi(c_{(l)}) < \phi(c_{(r)})$. We construct another feasible solution that is strictly better than $A^*, U^*$. Define $I_l = \{k: \phi(c_{(k)}) = \phi(c_{(l)})\}$ and $I_{r} = \{ k: \phi(c_{(k)}) = \phi(c_{(r)})\}$. $I_l, I_r$ should both contain consecutive numbers.  Let $R_l$ be the right end point of $I_l$ and let $L_r$ be the left end point of $I_r$.  First by definition of regularized virtual costs,  $\phi(c_{(r)}) = \phi(c_{(L_r)}) =  \min_{k: k\ge L_r} Avg(L_r, k)$, where $Avg$ is the average function of $\psi$. Then 
\begin{align*}
Avg(L_r, r) \ge \min_{k: k\ge L_r} Avg(L_r, k) = \phi(c_{(r)}).
\end{align*}
Meanwhile according to  Claim~\ref{unbiased_chara_claim2} in the proof of  Lemma~\ref{lem:unbiased_chara_lem1},  we should have 
\begin{align*}
Avg(l, R_l) \le Avg(I_l) = \phi(c_{(l)}) < \phi(c_{(r)}).
\end{align*}
 In addition, if $R_l+1 \le L_r-1$, then $\phi(c_{(j)}) < \phi(c_{(r)})$ for all $j \in [R_l+1, L_r-1]$, which means $Avg(R_l+1, L_r-1)<\phi(c_{(r)})$. Therefore we should have
\begin{align} \label{eqn:biased_chara_lem1_eq1}
Avg(l,L_r-1) < \phi(c_{(r)}) \le Avg(L_r, r).
\end{align}
Therefore we can construct a better solution by uniformly decreasing $U^*_l, \dots, U^*_{L_r-1}$ and uniformly increasing $U^*_{L_r}, \dots, U^*_r$ by a small amount, so that the total increase is equal to the total decrease. The value of the objective function $\beta^2 \cdot \frac{1}{n}  \sum_{k = 1}^n \frac{1-U^*_k}{A^*(c_{(k)})} + \left(\frac{1}{n} \sum_{k=1}^n U^*_k \right)^2$ remains the same, because (1) $\sum_{k=1}^n U^*_k$ does not change and (2) the value of $A^*(c_{(k)})$ is the same for all $k\in [l,r]$ according to Lemma~\ref{lem:biased_chara_lem1}.  But the total  expected spending $ \sum_{k=1}^n (1-U^*_k) \cdot A^*(c_{(k)}) \psi(c_{(k)})$ strictly decreases due to~\eqref{eqn:biased_chara_lem1_eq1}, which means that the objective value can be strictly improved. 
\end{proof}
So the lemma shows that $U^*_j \in (0,1)$ only for $j$'s that have the same $\phi(c_{(j)})$. We then prove the other direction: for $j$'s that have the same $\phi(c_{(j)})$, there exists an optimal solution that has the same $U^*_j$ for these $j$'s.
\begin{lemma} \label{lem:biased_chara_lem3}
There exists an optimal solution $(U^*, A^*)$ such that  for all $(j,k)$ that has $\phi(c_{(j)})= \phi(c_{(k)})$, $U^*_j = U^*_k$ and $A^*(c_{(j)}) = A^*(c_{(k)})$.
\end{lemma}
\begin{proof}
Define $P_H$ to be the set of $j$'s that have $\phi(c_{(j)}) = H$. Then the elements in $P_H$ must be consecutive. Suppose $P_H = \{ l, \dots, r\}$. 

Let $U^*, A^*$ be an optimal solution.  We show that the following solution $U', A'$ is also feasible and has objective value no larger than that of $U^*, A^*$:
\begin{align*}
&U'_j = \frac{1}{r-l+1} \sum_{k=l}^r  U^*_k, \quad \textrm{for all } j\in P_H\\
&A'(c_{(j)}) = \frac{\sum_{k=l}^r (1-U^*_k) A^*(c_{(k)}) }{\sum_{k=l}^r 1-U^*_k}, \quad \textrm{ for all } j \in P_H.
\end{align*}
We first prove that $U', A'$ is feasible. The monotonicity constraint and the range constraint are trivially satisfied.  For the budget constraint on $\sum_{j=1}^n (1-U_j) \cdot A(c_{(j)}) \psi(c_{(j)})$, first observe that the sum $\sum_{k=l}^r (1- U'_k)A'(c_{(k)})$ remains the same as $\sum_{k=l}^r (1-U^*_k) A^*(c_{(k)})$,
\begin{align*}
\sum_{k=l}^r (1- U'_k)A'(c_{(k)}) = &\left(\sum_{k=l}^r 1- U'_k \right) A'(c_{(l)})\\
= & \left(\sum_{k=l}^r 1- U^*_k\right)\frac{\sum_{k=l}^r (1-U^*_k) A^*(c_{(k)}) }{\sum_{k=l}^r 1-U^*_k}\\
= & \sum_{k=l}^r (1-U^*_k) A^*(c_{(k)}).
\end{align*}
By definition of regularized virtual costs,  $Avg(l,r) = \phi(c_{(l)}) =  \min_{k: k\ge l} Avg(l, k)$, where $Avg$ is the average function of $\psi$. Thus
\begin{align*}
Avg(l, r) \le Avg(l, k) \textrm{ for all } k\in[l,r].
\end{align*}
And since $(1- U^*(c))A^*(c)$ must be non-increasing, by shifting $(1- U^*(c))A^*(c)$ to be constant within $P_H$, the total expected spending must only decrease,
$$
\sum_{k=l}^r (1-U'_k) \cdot A'(c_{(k)}) \psi(c_{(k)}) \le \sum_{k=l}^r (1-U^*_j) \cdot A^*(c_{(k)}) \psi(c_{(k)}).
$$ 

Then we prove the optimality of $U', A'$. Because $\sum_{k=l}^r U'_k = \sum_{k=l}^r U^*_k$, the second term of the objective function remains the same. For the first term, since $f(x) = \frac{1}{x}$ is a convex function when $x>0$, by Jensen's inequality,
\begin{align*}
\sum_{k=l}^r \frac{1-U'_k}{A'(c_{(k)})} = & \left(\sum_{k=l}^r 1-U'_k\right)\frac{1}{A'(c_{(l)})} \\
= & \left(\sum_{k=l}^r 1- U^*_k\right) \frac{1}{\sum_{k=l}^r \frac{1-U^*_k}{\sum_{k=l}^r 1-U^*_k} A^*(c_{(k)}) }\\
\le  & \left(\sum_{k=l}^r 1- U^*_k\right) \sum_{k=l}^r \frac{1-U^*_k}{\sum_{k=l}^r 1-U^*_k} \frac{1}{ A^*(c_{(k)}) } \\
= & \sum_{k=l}^r (1-U^*_k) \frac{1}{ A^*(c_{(k)}) }.
\end{align*}
Therefore $U', A'$ is no worse than $U^*, A^*$, and is thus optimal.
\end{proof}

Combining the above three lemmas, there exists an optimal $U^*$ as defined in the theorem. Next, we characterize the optimal $A^*$.
\begin{lemma} \label{lem:biased_chara_lem2}
If we fixed an optimal solution of $U$: $U^*_1, \dots, U^*_j$, which should have $U^*_j= U^*_k$ if $\phi(c_{(j)}) = \phi(c_{(k)})$ according to Lemma~\ref{lem:biased_chara_lem3}, then the optimal solution of $A$ should be equal to
\begin{align*}
A^*(c_{(j)}) =   \min \left\{1, \frac{\lambda}{\sqrt{\phi(c_{(j)})}} \right\}, 
\end{align*}
where $\lambda$ is chosen such that the budget constraint is satisfied with equality $$ \sum_{j=1}^n (1-U^*_j) \cdot A^*(c_{(j)}) \psi(c_{(j)}) = B.$$
\end{lemma}
\begin{proof}
When $U$ is fixed as $U^*_1, \dots, U^*_j$,  then the optimal $A^*$ should be the solution of 
\begin{align*} 
A^* = &\arg\min_{A}  \quad  \sum_{j = 1}^n \frac{(1-U^*_j)}{A(c_{(j)})} \\
& \qquad  \textrm{s.t.} \quad \sum_{j=1}^n (1-U^*_j) \cdot A(c_{(j)}) \psi(c_{(j)}) \le B \notag\\
& \qquad \qquad (1-U^*_j) A(c_{(j)})  \textrm{ is monotone non-increasing in $j$} \notag\\
& \qquad \qquad  0\le A(c_{(j)}) \le 1,  \quad \forall j \notag
\end{align*}
The Lagrangian function is (for simplicity we write $A_j = A(c_{(j)})$)
\begin{align*}
L( A, \xi, \vec{\pi}, \vec{\eta}^0, \vec{\eta}^1) = &\sum_{j=1}^n \frac{1-U^*_j}{A_j} \\
 & + \xi \left( \sum_{j=1}^n (1-U^*_j) A_j \psi(c_{(j)}) - B\right) \\
 & + \sum_{j=1}^n \pi_j\left((1-U^*_{j+1})A_{j+1} - (1-U^*_j)A_j\right)\\
 & + \sum_{j=1}^n \eta^0_j A_j + \sum_{j=1}^n \eta^1_j (A_j - 1).
\end{align*}
We prove that the optimal primal variables are 
\begin{align*}
A^*_j = \min \left\{1, \frac{\lambda}{\sqrt{\phi(c_{(j)})}} \right\},
\end{align*}
where $\lambda$ is chosen such that the budget constraint is satisfied with equality $ \sum_{j=1}^n (1-U^*_j) \cdot A^*_j \psi(c_{(j)}) = B$, and the optimal dual variables are
\begin{align*}
&\xi = \frac{1}{\lambda^2}, \\
& \pi_j = \pi_{j-1} + \xi \left(\psi(c_{(j)}) - \phi(c_{(j)})\right),  \textrm{ (Here we assume $\pi_{0} = 0$.)}\\
& \eta_j^0 = 0 \textrm{ for all } j, \\
& \eta_j^1 = \left\{ \begin{array}{ll}
				(1-U^*_j) \left( 1-\xi \cdot \phi(c_{(j)}) \right), \textrm{ if } \xi \cdot \phi(c_{(j)}) < 1,\\
				0, \textrm{ otherwise.}
			\end{array} \right.
\end{align*}

\paragraph{Primal feasibility:} We first prove the primal feasibility.
\begin{enumerate}
\item By our definition of $A^*$, the budget constraint should be satisfied with equality.
\item By (1) in Lemma~\ref{lem:unbiased_chara_lem1}, $\phi(c)$ is a non-decreasing function of $c$. Then $A^*(c) = \min\left\{ 1, \frac{\lambda}{\sqrt{\phi(c)}} \right \}$ is a non-increasing function. And since $U^*_j$ is non-decreasing in $j$, $(1-U^*_j) A^*_j$ is non-increasing in $j$.
\item It is easy to verify that $0 \le A^*(c) \le 1$. 
\end{enumerate}
\paragraph{Dual feasibility:}  By (3) in Lemma~\ref{lem:unbiased_chara_lem1}, it is easy to verify that all of the dual variables greater or equal to $0$.  
\paragraph{Stationarity:} The partial derivative of $L( A, \xi, \vec{\pi}, \vec{\eta}^0, \vec{\eta}^1)$ with respect to each $A_j$ is
\begin{align*}
&\frac{\partial L( A, \xi, \vec{\eta}^0, \vec{\eta}^1) }{\partial A_j} \\
= & -   \frac{1- U_j}{A^2_j} +\xi  \cdot (1-U_j) \psi(c_{(j)}) + (1-U_{j})\pi_{j-1} - (1-U_j) \pi_j + \eta^0_j + \eta^1_j \\
= & -  (1- U_j) \max\left\{ 1, \  \xi \cdot \phi(c_{(j)}) \right\} +\xi  \cdot (1-U_j)  \cdot  \psi(c_{(j)}) + (1-U_{j})\pi_{j-1} - (1-U_j) \pi_j+ \eta^1_j\\
= & \ \xi(1- U_j)  \cdot (\psi(c_{(j)}) - \phi(c_{(j)}))  + (1-U_{j})\pi_{j-1} - (1-U_j) \pi_j \\
=& \ 0
\end{align*}
By the definition of $\pi_j$, we have the above quantity equal to $0$. 
\paragraph{Complementary slackness:} 
\begin{enumerate}
\item The budget constraint in the primal is satisfied with equality. 
\item For all $A^*_j \neq A^*_{j+1}$, we must have $\phi(c_{(j)}) \neq \phi(c_{(j+1)})$. Then by (3) in Lemma~\ref{lem:unbiased_chara_lem1}, $\pi_j = 0$. Since $U^*_j = U^*_{j+1}$ if $\phi(c_{(j)}) = \phi(c_{(j+1)})$ (which means $A^*_j = A^*_{j+1}$),  $\pi_j\left((1-U^*_{j+1})A^*_{j+1} - (1-U^*_j)A^*_j\right) = 0$ for all $j$.
\item We have $\eta_j^0 = 0$ for all $j$ and $\eta_j^1 = 0$ for all $A^*_j = \frac{1}{\sqrt{\xi \cdot \phi(c_{(j)})}} <1$.
\end{enumerate}
\end{proof}

Therefore there must exists an optimal solution $U^*, A^*$ as stated in the theorem. 

\subsection{Proof of Lemma~\ref{lem:baised_eq_point}} \label{app:biased_eq_point}

By Lemma~\ref{lem:unbiased_chrz},  the optimal value of $\lambda$ can be found using binary search after $U^*$ is decided.  Define $M = \sum_{k=1}^i U^*_k$. Then when $M$ is decided, we can immediately find a rule $U^*$ that has the characterization in the theorem, and the optimal $A^*$ can be computed as well. We now show how to find optimal $M$ that minimizes 
$$
\beta^2 \cdot \frac{1}{n}  \sum_{j = 1}^n \frac{1-U_j}{A(c_{(j)})} + \left(\frac{M}{n}  \right)^2
$$
using binary search within time $O(\log |C|)$. 

Let $g(M)$ be the the optimal value of the first term $\beta^2 \cdot \frac{1}{n}  \sum_{j = 1}^n \frac{1-U_j}{A(c_{(j)})}$ when $M$ is decided. Then the objective function is just $g(M) + \left(\frac{M}{n}  \right)^2$.  The second term $\left(\frac{M}{n}  \right)^2$ is a convex function of $M$. We prove that $g(M)$ is also a convex function of $M$.
Before the main proof, we first prove the following claim.
\begin{claim} \label{claim:biased_chara_cl1}
Let $U^*, A^*$ be an optimal solution as stated in the theorem. Then 
$$
 \sum_{j=1}^n (1-U^*_j) \cdot A^*(c_{(j)}) \phi(c_{(j)}) =  \sum_{j=1}^n (1-U^*_j) \cdot A^*(c_{(j)}) \psi(c_{(j)}) = B.
$$
\end{claim}
\begin{proof}
For each $j$, let $I_j$ be the set of all the $k$ that has $\phi(c_{(k)}) = \phi(c_{(j)})$. And let $\mathcal{I}$ be the set of all different $I_j$, then according to (2) in Lemma~\ref{lem:unbiased_chara_lem1} and the definition of $A^*$,
\begin{align*}
\sum_{j=1}^n (1-U^*_j) A^*_j \phi(c_{(j)}) = &\sum_{I \in \mathcal{I}} \sum_{j\in I} (1-U^*_j) A^*_j  \phi(c_{(j)}) \\
= & \sum_{I \in \mathcal{I}} (1-U^*_I)A^*_I \sum_{j\in I}  \phi(c_{(j)}) \\
= & \sum_{I \in \mathcal{I}} (1-U^*_I)A^*_I \sum_{j\in I}  \psi(c_{(j)})\\
=& \sum_{j=1}^n (1-U^*_j) A^*_j  \psi(c_{(j)}) =  B,
\end{align*}
where $A^*_I$ is the the same $A^*_j$ for all $j\in I$, $U^*_I$ is the same $U^*_j$ for all $j \in I$.
\end{proof}

\begin{lemma} 
The function $g(M)$ is a convex function of $M$. Furthermore, let $A_M$ be an optimal allocation rule  when $\sum_{j=1}^n U_j = M$ and let $c_{(r)}$ be the largest cost that is not ignored with probability $1$.  Then for non-integer $M$ that has $A_M(c_{(r)}) < 1$,
$$
\frac{\partial g(M)}{\partial M} =  - \beta^2 \cdot \frac{2}{n} \cdot  \frac{1}{A_M(c_{(r)})}.
$$
which is an non-decreasing function of $M$.
\end{lemma}

\begin{proof}
To simplify the notation, in the proof of this lemma, we use $c_1, \dots, c_n$ instead of $c_{(1)}, \dots, c_{(n)}$ to represent the costs in $C$ ordered from smallest to largest. To better illustrate how our objective value changes as $M$ increases, we consider an equivalent optimization problem that is defined as follows:
\begin{align} \label{CI_con_approx}
A^*, M^* = \arg\min_{A, M} &\quad \beta^2 \cdot  \frac{1}{n} \int_{0}^{n-M} \frac{1}{A(x)} \,dx + \left(\frac{M}{n} \right)^2\\
\textrm{s.t.} & \quad \int_{0}^{n-M} A(x) \cdot \phi(c_x) \, dx \le B  \notag \\
& \quad A \textrm{ is monotone non-increasing} \notag\\
&\quad 0\le A(x) \le 1, \quad \forall \, 0\le x \le n \notag
\end{align}
where
$$
c_x = c_{\lceil x \rceil}, \quad \forall \, 0\le x \le n,
$$
and thus $\phi(c_x) = \phi(c_{\lceil x \rceil})$. This optimization problem is equivalent to~\eqref{prog:biased_roundi_app} because 
\begin{enumerate}
    \item For any feasible solution $(A(x), M)$ of the above optimization problem, we can find a feasible solution of~\eqref{prog:biased_roundi_app} $(A(c_j), U_j)$ with better objective value by setting $A(c_j)$ equal to the mean of $A(x)$ for $x \in \{y: y \le n-M \textrm{ and } \phi(c_y) = \phi(c_j)\}$, and setting $U_j$ to be the mean of indicator function $\mathbbm{1}(x \ge n-M)$ for  $x \in \{y: \phi(c_y) = \phi(c_j)\}$. It can be verified that the budget constraint of~\eqref{prog:biased_roundi_app} is still satisfied according to (2) in Lemma~\ref{lem:unbiased_chara_lem1}. And the objective function only gets smaller because the mean of the reciprocals is no smaller than the reciprocal of the mean, as the reciprocal function is convex.
    \item For any optimal solution $A(c_j), U_j$ of~\eqref{prog:biased_roundi_app} that has the characterization as stated in the theorem, we can find a feasible solution $A(x), M$ of the above optimization problem with the same objective value by setting $M = \sum_j U_j$,  $A(x) = A(c_{\lceil x \rceil})$.
\end{enumerate}   
Therefore by the characterization in Lemma~\ref{lem:biased_chara_lem2}, if we fix a threshold $M$, there exists an optimal allocation rule, denoted by $A^*_M(x)$, which is equal to
$$ 
A^*_M(x) = \min\left\{ 1, \frac{\lambda_M}{\sqrt{\phi(c_x)}} \right\}
$$ 
for $0 \le x \le n-M$, and $\lambda_M$ is chosen such that the budget constraint is satisfied with equality. By Claim~\ref{claim:biased_chara_cl1}, the value of $\lambda_M$ satisfies
$$
 \sum_{j=1}^n (1-U^*_j) \cdot A^*(c_{(j)}) \phi(c_{(j)}) =  \sum_{j=1}^n (1-U^*_j) \cdot A^*(c_{(j)}) \psi(c_{(j)}) = B,
$$
or equivalently using $A^*_M(x)$ and $M$,
$$
\int_0^{n-M} A^*_M(x) \phi(c_x) \, dx = \int_0^{n-M} \min\left\{\phi(c_x), \lambda_M  \sqrt{\phi(c_x)}\right\} \,dx = B.
$$

Recall that 
$$
g(M) = \beta^2 \cdot \frac{1}{n} \int_{0}^{n-M} \frac{1}{A^*_M(x)} \,dx.
$$
Take derivative of the above function\footnote{The function $g(M)$ may not be differentiable at integer points, but is semi-differentiable, WLOG we use right derivatives at those points.}, by the Leibniz's rule for differentiation under the integral,
\begin{align} \label{eq:opt_A_derivative}
\frac{\partial g(M)}{\partial M} = &\frac{\partial }{\partial M} \left(\beta^2 \cdot \frac{1}{n} \int_{0}^{n-M} \frac{1}{A^*_M(x)} \,dx \right)\\
= & \beta^2 \cdot  \frac{1}{n} \left(- \frac{1}{A^*_M(n-M)} -  \int_{0}^{n-M}  \frac{1}{A^*_M(x)^2} \cdot \frac{\partial A^*_M(x)}{\partial M} \,dx \right). \notag
\end{align}
We denote the quantity inside the brackets by $Q$, then the derivative equals $\beta^2 \cdot \frac{1}{n} \cdot Q$.
Let $t$ be the maximum of $x$ that has $A^*_M(x)=1$. Then for all $x \in [0,t]$, we should have $\frac{\partial A^*_M(x)}{\partial M} = 0$, because $\lambda_M$ can only increase as $M$ increase, which means $A^*_M(x)$ can only increase for $x\in [0,t]$, but they have already reached the upper bound $1$, which means that $A^*_M(x)$ will stay the same as $M$ increase for all $x\in [0,t]$.
Then we can remove the part $0\le x < t$ of  $\int_{0}^{n-M} - \frac{1}{A^*_M(x)^2} \cdot \frac{\partial A^*_M(x)}{\partial M} \,dx $ because it integrates to zero. So the term inside the brackets is equal to
\begin{align} \label{eq:opt_A_d2}
Q = &  -  \frac{1}{A^*_M(n-M)} - \int_{t}^{n-M} \frac{1}{A^*_M(x)^2} \cdot \frac{\partial A^*_M(x)}{\partial M} \,dx \notag\\
=\, & \frac{\partial }{\partial M} \left(  \int_{t}^{n-M} \frac{1}{A^*_M(x)} \,dx  \right) \notag\\
=\, & \frac{\partial }{\partial M} \left( \int_{t}^{n-M} \frac{\sqrt{\phi(c_x)}}{\lambda_M} \,dx  \right).
\end{align}
To further simplify the derivative, we represent  $\lambda_M$ as a function of $\phi(c_x)$. Since $\lambda_M$ is chosen such that the budget constraint is satisfied with equality, i.e.,
$$
\int_0^{n-M} A^*_M(x) \phi(c_x) \, dx = B.
$$
And since $t$ is defined such that for $x\in[0,t]$, $A^*_M(x) = 1$; for $x \in (t, i-M]$,  $A^*_M(x) = \frac{\lambda_M}{\sqrt{\phi(c_x)}}$, the above equality is equivalent to 
$$
\int_0^{t} \phi(c_x) + \int_t^{n-M} \lambda_M \sqrt{\phi(c_x)} \, dx = B.
$$
Therefore the value of $\lambda_M$ should satisfy
$$
\lambda_M = \frac{B - \int_0^t \phi(c_x) \, dx}{\int_{t}^{n-M} \sqrt{\phi(c_x)}\,dx}.
$$ 
Plug the above expression of $\lambda_M$ into~\eqref{eq:opt_A_d2}, we get
\begin{align*}
Q = & \frac{\partial }{\partial M} \left(   \frac{\left(\int_{t}^{n-M} \sqrt{\phi(c_x)}\,dx\right)^2 }{B - \int_0^t \phi(c_x) \, dx}   \right) \\
=\, &   \frac{1}{B  - \int_0^t \phi(c_x) \, dx}  \left(2 \cdot \int_{t}^{n-M} \sqrt{\phi(c_x)}\,dx\right) \left(-\sqrt{\phi(c_{n-M})}\right)\\
=\, & \frac{2}{\lambda_M}  \left(-\sqrt{\phi(c_{n-M})}\right)
\end{align*}
by definition of $\lambda_M$. Plug it into~\eqref{eq:opt_A_derivative},
\begin{align*} 
\frac{\partial g(M)}{\partial M} 
=  \beta^2 \cdot \frac{1}{n} \cdot  \frac{2}{\lambda_M}  \left(-\sqrt{\phi(c_{n-M})}\right).
\end{align*}
 Recall that $M$ represents the amount of data that is ignored. When $M$ increases,   $\lambda_M$ will increase, and $\sqrt{\phi(c_{n-M})}$ will  decrease as the regularized virtual cost function is an increasing function. Therefore this partial derivative 
$$
 \beta^2 \cdot \frac{1}{n} \cdot  \frac{2}{\lambda_M}  \left(-\sqrt{\phi(c_{n-M})}\right) =  - \beta^2 \cdot \frac{2}{n} \cdot  \frac{1}{A^*_M(n-M)}
$$
is an increasing function of $M$. 
\end{proof}

Then the optimal threshold
\begin{align*} 
M^* = \arg\min_{M} & \ g(M) + \left( \frac{M}{n} \right)^2.
\end{align*}
 should have the derivative greater than zero on its right, i.e., for $M > M^*$,
$$
\frac{2M}{n^2} > \beta^2 \cdot \frac{2}{n} \cdot  \frac{1}{A^*_M(n-M)}.
$$
 and smaller than zero on the left, i.e., for $M < M^i$,
$$
\frac{2M}{n^2} < \beta^2 \cdot \frac{2}{n} \cdot  \frac{1}{A^*_M(n-M)}.
$$
Thus the optimal $M^*$ can be found by binary search.

\subsection{Optimal Mechanism for Adjacent Cost Sets} \label{app:CI_adj}
The following lemmas will be used in the proof of the main theorem.

\begin{lemma} \label{lem:biased_adj1}
Let $T_1$ and $T_2$ be two adjacent costs sets that have $T_2 = \{c_1, \dots, c_k, \dots, c_{i+1}\}$ with $c_1\le \cdots \le c_{i+1}$ and $T_1 = \{c_1, \dots, c_{k-1}, c_{k+1}, \dots, c_{i+1}\}$. Let $B$ be an arbitrary non-negative number. We use $OPT(T, B)$ to represent the optimal confidence interval mechanism (at round $i$) defined in~\eqref{prog:biased_roundi} when the cost set is $T$ and the budget is $B$. Then define 
\begin{itemize}
    \item $(A,U, P) = OPT(T_1, B/2)$, i.e., $(A,U,P)$ is the optimal mechanism when the cost set is $T_1$ and the budget is $B/2$.
    \item $(A', U', P') = OPT(T_2, B)$, i.e., $(A', U', P')$ is the optimal mechanism when the cost set is $T_2$ and the budget is $B$.
\end{itemize}  
Then we have
$$
(1-U_{k+1}) A(c_{k+1})P(c_{k+1}) \le 2 \cdot (1-U'_{k+1}) A'(c_{k+1})P'(c_{k+1}).
$$
\end{lemma}
\begin{proof}
The proof of this lemma is based on the proof of Lemma~\ref{lem:baised_eq_point}.

We use $c_r \in T_1$ to represent the largest cost that has $U_r < 1$, and use $c_{r'} \in T_2$ to represent the largest cost that has $U'_{r'} <1$. 
And define $M = \sum_{c\in T_1} U_c$ and $M' = \sum_{c\in T_2} U'_c$. To better illustrate the idea of our proof, we assume that the minimum point of our objective function is always achieved when its derivative with respect to $M$ is equal to $0$, which means that 
\begin{align*}
& a \cdot  \frac{1}{A(c_r)}  = M \\
& a \cdot \frac{1}{A'(c_{r'})} = M'
\end{align*}
where $a=4\alpha_\gamma^2 \cdot i/n$ is a constant (see the proof of Lemma~\ref{lem:baised_eq_point} for more details). The following analysis can be extended to the case when the derivative does not exists at the optimal point, but is greater than zero on the right and smaller than zero on the left.

First of all if $k > r'$, which means mechanism $(A', U', P')$ never purchase $c_k$, as well as $c_{k+1}$. Then we can remove $c_k$ from $T_2$ to get a new optimal mechanism $OPT(T_1, B)$. Then the new $r'$ should only move to the left,\footnote{This is because (1) $M'$ decreases and (2) $Avg(r'+1, j)$ will only increase for all $j\ge r'+1$. } which means we still have $k>r'$ and the new mechanism never purchase $c_k$. Because the budget of $(A,U, P) = OPT(T_1, B/2)$ is even smaller, it will not purchase $c_k$ or $c_{k+1}$ either. Therefore the inequality holds with both sides equal to zero.

Now suppose $k \le r'$. To compare the two mechanisms, we define $(A'', U'', P'')$ to be the optimal mechanism when (1) the cost set is $T_1$ and the budget is $B/2$ (which is the same as those of $(A,U,P)$) and (2) the total probability of ignoring the data points equals that of $U'$, i.e., $\sum_{c\in T_1} U''_c = \sum_{c\in T_2} U'_c$. Then we can apply Lemma~\ref{lem:unbiased_adj} to compare the non-zero parts of the two mechanisms. By the first line of the results in Lemma~\ref{lem:unbiased_adj},
$$
A''(c_{r'}) \le A'(c_{r'}), 
$$
and thus 
$$
a \cdot \frac{1}{A''(c_{r'})} \ge a \cdot \frac{1}{A'(c_{r'})} = M',
$$
According to the proof of Lemma~\ref{lem:baised_eq_point},  the optimal mechanism $(A,U,P)$ must have $r \le r'$ and $M \ge M'$. Then 
$$
A(c_r) = a/M \le a/M' = A'(c_{r'}) \le A'(c_r).
$$
According to Lemma~\ref{lem:unbiased_adj_lem1}, $\phi_2(c_r) \ge \frac{1}{2} \phi_1(c_r)$, together with the characterization of the optimal allocation rule in Theorem~\ref{thm:CI_roundi}, (for the mechanism to be non-trivial, we assume $A(c_r) <1$)
$$
\lambda = A(c_r) \sqrt{\phi_1(c_r)} \le A'(c_r) \sqrt{2\phi_2(c_r)} \le \sqrt{2} \lambda'.
$$
Again by Lemma~\ref{lem:unbiased_adj_lem1}, $\phi_2(c_j) \le 2 \phi_1(c_j)$ for all $j\ge k+1$. Thus for all $j\ge k+1$,
$$
A(c_j) = \min \left\{ 1, \frac{\lambda}{\sqrt{\phi_1(c_j)}} \right\} \le  \min \left\{ 1, \frac{\sqrt{2} \lambda'}{\sqrt{\phi_2(c_j)/2}} \right\} \le 2 \min \left\{ 1, \frac{ \lambda'}{\sqrt{\phi_2(c_j)}} \right\} = 2 A'(c_j).
$$
In addition, it should hold that $U^j \ge U'_j$ for all $j\ge k+1$.\footnote{This is because (1) $M\ge M'$ and (2) $Avg(l,r)$ will only increase for all $l\le k+1, r\ge k+1$ and $Avg(l,r)$ will only decrease for all $l,r>k+1$, so $I_{k+1}$ will only extend to the right.} According to the definition in Lemma~\ref{myerson} (see $(1-U(c))A(c)$ as an allocation rule here), the expected payment can be equivalently represented as an integral as in Lemma~\ref{lem:true_myerson},
\begin{align*}
(1-U_{k+1}) A(c_{k+1})P(c_{k+1}) = & (1-U_{k+1})  A(c_{k+1})c_{k+1} + \int_{c_{k+1}}^{\overline{C}} (1-U_v)A(v) \, dv \\
 \le & 2(1-U'_{k+1})  A'(c_{k+1})c_{k+1} + 2 \int_{c_{k+1}}^{\overline{C}} (1-U'_v)A'(v) \, dv \\
 = & 2 (1-U'_{k+1}) A'(c_{k+1})P'(c_{k+1}),
\end{align*}
which completes our proof.
\end{proof}

\begin{lemma} \label{lem:biased_adj2}
Let $T_1$ and $T_2$ be two adjacent costs sets that have $T_2 = \{c_1, \dots, c_k, \dots, c_{i+1}\}$ with $c_1\le \cdots \le c_{i+1}$ and $T_1 = \{c_1, \dots, c_{k-1}, c_{k+1}, \dots, c_{i+1}\}$. Let $B$ be an arbitrary non-negative number. We use $OPT(T, B)$ to represent the optimal confidence interval mechanism (at round $i$) defined in~\eqref{prog:biased_roundi} when the cost set is $T$ and the budget is $B$. Then define 
\begin{itemize}
    \item $(A,U, P) = OPT(T_1, B)$, i.e., $(A,U,P)$ is the optimal mechanism when the cost set is $T_1$ and the budget is $B$.
    \item $(A', U', P') = OPT(T_2, B/2)$, i.e., $(A', U', P')$ is the optimal mechanism when the cost set is $T_2$ and the budget is $B/2$.
\end{itemize} 
Define $M = \sum_{c\in T_1} U_c$ and $M' = \sum_{c\in T_2} U'_c$ and $W_{k+1} = \mathbbm{1}\left( U_{k+1} \ge \frac{1}{2} \right)$. Then we have 
\begin{itemize}
    \item $M\le M'$.
    \item 
$4 \cdot  \frac{\alpha_\gamma^2}{n} \cdot \frac{1 - W_{k+1}}{A(c_{k+1})} \le 8 \cdot  \frac{\alpha_\gamma^2}{n} \cdot \frac{1 - U'_{k+1}}{A'(c_{k+1})}+  U'_{k+1}\cdot \frac{M'}{i}$.
\end{itemize}
\end{lemma}
\begin{proof}
The proof of this lemma is based on the proof of Lemma~\ref{lem:baised_eq_point}.

Again we use $c_r \in T_1$ to represent the largest cost that has $U_r < 1$, and use $c_{r'} \in T_2$ to represent the largest cost that has $U'_{r'} <1$. To better illustrate the idea of our proof, we assume that the minimum point of our objective function is achieved when its derivative with respect to $M$ is equal to $0$, which means that 
\begin{align} 
& a \cdot  \frac{1}{A(c_r)}  = M \label{eqn:biased_adj2_eqn1} \\
& a \cdot \frac{1}{A'(c_{r'})} = M' \label{eqn:biased_adj2_eqn2}
\end{align}
where $a=4\alpha_\gamma^2 \cdot i/n$ is a constant (see the proof of Lemma~\ref{lem:baised_eq_point} for more details). The following analysis can be extended to the case when the derivative does not exists at the optimal point, but is greater than zero on the right and smaller than zero on the left.

First consider the case $k+1>r'$, which means $U'_{k+1} = 1$ and mechanism $(A', U', P')$ never purchase $c_{k+1}$. If $k+1$ is also greater than $r$, which means that $W_{k+1} = 1$ and $(A,U, P)$ never purchases $c_{k+1}$,  then the inequality trivially holds with both sides equal to zero. If $k+1 \le r$, then since $U'_{k+1} = 1$,
\begin{align*}
\frac{a}{i} \cdot \frac{1 - W_{k+1}}{A(c_{k+1})} \le & \frac{a}{i} \cdot \frac{1}{A(c_r)} \\
= & \frac{M}{i}  \\
= & \frac{a}{i} \cdot \frac{1 - U'_{k+1}}{A'(c_{k+1})}+  U'_{k+1} \cdot \frac{M}{i} \\
\le & \frac{a}{i} \cdot \frac{1 - U'_{k+1}}{A'(c_{k+1})}+  U'_{k+1}\cdot \frac{M'}{i}.
\end{align*} 


Now suppose $k+1\le r'$. To compare the two mechanisms, we define $(A'', U'', P'')$ to be the optimal mechanism when (1) the cost set is $T_1$ and the budget is $B$ (which is the same as those of $(A,U,P)$) and (2) the total probability of ignoring the data points equals that of $U'$, i.e., $\sum_{c\in T_1} U''_c = \sum_{c\in T_2} U'_c$. Then we can apply Lemma~\ref{lem:unbiased_adj} to compare the non-zero parts of the two mechanisms. By the second line of the results in Lemma~\ref{lem:unbiased_adj},
$$
A''(c_{r'}) \ge A'(c_{r'}), 
$$
and thus 
$$
a \cdot \frac{1}{A''(c_{r'})} \le a \cdot \frac{1}{A'(c_{r'})} = M'.
$$
According to the proof of Lemma~\ref{lem:baised_eq_point},  the optimal mechanism $(A,U,P)$ must have $M \le M'$ and $r \ge r'$.
Then by~\eqref{eqn:biased_adj2_eqn1} and~\eqref{eqn:biased_adj2_eqn2},
$$
A(c_{r'}) \ge A(c_r) = a/M \ge a/M' = A'(c_{r'}).
$$
According to Lemma~\ref{lem:unbiased_adj_lem1}, $\phi_2(c_{r'}) \le 2 \phi_1(c_{r'})$, together with the characterization of the optimal allocation rule in Theorem~\ref{thm:CI_roundi}, (for the mechanism to be non-trivial, we assume $A'(c_{r'}) <1$)
$$
\lambda \ge A(c_{r'}) \sqrt{\phi_1(c_{r'})} \ge A'(c_{r'}) \sqrt{\phi_2(c_{r'})/2} = \sqrt{1/2} \lambda'.
$$
Again by Lemma~\ref{lem:unbiased_adj_lem1}, $\phi_2(c_j) \ge \frac{1}{2} \phi_1(c_j)$ for all $j\ge k+1$. Thus for all $j\ge k+1$,
$$
A(c_j) = \min \left\{ 1, \frac{\lambda}{\sqrt{\phi_1(c_j)}} \right\} \ge  \min \left\{ 1, \frac{\sqrt{1/2} \lambda'}{\sqrt{2\cdot \phi_2(c_j)}} \right\} \ge \frac{1}{2} \min \left\{ 1, \frac{ \lambda'}{\sqrt{\phi_2(c_j)}} \right\} = \frac{1}{2} A'(c_j).
$$
In addition, $\frac{1}{A(c_{k+1})} \le \frac{1}{A(c_{r})} = \frac{M}{a}$ due to the monotonicity constraint. Therefore 
\begin{align*}
\frac{a}{i} \cdot \frac{1 - W_{k+1}}{A(c_{k+1})} =& \frac{a}{i} \cdot \frac{1 - U'_{k+1}+U'_{k+1}}{A(c_{k+1})} \\
\le & \frac{a}{i} \cdot \frac{1 - U'_{k+1}}{A(c_{k+1})} + U'_{k+1} \cdot \frac{M}{i} \\
\le & \frac{2a}{i}\cdot \frac{1 - U'_{k+1}}{A'(c_{k+1})}+  U'_{k+1}\cdot \frac{M}{i}\\
\le & \frac{2a}{i}\cdot \frac{1 - U'_{k+1}}{A'(c_{k+1})}+  U'_{k+1} \cdot \frac{M'}{i}.
\end{align*}
which completes our proof.
\end{proof}

\subsection{Proof of Theorem~\ref{thm:CI_main}} \label{app:CI_main}
We now prove the main theorem for the online confidence interval mechanism. Define set $T_i = \{c_1, \dots, c_{i-1}, \overline{C}\}$. We use $c_{(1)}, \dots, c_{(i)}$ to denote the elements in $T_i$ ordered from smallest to largest. Let
\begin{align*}
A^i, \ U^i = &\arg\min_{A, U}  \quad 4 \cdot  \frac{\alpha_\gamma^2}{n} \cdot \frac{1}{i}  \sum_{j = 1}^i \frac{(1-U_j)}{A(c_{(j)})} + \left(\frac{1}{i} \sum_{j=1}^i U_j  \right)^2 \\
& \qquad \textrm{s.t.} \  \sum_{j=1}^i (1-U_j) \cdot A(c_{(j)}) \psi^i(c_{(j)}) \le \frac{B}{16\sqrt{n/i}} \notag\\
& \qquad  \qquad  (1-U_j) \cdot A(c_{(j)})  \textrm{ is monotone non-increasing in $j$} \notag\\
& \qquad  \qquad 0\le A(c_{(j)}) \le 1, \quad 0 \le U_j \le 1,  \quad \forall j \notag
\end{align*}
We prove that when we use the extended allocation rule of $\left(A^i(c), \mathbbm{1}\left(U^i(c) \ge \frac{1}{2}\right)\right)$ at round $i$, and output confidence interval $\left[\sum_{i=1}^n y_i / n - \frac{\alpha_\gamma}{\sqrt{n}} \cdot \widehat{\sigma}, \quad \sum_{i=1}^n y_i/ n + \frac{\widehat{U}}{n} + \frac{\alpha_\gamma}{\sqrt{n}} \cdot \widehat{\sigma}\right]$ at the end, where $\widehat{U} = \sum_{i=1}^n \widehat{U}_i$, our mechanism (1)  is truthful in expectation and individually rational (2) satisfies the budget constraint $B$ in expectation; (3) and the for any cost distribution $\{c_1, \dots, c_n\}$, the worst-case expected length of the output confidence interval is no more than 
$$
L \le 8\sqrt{10} \cdot L^* +  \frac{2\sqrt{10}}{\sqrt{n}} + o(1/\sqrt{n}),
$$
where $L^*$ is the approximate worst-case expected length of the benchmark $A^*, U^*$ defined in Lemma~\ref{lem:biased_benchmark},
\begin{align} \label{prog:biased_benchmark_app}
L^* =  \min_{A} \max_{\vec{z}\in[0,1]^n} \quad & 2 \cdot  \frac{\alpha_\gamma}{\sqrt{n+1}} \cdot \sqrt{\frac{1}{n+1} \cdot \sum_{i=1}^{n+1} \frac{(1-U_i) z_i^2}{A(c_{(i)})}} + \frac{\sum_{i=1}^{n+1} z_i \cdot U_i}{n+1}  \\
\textrm{s.t.} \quad & \sum_{i=1}^{n+1} (1-U_i) \cdot A(c_{(i)}) \psi(c_{(i)})  \le B \notag\\
& (1-U_c)A(c) \textrm{ is monotone non-increasing in } c \notag\\
& 0\le A(c) \le 1, \quad  0 \le U_c \le 1, \quad \forall c \notag
\end{align}
For convenience we define 
 $$
 V^* = 4 \cdot  \frac{\alpha_\gamma^2}{n} \cdot \frac{1}{n} \sum_{i=1}^n \frac{1 - U^*(c_i)}{A^*(c_i)}
 $$ 
so that $L^* \ge \frac{n}{n+1} \sqrt{V^*} + \frac{\sum_{i=1}^n U^*_i}{n}$,  where $\sum_{i=1}^n U^*_i$ is the expected total number of data ignored by $A^*$.

The basic idea of the proof is the same as the unbiased case. We will again use an intermediate mechanism $(A', U', P')$ to compare $(A^i, U^i, P^i)$ with the benchmark at each round $i$. The difference of $(A',U',P')$ and $(A^i, U^i, P^i)$ is bounded using the results in the preceding section. 

We use the same notations as in the proofs of the unbiased case~\ref{app:unbiased_online}. The proofs of (1) in~\ref{app:unbiased_online} can be directly applied here, so we only prove (2) and (3). 

\subsubsection{Expected budget feasibility}
The proof of budget feasibility is similar to the unbiased case. 
To bound $\sum_{i=1}^n \E[(1-\widehat{U}_i) A^i(c_i)P^i(c_i)]$, we again define $A'$ conditioning on a fixed  $T_{i+1} = S_{i+1} \cup \{\overline{C}\}$. 
 \begin{align*}
A', U' = &\arg\min_{A, U}  \quad 4 \cdot  \frac{\alpha_\gamma^2}{n} \cdot \frac{1}{i}  \sum_{c \in T_{i+1}} \frac{(1-U_c)}{A(c)} + \left(\frac{1}{i} \sum_{c\in T_{i+1}} U_c \right)^2 \\
& \qquad  \textrm{s.t.} \sum_{c\in T_{i+1}} (1-U_c) \cdot A(c) \psi^i(c) \le \frac{B}{8\sqrt{n/i}}  \notag\\
& \qquad  \qquad (1-U_c) \cdot A(c)  \textrm{ is monotone non-increasing in $c$} \notag\\
& \qquad  \qquad 0\le A(c) \le 1, \quad 0 \le U_c \le 1,  \quad \forall c \notag
\end{align*}
Since $T_i,\,T_{i+1}$ only differ by one element $c_i$, and $A'$ uses twice the budget of $A^i$, according to Lemma~\ref{lem:biased_adj1}, 
 \begin{align*}
(1- U^i(c_i)) A^i(c_i) P^i(c_i) = & (1- U^i(c_{(k+1)})) A^i(c_{(k+1)})  P^i(c_{(k+1)})  \\
\le & 2(1-U'(c_{(k+1)}))  A'(c_{(k+1)})  P'(c_{(k+1)}).
 \end{align*}
And thus 
\begin{align*}
\left(1- \mathbbm{1}\left(U^i(c_{(k+1)})\ge 1/2\right) \right)A^i(c_{(k+1)})  P^i(c_{(k+1)} \le & 2(1- U^i(c_{(k+1)})) A^i(c_{(k+1)})  P^i(c_{(k+1)})\\
\le & 4(1-U'(c_{(k+1)}))  A'(c_{(k+1)})  P'(c_{(k+1)}).
\end{align*}
 Now assume the set of  the first $i$ costs is $S_{i+1}$. When the data holders come in random order, $c_i$ is a random element chosen from $S_{i+1}$. Therefore $c_i$'s  rank $k$ should be uniformly distributed over $\{1, \dots, i\}$, and thus
\begin{align*}
\E[(1-\widehat{U}_i)A^i(c_i) \cdot P^i(c_i) | S_{i+1}] = &\frac{1}{i} \sum_{j=1}^i \left(1- \mathbbm{1}\left(U^i(c_{(j+1)})\ge 1/2\right) \right)A^i(c_{(j+1)})  P^i(c_{(j+1)})\\
\le &  \frac{1}{i} \sum_{j=1}^{i} 4(1-U'(c_{(j+1)}))  A'(c_{(j+1)})  P'(c_{(j+1)})\\
\le & \frac{4}{i} \cdot \sum_{c\in T_{i+1}} (1-U'(c))A'(c)  P'(c)\\
\le & \frac{1}{2\sqrt{n\cdot i}}  \cdot B
\end{align*}
for all $S_{i+1}$.
Therefore the total spending of the mechanism is bounded as 
\begin{align*}
\sum_{i=1}^n \E[(1-\widehat{U}_i)A^i(c_i)\cdot P^i(c_i)] & \le  \sum_{i=1}^n \frac{B}{2\sqrt{ n \cdot i}} \le B,
\end{align*}
since $ \sum_{i=1}^n \frac{1}{\sqrt{i}} \le 2\sqrt{n}$.

\subsubsection{Competitive analysis}
The expected length of the output confidence interval is equal to 
\begin{align*}
L = 2\cdot \frac{\alpha_\gamma}{\sqrt{n}} \cdot \E[\widehat{\sigma}] + \E[\widehat{U}]/n.
\end{align*}
To compare it with $L^*$, we first upper bound the expectation of the sample variance $\E[\widehat{\sigma}^2]$ with $\E[y_i^2]$.
\begin{lemma} \label{lem:app_biased_lem1}
For any underlying distribution $\mathcal{D}$, our mechanism has 
$$
\E[\widehat{\sigma}^2 ] \le \frac{1}{n}\sum_{i=1}^n  \E[y_i^2],
$$
where $\widehat{\sigma}^2$ is the sample variance of the re-weighted data $y_1 = \frac{\widehat{x}_1}{A^1(c_1)}, \dots, y_n = \frac{\widehat{x}_n}{A^n(c_n)}$.
\end{lemma}
\begin{proof}
Because $y_i$'s are independent conditioned on $c_1, \dots, c_n$,
\begin{align*}
\E[\widehat{\sigma}^2 \vert c_1, \dots, c_n]= &\E\left[\sum_{i=1}^n \left(y_i - \sum_{i=1}^n y_i/n\right)^2/(n-1) \Big \vert c_1, \dots, c_n\right]\\
= & \frac{1}{n-1} \cdot \E\left[\sum_{i=1}^n y_i^2 - \frac{2}{n}\sum_{i=1}^n y_i \sum_{i=1}^n y_i + n \cdot \left( \sum_{i=1}^n y_i/n\right)^2 \Big\vert c_1, \dots, c_n\right]\\
= & \frac{1}{n-1} \cdot \E\left[\sum_{i=1}^n y_i^2 - \frac{1}{n} \left(\sum_{i=1}^n y_i \right)^2 \Big \vert c_1, \dots, c_n\right]\\
= & \frac{1}{n-1} \cdot \E\left[\sum_{i=1}^n y_i^2 - \frac{1}{n} \sum_{i=1}^n y_i^2 -\frac{1}{n} \sum_{i=1}^n \sum_{j\neq i} y_i y_j \Big \vert c_1, \dots, c_n\right]\\
= & \frac{1}{n} \cdot \sum_{i=1}^n \E[ y_i^2]  -\frac{1}{n(n-1)} \sum_{i=1}^n \sum_{j\neq i} \E[y_i \vert c_1, \dots, c_n] \E[ y_j \vert c_1, \dots, c_n]\\
\le &\frac{1}{n} \cdot \sum_{i=1}^n \E[ y_i^2] 
\end{align*}
\end{proof}

We then prove that the sum of the squares of the two terms in $L$ can be bounded using $V^*$ and $(U^*/n)^2$. We first prove the following lemma:
\begin{lemma} \label{lem:biased_decompose}
Let $S_{i+1}$ be the set of first $i$ agents' costs. When the agents come in random order, $S_{i+1}$ is a random subset of $\{c_1, \dots, c_n\}$ with length $i$. Then it satisfies that
\begin{align*}
4 \cdot  \frac{\alpha_\gamma^2}{n} \cdot  \E[\widehat{\sigma}]^2 + \left(\E[\widehat{U}]/n\right)^2 \le \frac{1}{n} \sum_{i=1}^n \E_{S_{i+1}} \left[4 \cdot  \frac{\alpha_\gamma^2}{n} \cdot \E[y_i^2|S_{i+1}] +  \E[\widehat{U}_i|S_{i+1}]^2 \right].
\end{align*}
\end{lemma}
\begin{proof}
First by Jensen's inequality,
\begin{align*}
 4 \cdot  \frac{\alpha_\gamma^2}{n} \cdot  \E[\widehat{\sigma}]^2 + (\E[\widehat{U}]/n)^2 \le & 4 \cdot  \frac{\alpha_\gamma^2}{n}\cdot  \E[\widehat{\sigma}^2] + \left(\sum_{i=1}^n \E[\widehat{U}_i]/n\right)^2 
\le  4 \cdot  \frac{\alpha_\gamma^2}{n}\cdot  \E[\widehat{\sigma}^2] + \frac{1}{n} \sum_{i=1}^n \E[\widehat{U}_i]^2.
\end{align*}
As we prove in Lemma~\ref{lem:app_biased_lem1}, $ \E[\widehat{\sigma}^2] \le \frac{1}{n} \sum_{i=1}^n \E[y_i^2]$, so we further have
\begin{align*}
4 \cdot  \frac{\alpha_\gamma^2}{n} \cdot  \E[\widehat{\sigma}]^2 + \left(\E[\widehat{U}]/n\right)^2 \le &4 \cdot  \frac{\alpha_\gamma^2}{n} \cdot \frac{1}{n} \sum_{i=1}^n \E[y_i^2] + \frac{1}{n} \sum_{i=1}^n \E[\widehat{U}_i]^2 \\
\le & 4 \cdot  \frac{\alpha_\gamma^2}{n} \cdot \frac{1}{n} \sum_{i=1}^n \E_{S_{i+1}} \E[y_i^2|S_{i+1}] + \frac{1}{n} \sum_{i=1}^n \E_{S_{i+1}} \E[\widehat{U}_i|S_{i+1}]^2 \\
 = & \frac{1}{n} \sum_{i=1}^n \E_{S_{i+1}} \left[4 \cdot  \frac{\alpha_\gamma^2}{n} \cdot \E[y_i^2|S_{i+1}] +  \E[\widehat{U}_i|S_{i+1}]^2 \right].
\end{align*}
\end{proof} 
Then the sum of the squares of the two terms in $L$ can be bounded as follows:
\begin{lemma} \label{lem:biased_bound_square}
For any underlying distribution $\mathcal{D}$, our mechanism has 
\begin{align*}
  & 4 \cdot  \frac{\alpha_\gamma^2}{n} \cdot \E[\widehat{\sigma}]^2 + \left(\E[\widehat{U}]/n\right)^2 \\
  \le &  320\cdot V^* + 20 \left(\frac{\sum_{i=1}^n U^*_i}{n}\right)^2 + 10 \left(\frac{1+\ln n}{n} \right)^2 +   \frac{2}{n} \left( 640  \cdot  \frac{\alpha_\gamma^2}{\sqrt{n}} \cdot \frac{1-U^*_{\overline{C}}}{A^*(\overline{C})} + 10\left( U^*_{\overline{C}}\right)^2\right).
\end{align*}
where $\widehat{\sigma}^2$ is the sample variance of the re-weighted data $y_1 = \frac{\widehat{x}_1}{A^1(c_1)}, \dots, y_n = \frac{\widehat{x}_n}{A^n(c_n)}$, and $\widehat{U}$ is the number of data points ignored by the mechanism, i.e., $\widehat{U} = \sum_{i=1}^n \widehat{U}_i$.
\end{lemma}
\begin{proof}
Let $S_{i+1}$ be the set of first $i$ agents' costs. When the agents come in random order, $S_{i+1}$ is a random subset of $\{c_1, \dots, c_n\}$ with length $i$. 
We first consider a fixed $S_{i+1}$ and bound $4 \cdot  \frac{\alpha_\gamma^2}{n} \cdot \E[y_i^2|S_{i+1}] +  \E[\widehat{U}_i|S_{i+1}]^2$. Let $S_i \subseteq S_{i+1}$ be the set of first $i-1$ costs is $S_i$, the allocation rule $A^i$ is uniquely decided as in~\eqref{prog:biased_roundi_app} where $T_i = S_i \cup \{\overline{C}\}$. Again if $c_i = c_{(k)}$, we will have $A^i(c_i) = A^i(c_{(k+1)})$. 

Let $T_{i+1} = S_{i+1} \cup \{\overline{C}\}$. We define the following intermediate allocation rule $A', U'$
 \begin{align} \label{prog:biased_Aprime}
A', U' = &\arg\min_{A, U}  \quad 4 \cdot  \frac{\alpha_\gamma^2}{n} \cdot \frac{1}{i}  \sum_{c \in T_{i+1}} \frac{(1-U_c)}{A(c)} + \left(\frac{1}{i} \sum_{c\in T_{i+1}} U_c \right)^2 \\
& \qquad  \textrm{s.t.} \sum_{c\in T_{i+1}} (1-U_c) \cdot A(c) \psi^i(c) \le \frac{B}{32\sqrt{n/i}} \notag\\
& \qquad  \qquad (1-U_c) \cdot A(c)  \textrm{ is monotone non-increasing in $c$} \notag\\
& \qquad  \qquad 0\le A(c) \le 1, \quad 0 \le U_c \le 1,  \quad \forall c \notag
\end{align}

\begin{claim}[Compare $\left(A^i, \mathbbm{1}\left(U^i \ge \frac{1}{2}\right)\right)$ with $(A', U')$] \label{clm:biased_clm1}
 Define $M' =\sum_{c\in T_{i+1}} U'_c$, then
$$
4 \cdot  \frac{\alpha_\gamma^2}{n} \cdot \E[y_i^2|S_{i+1}] +  \E[\widehat{U}_i|S_{i+1}]^2 \le 5 \left(4 \cdot  \frac{\alpha_\gamma^2}{n} \cdot \frac{1}{i} \sum_{c\in T_{i+1}} \frac{1 - U'_c}{A'(c)} + \left(\frac{M'}{i}\right)^2\right).
$$
\end{claim}
\begin{proof}
When the data holders come in random order,  the set of first $i-1$ costs $S_i$ is a random subset of $S_{i+1}$. Therefore $c_i$'s rank $k$ is uniformly distributed over $\{1,\dots, i\}$. Let $A^{i,k}, U^{i,k}$  be the optimal solution in round $i$ when the $i$-th cost $c_i$ has rank $k$, i.e., $c_i = c_{(k)}$.  Then 
\begin{align*} 
   & 4 \cdot  \frac{\alpha_\gamma^2}{n} \cdot \E[y_i^2 | S_{i+1}] +  \E[\widehat{U}_i | S_{i+1}]^2 \notag \\
  =  & 4 \cdot  \frac{\alpha_\gamma^2}{n}  \left( \frac{1}{i} \sum_{k=1}^i \frac{1 - \mathbbm{1}\left( U^{i,k}_{k+1} \ge \frac{1}{2} \right)}{A^{i,k}(c_{(k+1)})} \right) + \left( \frac{1}{i} \sum_{k=1}^i \mathbbm{1}\left( U^{i,k}_{k+1} \ge \frac{1}{2} \right) \right)^2.
\end{align*}
Denote $\mathbbm{1}\left( U^{i,k}_{k+1} \ge \frac{1}{2} \right)$ by $W^{i,k}_{k+1}$, the above equality becomes
\begin{align} \label{eq:biased_lalala}
   4 \cdot  \frac{\alpha_\gamma^2}{n} \cdot \E[y_i^2 | S_{i+1}] +  \E[\widehat{U}_i | S_{i+1}]^2  =   4 \cdot  \frac{\alpha_\gamma^2}{n}  \left( \frac{1}{i} \sum_{k=1}^i \frac{1 -W^{i,k}_{k+1}}{A^{i,k}(c_{(k+1)})} \right) + \left( \frac{1}{i} \sum_{k=1}^iW^{i,k}_{k+1} \right)^2.
\end{align}
To compare it with $(A', U')$, we first claim the follows
\begin{claim}
 Let $M' =  \sum_{c\in T_{i+1}} U'_c$, then 
$$
\frac{1}{i} \sum_{k=1}^i W^{i,k}_{k+1} = \frac{1}{i} \sum_{k=1}^i \mathbbm{1}\left( U^{i,k}_{k+1} \ge \frac{1}{2} \right) \le 2 M'/i.
$$
\end{claim}
\begin{proof}
Define $$M^{i,k} =  U^{i,k}_1 + \cdots + U^{i,k}_{k-1} + U^{i,k}_{k+1} + \cdots + U^{i,k}_i.$$ By Lemma~\ref{lem:biased_adj2}, 
$M^{i,k} \le M'$  for every $k\in [i]$, which means
$$
U^{i,k}_1 + \cdots + U^{i,k}_{k-1} + U^{i,k}_{k+1} + \cdots + U^{i,k}_i \le  M',
$$
so each $U^{i,k}$ has at most $2 M'$ entries that are no less than $\frac{1}{2}$. And since $U^{i,k}$ is monotone non-decreasing, i.e., $U^{i,k}_1 \le \cdots \le U^{i,k}_{k-1} \le U^{i,k}_{k+1} \le \cdots \le U^{i,k}_i$, only the last $2M'$ entries can be no less than $\frac{1}{2}$. Therefore there are at most $2M'$ non-zero $\mathbbm{1}\left( U^{i,k}_{k+1} \ge \frac{1}{2} \right)$, thus
$$
\frac{1}{i} \sum_{k=1}^i \mathbbm{1}\left( U^{i,k}_{k+1} \ge \frac{1}{2} \right) \le 2 M'/i.
$$
\end{proof}
Meanwhile,  according to Lemma~\ref{lem:biased_adj2}, 
\begin{equation} \label{eqn:biased_main_eq1}
4 \cdot  \frac{\alpha_\gamma^2}{n} \cdot \frac{1 - W^{i,k}_{k+1}}{A^{i,k}(c_{(k+1)})} \le 8 \cdot  \frac{\alpha_\gamma^2}{n} \cdot \frac{1 - U'_{k+1}}{A'(c_{(k+1)})}+  U'_{k+1} \cdot \frac{M'}{i}
\end{equation}
for all $k$.
Then \eqref{eq:biased_lalala} can be bounded as 
\begin{align*}
&  4 \cdot  \frac{\alpha_\gamma^2}{n} \cdot \E[y_i^2 | S_{i+1}] +  \E[\widehat{U}_i | S_{i+1}]^2 \\
=  & 4 \cdot  \frac{\alpha_\gamma^2}{n}  \left( \frac{1}{i} \sum_{k=1}^i \frac{1 -W^{i,k}_{k+1}}{A^{i,k}(c_{(k+1)})} \right) + \left( \frac{1}{i} \sum_{k=1}^iW^{i,k}_{k+1} \right)^2\\
\le & 8 \cdot    \frac{\alpha_\gamma^2}{n} \cdot \frac{1}{i} \sum_{c\in T_{i+1}} \frac{1 - U'_c}{A'(c)} + \frac{1}{i} \sum_{k=1}^i   U'_{k+1} \cdot \frac{M'}{i} +  (2 M'/i)^2\\
\le & 8 \cdot    \frac{\alpha_\gamma^2}{n} \cdot \frac{1}{i} \sum_{c\in T_{i+1}} \frac{1 - U'_c}{A'(c)} + \frac{M'}{i} \cdot \frac{M'}{i} +  (2 M'/i)^2\\
= & 8 \cdot    \frac{\alpha_\gamma^2}{n} \cdot \frac{1}{i} \sum_{c\in T_{i+1}} \frac{1 - U'_c}{A'(c)} + 5 \left(\frac{M'}{i}\right)^2.  
\end{align*}
which completes the proof of Claim~\ref{clm:biased_clm1}.
\end{proof}

\begin{claim}[Compare $A^*$ with $A'$] \label{clm:biased_clm2}
\begin{align*}
 &4 \cdot  \frac{\alpha_\gamma^2}{n} \cdot \frac{1}{i}  \sum_{c \in T_{i+1}} \frac{(1-U'_c)}{A'(c)} + \left(\frac{1}{i} \sum_{c\in T_{i+1}} U'_c \right)^2 
 \le  128\sqrt{\frac{n}{i}} \cdot  \frac{\alpha_\gamma^2}{n} \cdot \frac{1}{i}  \sum_{c \in T_{i+1}} \frac{(1-U^*_c)}{A^*(c)} + \left(\frac{1}{i} \sum_{c\in T_{i+1}} U^*_c \right)^2.
\end{align*}
\end{claim}
\begin{proof}
By the same reasoning as the proof of unbiased case (see Section~\ref{prg:compareA}), it can be proved that $\frac{A^*}{32\sqrt{n/i}}, U^*$ is a feasible solution of~\eqref{prog:biased_Aprime}.
\end{proof}

Combine Claim~\ref{clm:biased_clm1} and Claim~\ref{clm:biased_clm2}, 
\begin{align*}
 &4 \cdot  \frac{\alpha_\gamma^2}{n} \cdot \E[y_i^2 | S_{i+1}] +  \E[\widehat{U}_i | S_{i+1}]^2 \\
\le &5 \left(4 \cdot  \frac{\alpha_\gamma^2}{n} \cdot \frac{1}{i} \sum_{c\in T_{i+1}} \frac{1 - U'_c}{A'(c)} + (M')^2\right)\\
 \le & 5 \left( 128\sqrt{\frac{n}{i}} \cdot  \frac{\alpha_\gamma^2}{n} \cdot \frac{1}{i}  \sum_{c \in T_{i+1}} \frac{(1-U^*_c)}{A^*(c)} + \left(\frac{1}{i} \sum_{c\in T_{i+1}} U^*_c \right)^2 \right).
\end{align*}

Now we compute $\E_{S_{i+1}} \left[4 \cdot  \frac{\alpha_\gamma^2}{n} \cdot \E[y_i^2|S_{i+1}] +  \E[\widehat{U}_i|S_{i+1}]^2 \right] $ by averaging over $S_{i+1}$
Since $S_{i+1}$ is a random subset,  the first term of the right-hand side is just 
\begin{align*}
  128\sqrt{\frac{n}{i}} \cdot  \frac{\alpha_\gamma^2}{n} \cdot \frac{1}{i} \left( \sum_{c \in S_{i+1}} \frac{1-U^*_c}{A^*(c)} + \frac{1-U^*_{\overline{C}}}{A^*(\overline{C})} \right)= &  128\sqrt{\frac{n}{i}} \cdot  \frac{\alpha_\gamma^2}{n} \left( \frac{1}{n}  \sum_{j=1}^n \frac{1-U^*_j}{A^*(c_j)}  + \frac{1}{i} \cdot \frac{1-U^*_{\overline{C}}}{A^*(\overline{C})}\right) \\
= & 32\sqrt{\frac{n}{i}}\cdot V^* +  128\sqrt{\frac{n}{i}} \cdot  \frac{\alpha_\gamma^2}{n} \cdot \frac{1}{i} \cdot \frac{1-U^*_{\overline{C}}}{A^*(\overline{C})}.
\end{align*}
Next, we upper bound the second term $\E_{S_{i+1}} \left[\left(\frac{1}{i} \sum_{c\in S_{i+1}} U^*_c\right)^2\right]$. Since $S_{i+1} = \{c_1, \dots, c_i\}$ is a random subset, 
\begin{align*}
  &\E_{S_{i+1}} \left[\left(\frac{1}{i} \sum_{c\in T_{i+1}} U^*_c\right)^2\right] \\
 =  &\E_{S_{i+1}} \left[\left(\frac{1}{i} \sum_{c\in S_{i+1}} U^*_c + \frac{1}{i} \cdot U^*_{\overline{C}}\right)^2\right] \\
 \le &  \E_{S_{i+1}} \left[2 \left(  \frac{1}{i} \sum_{c\in S_{i+1}} U^*_c \right)^2 + 2 \left(   \frac{1}{i} \cdot  U^*_{\overline{C}}\right)^2\right]\\
 = & \frac{2}{i^2} \cdot \E_{S_{i+1}}  \left[\sum_{c \in S_{i+1}} (U^*_c)^2 + \sum_{c, c' \in S_{i+1}: c\neq c'} U^*_c \cdot U^*_{c'} \right] +  \frac{2}{i^2}  \left(  U^*_{\overline{C}}\right)^2 \\
  = & \frac{2}{i^2} \cdot \E_{S_{i+1}}  \left[\sum_{c} \mathbbm{1}(c \in S_{i+1}) (U^*_c)^2 + \sum_{c, c': c\neq c' } \mathbbm{1}(c, c'\in S_{i+1}) U^*_c \cdot U^*_{c'} \right] +  \frac{2}{i^2}  \left(  U^*_{\overline{C}}\right)^2 \\
= &\frac{2}{i^2} \cdot \frac{i}{n} \sum_{i=1}^n (U^*_i)^2  + \frac{2}{i^2} \cdot \frac{i(i-1)}{n(n-1)} \sum_{j=1}^n \sum_{k\in[n], k\neq j} U^*_j \cdot U^*_k  +  \frac{2}{i^2}  \left(  U^*_{\overline{C}}\right)^2 \\
\le &\frac{2}{i} \cdot \frac{1}{n} \sum_{i=1}^n U^*_i  + 2 \cdot \frac{(i-1)n}{i(n-1)} \cdot \frac{1}{n^2} \sum_{j=1}^n \sum_{k=1}^n U^*_j \cdot U^*_k  +  \frac{2}{i^2}  \left(  U^*_{\overline{C}}\right)^2 \\
 \le &\frac{2}{i} \cdot \frac{\sum_{i=1}^n U^*_i }{n}  + 2 \left(\frac{\sum_{i=1}^n U^*_i}{n}\right)^2 +  \frac{2}{i^2}  \left(  U^*_{\overline{C}}\right)^2.
\end{align*}
Therefore 
\begin{align*}
&\E_{S_{i+1}} \left[4 \cdot  \frac{\alpha_\gamma^2}{n} \cdot \E[y_i^2|S_{i+1}] +  \E[\widehat{U}_i|S_{i+1}]^2 \right]  \\
\le &5 \cdot  \E_{S_{i+1}} \left[ 128\sqrt{\frac{n}{i}} \cdot  \frac{\alpha_\gamma^2}{n} \cdot \frac{1}{i}  \sum_{c \in T_{i+1}} \frac{(1-U^*_c)}{A^*(c)} + \left(\frac{1}{i} \sum_{c\in T_{i+1}} U^*_c \right)^2 \right]\\
\le & 5  \left(32\sqrt{\frac{n}{i}}\cdot V^* +  128\sqrt{\frac{n}{i}} \cdot  \frac{\alpha_\gamma^2}{n} \cdot \frac{1}{i} \cdot \frac{1-U^*_{\overline{C}}}{A^*(\overline{C})} + \frac{2}{i} \cdot \frac{\sum_{i=1}^n U^*_i }{n}  + 2 \left(\frac{\sum_{i=1}^n U^*_i}{n}\right)^2 +  \frac{2}{i^2}  \left(  U^*_{\overline{C}}\right)^2\right)\\
= & 10 \left(\frac{\sum_{i=1}^n U^*_i}{n}\right)^2 + 160\sqrt{\frac{n}{i}}\cdot V^* + \frac{10}{i} \cdot \frac{\sum_{i=1}^n U^*_i }{n}  +  \frac{640}{i \sqrt{i}}  \cdot  \frac{\alpha_\gamma^2}{\sqrt{n}} \cdot \frac{1-U^*_{\overline{C}}}{A^*(\overline{C})}  +  \frac{10}{i^2}  \left(  U^*_{\overline{C}}\right)^2.
\end{align*}
Define $D_1 =10 \left(\frac{\sum_{i=1}^n U^*_i}{n}\right)^2$, $D_2 = 160\cdot V^*$, $D_3 =\frac{10\sum_{i=1}^n U^*_i }{n} $, $D_4 = 640 \cdot  \frac{\alpha_\gamma^2}{\sqrt{n}} \cdot \frac{1-U^*_{\overline{C}}}{A^*(\overline{C})}$, $D_5 = 10\left( U^*_{\overline{C}}\right)^2$, then the above expectation becomes $ D_1 + \frac{\sqrt{n}\cdot D_2}{\sqrt{i}} +\frac{D_3}{i} + \frac{D_4}{i\sqrt{i}} +\frac{D_5}{i^2}$.  Finally, by Lemma~\ref{lem:biased_decompose}, we get
\begin{align*}
& 4 \cdot  \frac{\alpha_\gamma^2}{n} \cdot  \E[\widehat{\sigma}]^2 + \left(\E[\widehat{U}]/n\right)^2 \\
\le & \frac{1}{n} \sum_{i=1}^n \E_{S_{i+1}} \left[4 \cdot  \frac{\alpha_\gamma^2}{n} \cdot \E[y_i^2|S_{i+1}] +  \E[\widehat{U}_i|S_{i+1}]^2 \right]\\
\le & \frac{1}{n} \sum_{i=1}^n \left( D_1 + \frac{\sqrt{n} \cdot D_2}{\sqrt{i}} +\frac{D_3}{i} + \frac{D_4}{i\sqrt{i}} +\frac{D_5}{i^2} \right)\\
\le & D_1 + 2D_2 + \frac{1+\ln n}{n} \cdot D_3 + \frac{2}{n} (D_4 + D_5)
\end{align*}
since $\sum_{i=1}^n \frac{1}{\sqrt{i}} \le 2\sqrt{n}$ and $\sum_{i=1}^n \frac{1}{i} \le 1 + \ln n$ and $\sum_{i=1}^n \frac{1}{\sqrt{i}\cdot i} \le 2$. 
We can further replace $D_1 + \frac{1+\ln n}{n} \cdot D_3$ with a sum-of-square term $20 \left(\frac{\sum_{i=1}^n U^*_i}{n}\right)^2 + 10 \left(\frac{1+\ln n}{n} \right)^2$, which will make it easier to compare the quantity with the benchmark (details in the last paragraph of the proof).
\begin{claim}
$$
D_1 + \frac{1+\ln n}{n} \cdot D_3 \le 20 \left(\frac{\sum_{i=1}^n U^*_i}{n}\right)^2 + 10 \left(\frac{1+\ln n}{n} \right)^2.
$$
\end{claim}
\begin{proof}
By defintioin, 
$$
D_1 + \frac{1+\ln n}{n} \cdot D_3 = 10 \left(\frac{\sum_{i=1}^n U^*_i}{n}\right)^2 + \frac{1+\ln n}{n} \cdot \frac{10\sum_{i=1}^n U^*_i }{n}
$$
We prove the inequality by cases.
First, if $\sum_{i=1}^n U^*_i \le 1+\ln n$  then the second term on the right-hand side $ \frac{1+\ln n}{n} \cdot \frac{10\sum_{i=1}^n U^*_i }{n} \le  10 \frac{(1+\ln n)^2}{n^2}$.
If $\sum_{i=1}^n U^*_i \ge 1+\ln n$ then $\frac{1+\ln n}{n} \cdot \frac{10\sum_{i=1}^n U^*_i }{n} \le 10\left(\frac{\sum_{i=1}^n U^*_i }{n}\right)^2$. In any cases, we get
$$
D_1 + \frac{1+\ln n}{n} \cdot D_3 \le 20 \left(\frac{\sum_{i=1}^n U^*_i}{n}\right)^2 + 10 \left(\frac{1+\ln n}{n} \right)^2.
$$
\end{proof}
Finally by plugging in all the numbers,
\begin{align*}
 & 4 \cdot  \frac{\alpha_\gamma^2}{n} \cdot  \E[\widehat{\sigma}]^2 + \left(\E[\widehat{U}]/n\right)^2 \\
 \le & 20 \left(\frac{\sum_{i=1}^n U^*_i}{n}\right)^2 + 10 \left(\frac{1+\ln n}{n} \right)^2 + 2D_2 + \frac{2}{n} (D_4+D_5)\\
\le & 20 \left(\frac{\sum_{i=1}^n U^*_i}{n}\right)^2 + 10 \left(\frac{1+\ln n}{n} \right)^2 + 320\cdot V^* +  \frac{2}{n} \left( 640  \cdot  \frac{\alpha_\gamma^2}{\sqrt{n}} \cdot \frac{1-U^*_{\overline{C}}}{A^*(\overline{C})} + 10\left( U^*_{\overline{C}}\right)^2\right)\\
= & 320\cdot V^* + 20 \left(\frac{\sum_{i=1}^n U^*_i}{n}\right)^2 + 10 \left(\frac{1+\ln n}{n} \right)^2 +   \frac{2}{n} \left( 640  \cdot  \frac{\alpha_\gamma^2}{\sqrt{n}} \cdot \frac{1-U^*_{\overline{C}}}{A^*(\overline{C})} + 10\left( U^*_{\overline{C}}\right)^2\right).
\end{align*}
This completes the proof of the Lemma~\ref{lem:biased_bound_square}.
\end{proof}

Finally we compare the worst-case expected length of our confidence interval with the benchmark $L^*$. Our output confidence interval $\left[\sum_{i=1}^n y_i / n - \frac{\alpha_\gamma}{\sqrt{n}} \cdot \widehat{\sigma}, \quad \sum_{i=1}^n y_i/ n + \frac{\widehat{U}}{n} + \frac{\alpha_\gamma}{\sqrt{n}} \cdot \widehat{\sigma}\right]$ has expected length $L = \E[2\cdot \frac{\alpha_\gamma}{\sqrt{n}} \cdot \widehat{\sigma} + \frac{\widehat{U}}{n}]$. Then for any underlying cost-data distribution $\mathcal{D}$,  $L^2$ can be bounded as follows using Lemma~\ref{lem:biased_bound_square} and the inequalities $(a+b)^2 \le 2(a^2+b^2), \ a^2+b^2 \le (a+b)^2$ for all $a,b \ge 0$.
 \begin{align*}
L^2 & = \left(2\cdot \frac{\alpha_\gamma}{\sqrt{n}} \cdot \E[\widehat{\sigma}] + \frac{\E[\widehat{U}]}{n}\right)^2\\
& \le 2 \left( 4\cdot \frac{\alpha_\gamma^2}{n} \cdot \E[\widehat{\sigma}]^2 + \left(\frac{\E[\widehat{U}]}{n}\right)^2 \right)\\
& \le 2 \left(320\cdot V^* + 20 \left(\frac{\sum_{i=1}^n U^*_i}{n}\right)^2 + 10 \left(\frac{1+\ln n}{n} \right)^2 +   \frac{2}{n} \left( 640  \cdot  \frac{\alpha_\gamma^2}{\sqrt{n}} \cdot \frac{1-U^*_{\overline{C}}}{A^*(\overline{C})} + 10\left( U^*_{\overline{C}}\right)^2\right)\right)\\
& =  640\cdot V^* + 40 \left(\frac{\sum_{i=1}^n U^*_i}{n}\right)^2 + 20 \left(\frac{1+\ln n}{n} \right)^2 +    \frac{2560}{n}  \cdot  \frac{\alpha_\gamma^2}{\sqrt{n}} \cdot \frac{1-U^*_{\overline{C}}}{A^*(\overline{C})} + \frac{40}{n}\left( U^*_{\overline{C}}\right)^2\\
& \le 640 \left( \sqrt{V^*} + \frac{\sum_{i=1}^n U^*_i}{n}\right)^2 + 20 \left(\frac{1+\ln n}{n} \right)^2 +    \frac{2560}{n}  \cdot  \frac{\alpha_\gamma^2}{\sqrt{n}} \cdot \frac{1-U^*_{\overline{C}}}{A^*(\overline{C})} + \frac{40}{n}\left( U^*_{\overline{C}}\right)^2\\
&= 640\cdot\left(\frac{n+1}{n} \cdot L^*\right)^2 +  20 \left(\frac{1+\ln n}{n} \right)^2 +    \frac{2560}{n}  \cdot  \frac{\alpha_\gamma^2}{\sqrt{n}} \cdot \frac{1-U^*_{\overline{C}}}{A^*(\overline{C})} + \frac{40}{n}\left( U^*_{\overline{C}}\right)^2.
\end{align*}
 Since $\sqrt{\sum_j a_j} \le \sum_j \sqrt{a_j}$,
$$
L \le 8\sqrt{10} \cdot \frac{n+1}{n} \cdot L^* + \frac{2\sqrt{5}\cdot (1 + \ln n)}{n} +  \frac{16\sqrt{10}}{\sqrt{n}}  \cdot\frac{\alpha_\gamma}{n^{1/4}} \cdot \sqrt{\frac{1-U^*_{\overline{C}}}{A^*(\overline{C})} } + \frac{2\sqrt{10}}{\sqrt{n}}.
$$
 This completes the proof of Theorem~\ref{thm:CI_main}.

\end{document}